\documentclass[10pt,twocolumn,twoside,final]{IEEEtran}

\IEEEoverridecommandlockouts  
\newcommand{\beginsupplement}{
        
        \setcounter{table}{0}
        \renewcommand{\thetable}{S\arabic{table}}%
        \setcounter{figure}{0}
        \renewcommand{\thefigure}{S\arabic{figure}}%
     }
\usepackage{float}
\usepackage[utf8]{inputenc}
\usepackage{amsmath,amssymb}
\usepackage{amsfonts}
\usepackage{amsthm}
\usepackage{url}
\usepackage{changes}
\newtheorem{thm}{Theorem}
\newtheorem{defi}{Definition}
\newtheorem{lem}{Lemma}

\newtheorem{rem}{Remark}
\usepackage{tikz}
\usetikzlibrary{automata,positioning,shapes,arrows}
\usetikzlibrary{pgfplots.groupplots}
\usetikzlibrary{plotmarks}
\usetikzlibrary{calc}

\newtheorem{example}{Example}
\newtheorem*{thm*}{Theorem}
\newcommand{\norm}[1]{\left\lVert#1\right\rVert}
\usepackage{algpseudocode}
\usepackage{mathtools}
\usepackage{acronym}
\usepackage[lined,ruled,commentsnumbered,noend]{algorithm2e}
\usepackage{pgfplots}
\usepackage{subfig}
\definecolor{maroon}{RGB}{139,0,0}
\definecolor{blueish}{RGB}{125,100,255}
\definecolor{greenish}{RGB}{26,102,46}
\definecolor{pink}{RGB}{255,0,144}
\definecolor{cyan}{RGB}{0,255,255}

\makeatletter
\newcommand{\removelatexerror}{\let\@latex@error\@gobble}
\makeatother
\title{Byzantine-Resilient Distributed Hypothesis Testing \\ With Time-Varying Network Topology}
\author{Bo Wu, Steven Carr, Suda Bharadwaj, Zhe Xu, and Ufuk Topcu
\thanks{Bo Wu, Steven Carr, Suda Bharadwaj, and Ufuk Topcu are with the Department of Aerospace Engineering
and Engineering Mechanics, and the Oden Institute for Computational
Engineering and Sciences, University of Texas, Austin, 201 E 24th
St, Austin, TX 78712. Zhe Xu is with the School for Engineering of Matter, Transport, and Energy, Arizona State University, Tempe, AZ 85287. email: {\tt\small $\{$bwu3, stevencarr, suda.b, utopcu$\}$@utexas.edu, xzhe1@asu.edu}.  This work was partly funded by grants AFRL FA9550-19-1-0169 and DARPA D19AP00004.}}

\begin{document}

\maketitle
\begin{abstract}
\added{We study the problem of distributed hypothesis testing over a network of mobile agents with limited communication and sensing ranges to infer the true hypothesis collaboratively. 
In particular, we consider a scenario where there is an unknown subset of compromised agents that may deliberately share altered information to undermine the team objective. We propose two distributed algorithms where each agent maintains and updates two sets of beliefs  (i.e., probability distributions over the hypotheses), namely \emph{local} and \emph{actual} beliefs (LB and AB respectively for brevity). 
In both algorithms, at every time step, each agent shares its AB with other agents within its communication range and makes a local observation to update its LB. 
Then both algorithms can use the shared information to update ABs under certain conditions.
One requires receiving a certain number of shared  ABs at \emph{each time instant}; the other accumulates shared ABs \emph{over time} and updates after the number of shared ABs exceeds a prescribed threshold.
Otherwise, both algorithms rely on the agent's current LB and AB to update the new AB. We prove under mild assumptions that the AB for every non-compromised agent converges almost surely to the true hypothesis, without requiring connectivity in the underlying time-varying network topology. Using a simulation of a team of unmanned aerial vehicles aiming to classify adversarial agents among themselves, we illustrate and compare the proposed algorithms. Finally, we show experimentally that the second algorithm consistently outperforms the first algorithm in terms of the speed of convergence.}
\end{abstract}


\section{Introduction}
This paper studies a problem in distributed teams of cooperating agents performing tasks that are beyond the capability of an individual agent. Similar problems have attracted recent interest, see, e.g., \cite{olfati2006belief,cubuktepe2020policy,Allerton2019,tarighati2017decentralized,nedic2016distributed,liu2017distributed,liu2017communication,liu2018distributed,wu2015combined,Djeumou2020,CensusSTL2016}. As a running example, consider a team of mobile agents performing  persistent surveillance tasks as shown in Fig. \ref{fig:motivating example}. Each agent monitors a certain region by following a given trajectory for an indefinite period of time. Such a team of agents offers real-time surveillance and rapid response that covers a massive environment.

In adversarial environments, the agents may be subject to external influence (e.g., through a cyber attack) resulting in an a priori unknown subset of compromised (bad) agents that may behave adversely and follow different trajectories. To classify those bad agents, each non-compromised (good) agent may need to repeatedly sense the other agents' positions. Because of limited ranges, noisy sensor data, and individual surveillance task constraints, it may not be reasonable to anticipate that a single good agent can classify all bad agents. Instead, the agents must share their local information with their neighbors, i.e., the mobile agents within their communication range, to identify those bad agents collaboratively. Note that a bad agent may share arbitrarily altered information  to prevent itself from being identified.  Collaboration under the existence of bad agents raises the question of how to process the local and shared information so that the good agents can reach a consensus on the subset of bad agents correctly.  This classification problem fits into the framework of distributed hypothesis testing, where every possible subset of bad agents is a hypothesis. 

\begin{figure}[t]
\centering
\includegraphics[width=0.475\textwidth]{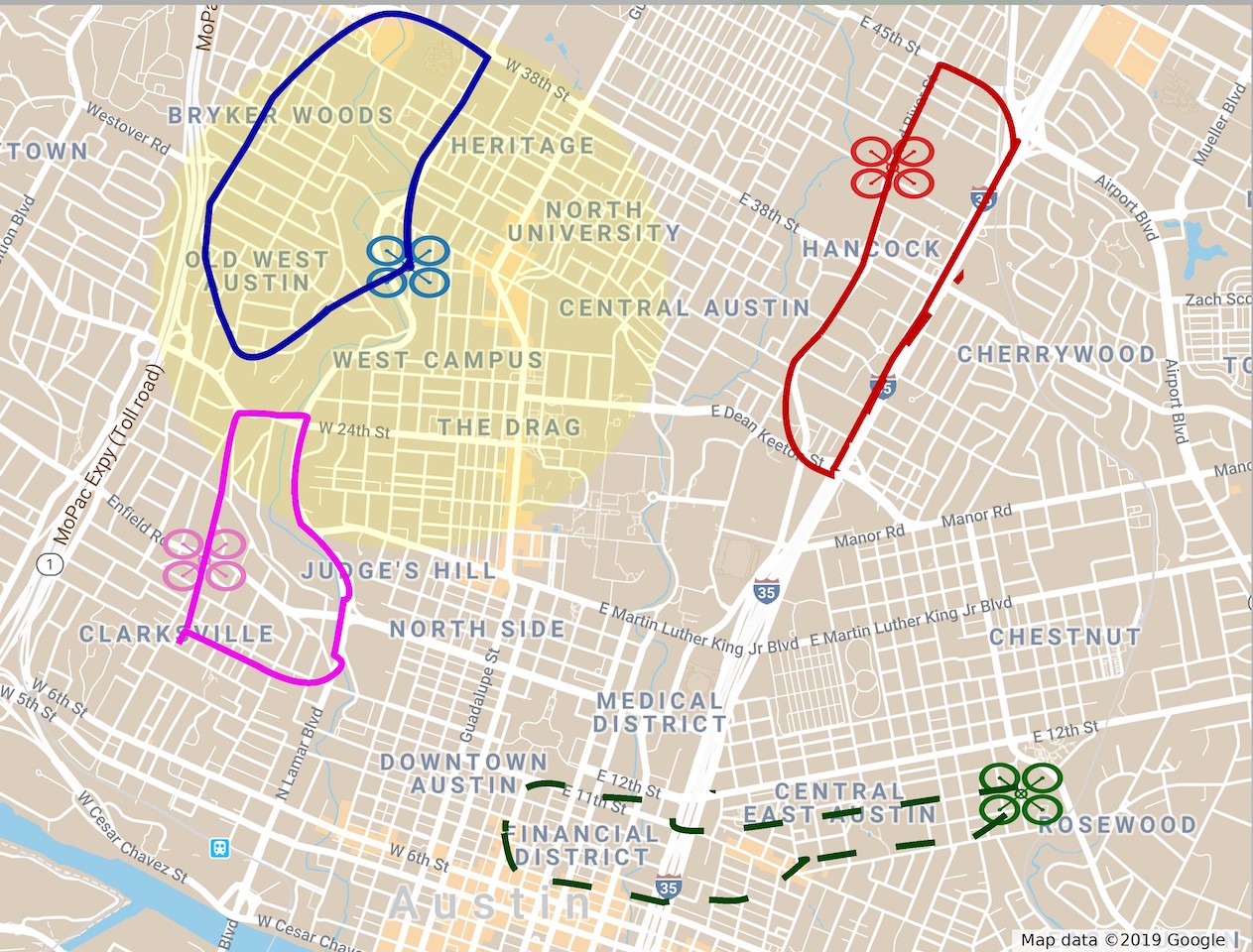}
\caption{{The motivating example consists of four agents (unmanned aerial vehicles or UAVs). The shaded yellow region represents the sensing and communication ranges.  Solid and dashed lines represent the trajectories of an agent, depending on whether it is good or bad.} The green agent is bad, and {it is following a dashed trajectory.}} \label{fig:motivating example}
\end{figure}

\added{Take Fig. \ref{fig:motivating example} as an example, the compromised green UAV  follows a different trajectory from its assigned one. Due to the limited sensing range, no good agent may observe this bad agent at every time step. Consequently each agent cannot infer which agent is bad individually based on its local observations of agent positions. Also due to limited communication ranges, agents may only share their local information occasionally to mobile agents within communication ranges, where the bad agent may also deliberately share contrived information that can trick other agents. Therefore, in this scenario, we need a resilient distributed solution so that each agent can make local observations, share its local information, and collaboratively identify the bad agent over time, regardless of the influence of the bad agent.}

{In distributed hypothesis testing,  a team of agents} makes local observations  and  collaboratively infer the unknown true hypothesis that generates their observations.  Distributed hypothesis testing finds a wide spectrum of applications, for example, in social learning \cite{rhim2014distributed,jadbabaie2012non,lalitha2018social}, sensor networks \cite{tarighati2017decentralized,alanyali2004distributed,olfati2006belief}, and wireless communication \cite{salehkalaibar2018hypothesis,rahman2012optimality}.
{The major challenge of distributed hypothesis testing is to design interaction rules to process local and shared information so that the agents will converge to the unknown true hypothesis.}

In one approach, the agents do not directly communicate with each other but send their local information to a fusion center for centralized processing \cite{veeravalli1993decentralized,rhim2014distributed,tarighati2017decentralized}. However, {such centralized processing} may place communication and computation burdens  on the fusion center as the number of  agents increases. Furthermore, the team objective will fail with a compromised fusion center.
To improve the scalability and resilience, distributed solutions where each agent communicates along a graph to its neighbors without a fusion center are growing in popularity, e.g., \cite{alanyali2004distributed,olfati2006belief,jadbabaie2012non,nedic2017fast,lalitha2018social,mitra2019new}. 

This paper considers a  distributed hypothesis testing problem over a network of mobile agents with a time-varying network topology.   
\added{Specifically, each agent maintains and updates two sets of \emph{beliefs}, namely \emph{local} and \emph{actual} beliefs \cite{mitra2019new} (LB and AB for brevity), based on its local observations and neighbors' ABs. A belief is a probability distribution over the hypotheses. We are interested in designing algorithms that perform belief updates guaranteeing that each agent's AB converges to the true hypothesis with resilience to bad agents that share arbitrarily altered information.}


In a preliminary version \cite{Bo2019distributed} of this paper, we proposed a resilient belief update algorithm. \added{At every time step, each agent shares its AB to its neighbors, makes a local observation, and updates its LB. To perform the AB update with the shared beliefs, the algorithm requires a sufficient number of shared beliefs to filter out the impact of the bad agents \emph{at the same time instant}}. Thus, {we refer to this updating method as a \emph{synchronous} belief update}. \added{When there are insufficient shared ABs to perform an AB update, the algorithm updates the agents' ABs as a function of their local and ABs. In \cite{Bo2019distributed}, the algorithm, after filtering out the impact of the bad agents, then takes the minimum of the neighbors' ABs on each hypothesis.} 
{In \cite{Bo2019distributed}, we proved the almost-sure convergence to the true hypothesis without requiring connectivity in the underlying network topology.}

This paper makes significant extensions on \cite{Bo2019distributed} and introduces additional belief update algorithms. \added{In the new algorithm, each agent collects the shared ABs \emph{over time} until there are enough of them to make the AB update.} Since there is no explicit time dependence on information, we call this process an \emph{asynchronous} belief update algorithm. We prove the almost-sure convergence to the true hypothesis under mild assumptions. \added{We also show that, besides taking the minimum, taking the average of the shared ABs over each hypothesis guarantees the convergence.} With low sensor noise, the minimum rule converges faster than the average rule since it can quickly rule out the unlikely hypotheses. Conversely, when the sensor noise is high, the average rule converges faster with lower variance. 

We conduct simulations with a team of UAVs that collaboratively tries to classify the compromised agents in the team. These results empirically demonstrate the validity and compare the performance of the synchronous and asynchronous algorithms. We show that the asynchronous algorithm consistently outperforms the synchronous algorithm. We also compare the performance between the average and minimum rules under different sensor noises.
\added{Finally, we show that the algorithm convergences even when multiple bad agents coordinate to deceive the others.}

\noindent {\bf{Related work.}} Most existing belief update algorithms make use of consensus-based belief aggregation assuming a strongly connected (potentially time-varying) network topology, see  e.g., \cite{lalitha2018social,olfati2006belief,shahrampour2015distributed,nedic2016distributed,nedic2017fast}. However, none of these methods consider adversarial agents  that do not follow the update rule and may share arbitrarily altered beliefs. As a result, these rules will fail in the presence of compromised agents. \added{Recent results in \cite{7322210,9069226} consider the vulnerability of distributed algorithms. However, their settings are in cyber-physical systems that involve continuous dynamics, and the focus is on the stability of the system.} 

\added{Belief propagation (BP) \cite{pearl1982reverend, braunstein2005survey, weiss2001correctness,sui2018accuracy} considers computing the marginal distribution for each agent based on local and shared information. However, BP  generally also assumes certain connectivity constraints for convergence, does not consider time varying graphs and Byzantine agents, and mostly focuses on sum-product belief update rule.}

The works most related to this paper are \cite{mitra2019new} and \cite{su2019defending}, where the belief update algorithms are resilient against bad agents. These bad agents follow a \emph{Byzantine} adversary model where they may have access to complete knowledge of the team task, belief update algorithm, shared information, and true hypothesis. These adversaries may send arbitrarily altered beliefs to undermine the team objective. {The belief update algorithm proposed in \cite{mitra2019new} is resilient to adversarial agents and {almost surely converges} to the true hypothesis.} However, the guarantee in \cite{mitra2019new} assumes a fixed network topology ({the extended version \cite{mitra2019new1} considers a time-varying network topology, but it only applies in settings without adversarial agents}). Furthermore, the guarantee of convergence in \cite{mitra2019new}  relies on some  graph-theoretic connectivity requirements of the network topology.

Compared to the existing literature for distributed hypothesis testing, this paper has three principal contributions. First, we design belief update algorithms resilient against compromised agents considering  a \emph{time-varying} network topology. Second, we prove that every non-compromised agent will converge almost surely to the true hypothesis \emph{without} requiring connectivity in  the underlying network topology. The proposed approaches are not only applicable  to the classification  problem  in Fig.~\ref{fig:motivating example} but  also to other applications such as collaborative localization and distributed intrusion detection. In these cases, the proposed framework also naturally extends to   settings without adversarial agents. \added{Third, we show the validity of the proposed algorithms experimentally, where the asynchronous algorithm consistently converges faster than the synchronous algorithm. We also compare the performance between average and minimum rules that make use of shared ABs.   }


\section{Preliminaries and Modeling Framework}\label{section:Preliminaries}

We consider a set $\mathcal{N}=\{0,...,N-1\}$ of agents that move in a gridworld with a finite grid set $Q$. Let $\mathbb{Z}^{\geq 0}$ denote non-negative integers.  At time step $t\in\mathbb{Z}^{\geq 0}$, we denote  $q_{i,t}\in Q$ as the \emph{state} of an agent $i$ that represents its position at time $t$. \added{Each agent is moving under the  constraints of a directed graph $\mathcal{G}_m=(Q,E_m)$ where $E_m\subseteq Q\times Q$  and $m$ in the subscript indicates that this graph characterizes the \emph{motion} of an agent}. An agent can move from  $q$ to $q'$ in one time step if and only if $(q,q')\in E_m$.  

For agent $i$, we characterize its communication range by a function $H_i:Q\rightarrow 2^Q$. Agent $i$ at state $q$ can communicate to another agent $j$ at state $q'$ if and only if $q'\in H_i(q)$ (note that we set $q\in H_i(q)$). \added{Then we characterize the network topology at time $t$ for the team of agents by a directed graph $\mathcal{G}_{c,t}=(\mathcal{N},E_{c,t})$, where the subscript $c$ indicates that this graph is a result of an agent's \emph{communication} between the agents that are within its communication range}.  An edge $(i,j)\in E_{c,t}\subseteq \mathcal{N}\times \mathcal{N}$ if and only if $q_{j,t}\in H_i(q_{i,t})$.  \added{In such a case, we say that agent $i$ is a \emph{neighbor} of agent $j$ at time $t$ meaning that agent $j$ is within agent $i$'s communication range, and thus, agent $i$ can communicate to agent $j$ (but not necessarily vice versa since we consider a general case where each agent may have difference communication range).} We denote $\mathcal{N}_{i,t}:=\{j\in\mathcal{N}|q_{i,t}\in H_j(q_{j,t})\}\subseteq\mathcal{N}$ as the set of all neighbors of  agent $i$ at time $t$.
\subsection{Hypothesis, Observations, and Local Likelihood Functions}
There is a finite set $\Theta$ of possible hypotheses. We denote the total number of hypotheses as $m=|\Theta|$. At each time step $t$, an agent $i$ at a state $q_t\in Q$ makes an observation $s\in S_i$ where $S_i$ denotes a set of observations for agent $i$. 

The probability of observing $s$ is given by a conditional likelihood function $l_i(s|\theta^*,q_{i,t})$, where $l_i(s|\theta^*,q_{i,t})\in [0,1]$, and $\sum_{s \in S_i}l_i(s|\theta^*,q_{i,t})=1$. We denote $\theta^*\in\Theta$ as the unknown but fixed \emph{true} hypothesis to be learned. {The conditional likelihood functions characterize the sensor noise conditioned on the agent's position and the true hypothesis.} \added{Each agent $i$ only has the knowledge of its likelihood functions $\{l_i(\cdot|\theta,q_{i,t}), \forall \theta\in\Theta, q_{i,t}\in Q\}$, which may not be identical across the agents.}
  
\subsection{Agent Trajectories and Identities}
Each agent $i$, starting at $t=0$, moves in the gridworld following a sequence of states $(q_{i,0},q_{i,1},q_{i,2},...)$ which we denote as a \emph{local state path}. Obviously, at any time $t$, $(q_{i,t},q_{i,t+1})\in E_m$. We assume each agent follows a given local state path. Furthermore, the local likelihood function for $\theta^*$ only depends on an agent's current state $q_t$. Therefore, the observation sequence for each agent is an i.i.d random process. We define the set of state observation paths as follows.

\begin{defi}[State observation paths]
\added{Given an agent $i$ and a local state path $(q_{i,0},q_{i,1},q_{i,2}...)$, its set $\Omega_i$ of local state observation paths is defined as $\Omega_i:=\{\omega_i|\omega_i=(q_{i,0},s_{i,0})(q_{i,1},s_{i,1})(q_{i,2},s_{i,2})...,\forall s_{i,t}\in S_i,q_{i,t}\in Q,\forall t\in\mathbb{N}\}$ with $P_{i,\theta^*}(\omega_i)=\prod_{t=0}^\infty l_i(s_{i,t}|\theta^*,q_{i,t})$.}  The set $\Omega$ of global state observation paths is defined as $\Omega:=\prod_{i}\Omega_i$.
\end{defi}

Within the team of agents, there is a subset of non-compromised (good) agents defined as $G\subseteq\mathcal{N}$. Good agents follow their given state paths and the distributed hypothesis testing rule. We assume that, for an agent $i\in G$, at any time $t$, there are at most $f$ bad neighboring agents, even though the identities of these bad agents are not known. The bad agents are characterized by the Byzantine fault model \cite{dolev1986reaching}. Each of them has full access to all agents' state paths, their local likelihood functions, any information shared over the network topology, and the distributed hypothesis testing rule used by the team. If an agent is bad, it may follow a different state path. To prevent the team of agents from achieving the hypothesis testing objective, bad agents may collaboratively share arbitrarily altered information to their neighbors. 

\subsection{Source Location and Source Agent}
The objective of this paper is to design a distributed hypothesis testing rule such that, when time goes to infinity, every good agent $i\in G$ is able to determine the true hypothesis $\theta^*\in\Theta$ almost surely. To this end, we define the following:
\begin{defi}[Kullback–Leibler (KL) divergence \cite{kullback1951information}]
KL divergence  $D(P_1||P_2)$ of two discrete probabilistic distributions $P_1$ and $P_2$ is given by 
\begin{equation}
D(P_1||P_2) := \sum_x P_1(x)\log(\frac{P_1(x)}{P_2(x)}).  
\end{equation}
\end{defi}

\begin{defi}[Source state]
A state $q\in Q$ is called a source state for a pair of hypothesis $\theta$ and $\theta'\in\Theta$ and an agent $i$  if and only if  $D(l_i(\cdot|\theta,q)||l_i(\cdot|\theta',q))>0$\footnote{\added{Here $>0$ indicates an information gain over $\theta$.}}.
\end{defi}
{We further define a source state set $O_i(\theta,\theta')\subseteq Q$ for agent $i$ as   $O_i(\theta,\theta'):=\{q\in Q|D(l_i(\cdot|\theta,q)||l_i(\cdot|\theta',q))>0\}.$
Intuitively, $O_i(\theta,\theta')$ denotes all the source states where $\theta$ and $\theta'$ incur different likelihood functions for agent $i$. However, as we will see in Section \ref{section:learning rules}, it requires an infinite number of visits to at least one source state in $O_i(\theta,\theta')$ for agent $i$  to distinguish $\theta$ and $\theta'$. Therefore, we define: }

\begin{defi}[Source agent]\label{def:S(theta,theta')}
{An agent $i$ with a local state path $(q_{i,0},q_{i,1},q_{i,2}...)$ is a source agent for a pair of hypothesis $\theta$ and $\theta'\in\Theta$ if and only if}

\begin{equation}\label{equation:lemma convergence}
 \added{   \lim_{T\rightarrow\infty}\sum_{t=0}^{T}I_{ O_i(\theta,\theta')}(q_{i,t})=\infty} \footnote{\added{Here the term $\infty$ indicates that we are interested in experiencing $q_{i,t} \in O_i(\theta,\theta^{\prime}),$ for all the time instants as much as possible.}},
\end{equation}
where {$I_{ O_i(\theta,\theta')}(q_{i,t}):Q\rightarrow\{0,1\}$} is the indicator function. {$I_{ O_i(\theta,\theta')}(q_{i,t})=1$ if $q_{i,t}\in O_i(\theta,\theta')$, and $I_{ O_i(\theta,\theta')}(q_{i,t})=0$ otherwise.}
\end{defi}

Similarly, we define a \textit{source agent set} $S(\theta,\theta')\subseteq\mathcal{N}$ where
\added{$$S(\theta,\theta'):=\{i\in \mathcal{N}|    \lim_{T\rightarrow\infty}\sum_{t=0}^{T}I_{ O_i(\theta,\theta')}(q_{i,t})=\infty\}.$$}

{By Definition \ref{def:S(theta,theta')}, agent $i$ belongs to the set $S(\theta,\theta')$ if it visits at least one source state $q\in O_i(\theta,\theta')$ infinitely often.}

\section{Synchronous Distributed Hypothesis Algorithm}\label{section:learning rules}
\begin{figure}[!t]
		\removelatexerror
		\begin{algorithm}[H]
			\SetKwData{Left}{left}\SetKwData{This}{this}\SetKwData{Up}{up}
			\SetKwFunction{Union}{Union}\SetKwFunction{FindCompress}{FindCompress}
			\SetKwInOut{Input}{input}\SetKwInOut{Output}{output}
			\Input{Agent $i$, its location $q_{i,t+1}$, neighbor set $\mathcal{N}_{i,t+1}$, and observation $s_{i,t+1}$.}
			\BlankLine
			\nl     \For{$\theta\in\Theta$}{
			\nl \added{Compute the new LB} \label{algorithm:LB update }\begin{equation}\label{equation:LB update rule}
			        b^l_{i,t+1}(\theta) = \frac{l_i(s_{i,t+1}|\theta,q_{i,t+1})b^l_{i,t}(\theta)}{\sum_{p=1}^m l_i(s_{i,t+1}|\theta_p,q_{i,t+1})b^l_{i,t}(\theta_p)}.
    \end{equation}  \Comment{\added{LB update with Bayesian rule}}\;
			\nl     \eIf{  for all $\theta'\neq\theta$, $|S(\theta,\theta')\cap\mathcal{N}_{i,t+1}|\geq 2f+1$\label{algorithm:AB update }}{
                                        \nl Remove $f$ neighboring agents with the lowest $f$ beliefs and save the   rest of agents to $\mathcal{N}^\theta_{i,t+1}$\label{algorithm:remove beliefs}\;
                                        \nl \added{Compute the new AB as}
                                            \begin{equation}\label{equation:acutual belief update rule for case one}
                                    \tilde{b}^a_{i,t+1}(\theta) = \min\{\{b^a_{j,t}(\theta)\}_{j\in \mathcal{N}^\theta_{i,t+1}},b^l_{i,t+1}(\theta)\}.
                            \end{equation}	\label{algorithm:AB update case one}	\Comment{\added{Case one for AB update.}}}	
			                {\nl \added{Compute the new AB as} \begin{equation}\label{equation:acutual belief update rule for case two}
        \tilde{b}^a_{i,t+1}(\theta) = \min\{b^a_{i,t}(\theta),b^l_{i,t+1}(\theta)\}.
    \end{equation}\label{algorithm:AB update case two}\Comment{\added{Case two for AB update.}}}

			        }
	    \nl Normalization step. For each $\theta\in\Theta$, perform
	    \begin{equation}\label{equation:norm}
	        b^a_{i,t+1}(\theta) = \frac{\tilde{b}^a_{i,t+1}(\theta)}{\sum_{p=1}^m \tilde{b}^a_{i,t+1}(\theta_p)}.
	    \end{equation} \label{algorithm:normalization}

			\caption{Synchronized  Distributed Hypothesis Testing (SDHT)}\label{alg:Synchronized Distributed Hypothesis Testing}			
		\end{algorithm}
	\end{figure}
	
In this section, we propose  an algorithm that describes the belief update rule for each agent. Before making an observation at time $t+1$, agent $i$ maintains a local belief and an actual belief \cite{mitra2019new}:
\begin{itemize}
    \item The local belief (LB) 
    $$b^l_{i,t}:\Theta\rightarrow [0,1], 
    \sum_{\theta\in\Theta}b^l_{i,t}(\theta)=1.$$
    \item The actual belief (AB) 
    $$b^a_{i,t}:\Theta\rightarrow [0,1], 
    \sum_{\theta\in\Theta}b^a_{i,t}(\theta)=1.$$
\end{itemize} 
At $t=0$, the beliefs $b^l_{i,0}$ and $b^a_{i,0}$ are initialized according to  some \emph{a priori}  distribution. 

 We summarize the belief update procedure for one time step in Algorithm \ref{alg:Synchronized Distributed Hypothesis Testing} (SDHT). At time $t+1$, agent $i$ is at $q_{i,t+1}$ and makes an observation $s_{i,t+1}$. The algorithm proceeds as follows. 

\added{For each $\theta\in\Theta$, as shown in Line \ref{algorithm:LB update } of SDHT, the algorithm  first updates the LB $b^l_{i,t+1}(\theta)$ with \eqref{equation:LB update rule} following Bayesian rule.}
\added{Then the algorithm moves on to update the AB  as shown from Line \ref{algorithm:AB update } to Line \ref{algorithm:AB update case two} of SDHT. We update AB $b^a_{i,t+1}(\theta)$ according to one of the two cases. As shown in Line \ref{algorithm:AB update }, if for all $\theta'\neq\theta$, $|S(\theta,\theta')\cap\mathcal{N}_{i,t+1}|\geq 2f+1$, {i.e., the number of source agents for $\theta$ and $\theta'$ that are agent $i$'s neighbors at time $t+1$ exceeds $2f+1$, then agent $i$ updates its AB in case one. SDHT then sorts $b^a_{j,t+1}(\theta)$ for all $j\in \mathcal{N}_{i,t+1}$ and removes  $f$ neighbors with the lowest $f$ ABs on $\theta$.} We denote the set $\mathcal{N}^\theta_{i,t+1}$ as the remaining neighbors. Then the algorithm updates the  AB as in \eqref{equation:acutual belief update rule for case one}. }

\added{On the other hand, if the condition for case one is not satisfied, the AB is updated in case two, as shown in Line \ref{algorithm:AB update case two} of SDHT. In \eqref{equation:acutual belief update rule for case two}, we update the AB with the smaller value between the newly updated LB and the AB at time $t$.  Then the algorithm normalizes the ABs to make sure they sum up to one.}

\added{We start with the following lemma to show how LBs for any good agent $i$ evolve.}
\begin{lem}\label{lemma:convergence}
Consider a good agent $i\in G$, a local state path $\omega_i=(q_{i,0},q_{i,1},q_{i,2}...)$ and a pair of hypotheses $\theta^*$ and $\theta$, where $\theta^*$ denotes the true hypothesis and $\theta^*\neq\theta$. If $b^l_{i,0}(\theta^*)>0$ and $i\in S(\theta,\theta^*)$, then
\begin{equation}\label{equation:zero of nontrue hypothesis}
  b^l_{i,t}(\theta)\rightarrow 0  \text{ almost surely},
\end{equation}
and
\begin{equation}\label{equation:nonzero of true hypothesis}
  b^l_{i,t}(\theta^*)>0  \text{ for all $t$ almost surely}.
\end{equation}
\end{lem}
\begin{proof}
\added{Please find the proof to this lemma in the appendix.}
\end{proof}

\begin{rem}\label{remark:for lemma 1}
From Lemma \ref{lemma:convergence}, we can see the intuitive meaning of a source agent set $ S(\theta,\theta')$ for any hypothesis pair $\theta$ and $\theta'$ where $\theta\neq\theta'$. If $\theta=\theta^*$, we know that $b_{i,t}^l(\theta')\rightarrow 0$ almost surely for any $i\in S(\theta,\theta')$, which implies that any source agent for the hypothesis pair $\theta^*$ and $\theta'$ is able to distinguish between $\theta^*$ and $\theta'$ and rules out $\theta'$. \added{The AB $b_{i,t}^a(\theta')$ for $\theta'$ will also approach zero since it is upper-bounded by $b_{i,t}^l(\theta')$ as can be observed from \eqref{equation:acutual belief update rule for case one} and \eqref{equation:acutual belief update rule for case two}.}
\end{rem}

\begin{rem}\label{remark:lemma 1}
If we define a set $\hat{\Omega}\subseteq\Omega$ of global state observation path  such that $\omega=\prod\omega_j\in\hat{\Omega}$ if and only if for any good agent $i$,
\begin{itemize}
    \item for each $\theta\neq\theta^*$, if $i\in S(\theta,\theta^*)$,  $b^l_{i,t}(\theta)\rightarrow 0$, and 
    \item $\lim_{t\rightarrow\infty}b^l_{i,t}(\theta^*)$ exists with a given $\omega_i$.
\end{itemize}
By Lemma \ref{lemma:convergence} we know that $\hat{\Omega}$ has measure one. 
\end{rem}

\added{Lemma \ref{lemma:convergence} also states that, for a good agent $i\in G$, its LB $b^l_{i,t}(\theta^*)>0$ for all $t$ almost surely.} \added{But is it possible for the bad agents to influence their neighboring good agents such that the good agents' ABs on $\theta^*$ are set to zero?} The following lemma shows that this situation cannot happen with the proposed belief update rule.
\begin{lem}\label{lemma:nonzero of true hypothesis}
For any good agent $i\in G$, $b^a_{i,t}(\theta^*)>0$ for all $t$ almost surely.
\end{lem}
\begin{proof}
\added{Please find the proof in the appendix. }
\end{proof}

The following theorem guarantees that SDHT almost surely converges to the true hypothesis.

\begin{thm}\label{theorem:main result}
For each agent $i\in G$ and its corresponding local state path $(q_{i,0},q_{i,1},q_{i,2}...)$, suppose the following conditions hold:
\begin{enumerate}
    \item The initial beliefs $b^l_{i,0}(\theta)>0$ and $b^a_{i,0}(\theta)>0$ for any $\theta\in\Theta$ and any  agent $i$.
    \item If case one in SDHT happens only finitely often for a  hypothesis $\theta\in\Theta$, then $i\in S(\theta,\theta')$ for any $\theta'\neq\theta$.
\end{enumerate}
Then SDHT ensures that $b^a_{i,t}(\theta^*)\rightarrow 1$ almost surely for every good agent $i\in G$ as $t\rightarrow\infty$.
\end{thm}
\begin{proof}
\added{Please find the proof in the appendix.}
\end{proof}
\begin{rem}
\added{
Intuitively, the second condition requires that, for a hypothesis $\theta$, if agent $i$ cannot distinguish between $\theta$ and $\theta'$ for every $\theta'\neq\theta$ with the help from its neighbors, then it  must be able to do so by itself. In the extreme case, if one agent has
no neighbors on its path, then it must be able to distinguish any hypothesis pair to converge to true hypothesis
by itself. This condition asks for just \emph{enough} level of interactions among the agent in the sense that an agent only needs to communication to other agents if that agent cannot distinguish between two hypotheses while other agents can. And enforcing such requirement is not prohibitively hard since there is no hard limit on the communication interval, as long as it happens infinitely often. Such a condition can be enforced by heuristics (like the ones that we used in our experiments), or theoretically and systematically guaranteed using formal methods like reactive synthesis \cite{bharadwaj2018synthesis}.}
\end{rem}


\section{{Learning rule with asynchronous updates}}\label{section:Learning rule with asynchronous update}
\added{While we prove that the learning rule proposed in Section \ref{section:learning rules}  converges almost surely, the algorithm  requires that  case one of AB update in ADHT must occur infinitely often for a hypothesis $\theta$ if there exists another hypothesis $\theta'\neq\theta$ such that $i\notin S(\theta,\theta')$, i.e., agent $i$ cannot distinguish $\theta$ and $\theta'$ on its own. However, to enter  case one of AB update, the algorithm requires $|S(\theta,\theta')\cap\mathcal{N}_{i,t}|\geq 2f+1$ for all $\theta'\neq\theta$, which implies that the number of neighbors that are source agents for $\theta$ and $\theta'$ must be at least $2f+1$ for all $\theta'\neq\theta$  at a \emph{single} time instant.} Such a requirement may be conservative in some cases, which may make the convergence slow, since case one may rarely happen. Therefore, in this section, we discuss how to relax such a condition while still guaranteeing convergence.

\begin{figure}[!t]
		\removelatexerror
		\begin{algorithm}[H]
			\SetKwData{Left}{left}\SetKwData{This}{this}\SetKwData{Up}{up}
			\SetKwFunction{Union}{Union}\SetKwFunction{FindCompress}{FindCompress}
			\SetKwInOut{Input}{input}\SetKwInOut{Output}{output}
			\Input{Agent $i$, its location $q_{i,t+1}$, neighbor set $\mathcal{N}_{i,t+1}$, and observation $s_{i,t+1}$.}
			\BlankLine
			\nl     \For{$\theta\in\Theta$}{
			\nl \added{Compute the new LB as in \eqref{equation:LB update rule}}\label{algorithm:asynchronous LB update }\;\Comment{\added{LB update with Bayesian rule}}\;
			\nl     \eIf{ABU($i,\theta,\mathcal{N}_{i,t+1}$)==True\label{algorithm:asynchronous AB update }}{
                                        \nl Remove $f$  agents with the $f$ lowest beliefs in $\{b^a_{j}(\theta)|j\in  \mathcal{N}^{\theta}_{i}\}$ and save the   rest of agents to $\tilde{\mathcal{N}}^\theta_{i}$\label{algorithm:asynchronous  remove beliefs}\;
                                        \nl \added{Compute the new AB}
                                            \begin{equation}\label{equation:asynchronous acutual belief update rule for case one}
                                    \tilde{b}^a_{i,t+1}(\theta) = \min\{\{b^a_{j}(\theta)\}_{j\in \tilde{\mathcal{N}}^\theta_{i}},b^l_{i,t+1}(\theta)\}.
                            \end{equation}	\label{algorithm: asynchronous AB update case one}	\Comment{\added{Case one for AB update}}\;}	
			                {\nl Compute $\tilde{b}^a_{i,t+1}(\theta)$ as in \eqref{equation:acutual belief update rule for case two}\label{algorithm:asynchronous  AB update case two}\;\Comment{\added{Case two for AB update}}\;}

			        }
	    \nl \added{Normalize ABs following \eqref{equation:norm}} \label{algorithm:asynchronous normalization}\;

			\caption{Asynchronous Distributed Hypothesis Testing (ADHT)}\label{alg:Asynchronous Distributed Hypothesis Testing}			
		\end{algorithm}
	\end{figure}

We summarize the proposed algorithm in Algorithm \ref{alg:Asynchronous Distributed Hypothesis Testing} (ADHT). \added{The LB update is identical to that of SDHT. The main difference is case one for AB update and the condition to enter it from Line \ref{algorithm:asynchronous AB update } to Line \ref{algorithm: asynchronous AB update case one}.} In Line \ref{algorithm:asynchronous AB update }, we use Algorithm \ref{alg:ABU} such that, at any time $t$ and for any $\theta\in\Theta$, if it returns true, the update rule will choose case one. 

\begin{figure}[!t]
		\removelatexerror
		\begin{algorithm}[H]
			\SetKwData{Left}{left}\SetKwData{This}{this}\SetKwData{Up}{up}
			\SetKwFunction{Union}{Union}\SetKwFunction{FindCompress}{FindCompress}
			\SetKwInOut{Input}{input}\SetKwInOut{Output}{output}
			\Input{Agent $i$, $\theta\in\Theta$, and $\mathcal{N}_{i,t+1}$}
			\Output{\emph{True} or \emph{False}}
			\BlankLine
			\nl \If{$t==0$ or $ResetFlag==True$\label{algorithm:initialization}}{
			\nl     Reset($i,\theta$)\;
			    }
			\nl	\For {$j\in \mathcal{N}_{i,t+1}$ \label{algorithm:neighbor}}{
			        \nl	 $\mathcal{N}^{\theta}_{i}=\mathcal{N}^{\theta}_{i}\bigcup j$\label{algorithm:increase by one}\;
			        
			        \nl $b_j^a(\theta)=b_{j,t}^a(\theta)$\label{algorithm:acutual belief}\;
			        
    			}
            \nl	\For {$\theta'\in\Theta,\theta'\neq\theta$}{
            			        \nl \If{$|\mathcal{N}^{\theta}_{i}\cap S(\theta,\theta')|< 2f+1$\label{algorithm:number of neighbor}}{
            			                \nl \Return{False}\;
            			        }
            			        }
			\nl $ResetFlag=True$\label{algorithm:reset}\;
		\nl \Return {True}\;
			
			\caption{Asynchronous belief update (ABU)}\label{alg:ABU}			
		\end{algorithm}
	\end{figure}
\begin{figure}[!t]
		\removelatexerror
		\begin{algorithm}[H]
			\SetKwData{Left}{left}\SetKwData{This}{this}\SetKwData{Up}{up}
			\SetKwFunction{Union}{Union}\SetKwFunction{FindCompress}{FindCompress}
			\SetKwInOut{Input}{input}\SetKwInOut{Output}{output}
			\Input{Agent $i$ and $\theta\in\Theta$.}
			\BlankLine

			\nl     \For{$j\in\mathcal{N},j\neq i$}{
			\nl    $b_j^a(\theta)=0$ \;
			        }
		
			\nl    $\mathcal{N}^{\theta}_{i}=\{\}$ \;
			        
			\nl $ResetFlag = False$\;

			\caption{Reset}\label{alg:reset}			
		\end{algorithm}
	\end{figure}

In Algorithm \ref{alg:ABU}, for agent $i$, hypothesis $\theta\in\Theta$ and neighbor set $\mathcal{N}_{i,t}$, Line \ref{algorithm:initialization} performs the initialization when $t=0$ or reset when $ResetFlag$ is true. \added{From Algorithm \ref{alg:reset}, the initialization sets  $b^a_j(\theta)$ to $0$ for all $j\in\mathcal{N},j\neq i$, where  $b^a_j(\theta)$ denotes the most recent AB of $\theta$  received from agent $j$.} Furthermore, Algorithm \ref{alg:reset} initializes $\mathcal{N}^\theta_{i}$ to an empty set. \added{The set $\mathcal{N}^\theta_{i}$ denotes the set of agents $j\neq i$ from which ABs are received and $j\in S(\theta,\theta')$ from some $\theta'\neq\theta$.} Finally, $ResetFlag$ gets set to $False$ to indicate that a reset has just been performed. Then Algorithm \ref{alg:ABU} loops over all agent $i$'s neighbors $j\in\mathcal{N}_{i,t}$ as shown in Line \ref{algorithm:neighbor}. The set $\mathcal{N}^\theta_{i}$ will include $j$ as shown in Line \ref{algorithm:increase by one}. Then we assign $b^a_j(\theta)$ the value of $b^a_{j,t}(\theta)$ in Line \ref{algorithm:acutual belief}.

\added{After all the ABs from neighbors are saved, as shown in Line \ref{algorithm:number of neighbor} we check if $|\mathcal{N}^\theta_{i}\cap S(\theta,\theta')|< 2f+1$ for any $\theta'\neq\theta$, i.e., the number of source agents for $\theta$ and $\theta'$ that  also have been agent $i$'s neighbors by time $t+1$ after last reset is less than $2f+1$. If yes, Algorithm \ref{alg:ABU} returns false to indicate there are not enough ABs received for $\theta'$ from the agents that can tell $\theta$ and $\theta'$ apart.} As a result, agent $i$ must select case two for $\theta$.

\added{If we reach Line \ref{algorithm:reset} in Algorithm \ref{alg:ABU}, it indicates that agent $i$ can safely update its AB of $\theta$ with case one. Therefore, we can use all the saved ABs for \eqref{equation:asynchronous acutual belief update rule for case one}, making them obsolete, and thus we need a reset at the next time step. Then Algorithm \ref{alg:ABU} returns true. }

In case one, like SDHT, we remove the $f$ lowest beliefs collected so far and use the minimum rule. \added{Note that, different from SDHT, we use ABs  $\{b^a_{j}(\theta)|j\in\mathcal{N}^{\theta}_{i}\}$ that are collected over time instead of  ABs of the neighboring agents at time $t+1$. If Algorithm \ref{alg:ABU} returns false, we will enter case two in ADHT where the rest will follow the same procedure as in SDHT.}

\begin{rem}\label{remark:ADHT}
\added{A key difference from ADHT from SDHT is the relaxed conditions to enter case one in the update rule.} In Section \ref{section:learning rules}, for an agent $i$ and hypothesis $\theta\in\Theta$, to enter case one, at a given time instant $t$, $|S(\theta,\theta')\cap\mathcal{N}_{i,t}|\geq 2f+1$ must be satisfied for all $\theta'\neq\theta$. 
\added{That is, the number of neighbors of agent $i$ at time $t$ that can differentiate $\theta$ and $\theta'$ must be no less than $2f+1$ {\bf{at that time instant.}}} \added{In ADHT, instead, we simply keep collecting the ABs for a hypothesis $\theta$ from agent $j\in S(\theta,\theta')$ {\bf{across possibly multiple time instants}}, until  the number of collected ABs from agents $j\in S(\theta,\theta')$ is at least $2f+1$ for any $\theta'\neq\theta$. This condition is also when Algorithm \ref{alg:ABU} returns true. It means that agent $i$ has collected enough ABs from agents that are once its neighbors up to time $t$ after the most recent reset to safely update its AB using \eqref{equation:acutual belief update rule for case one}.}

One can readily observe that the conditions to enter case one in SDHT imply that in ADHT. \added{Thus, the conditions in ADHT to update the AB using neighbor information are less conservative and more likely to be satisfied.} Therefore, the convergence rate can potentially improve due to more frequent use of non-local information. 
\end{rem}

With the proposed ADHT algorithm, we have the following theorem to show that the new update rule  also converges almost surely. 
\begin{thm}\label{theorem:main result1}
If the following conditions hold:
\begin{enumerate}
    \item The initial beliefs $b^l_{i,0}(\theta)>0$ and $b^a_{i,0}(\theta)>0$ for any $\theta\in\Theta$ and any  agent $i$.
    \item For any agent $i$, if case one in ADHT happens only finitely often for a  hypothesis $\theta\in\Theta$, then $i\in S(\theta,\theta')$ for any $\theta'\neq\theta$.
\end{enumerate}
Then ADHT ensures that $b^a_{i,t}(\theta^*)\rightarrow 1$ almost surely for any good agent $i\in G$ as $t\rightarrow\infty$.
\end{thm}
\begin{proof}
We only consider paths $\omega\in\hat{\Omega}$ as defined in Remark \ref{remark:lemma 1}.
\added{The proof consists of two parts where we only consider any good agent $i\in G$. First, we prove that the AB over the true hypothesis $b^a_{i,t}(\theta^*)$ is lower-bounded. Then we show that  the AB over the rest of the hypotheses will become arbitrarily small. These two parts together are sufficient to prove that almost surely the $b^a_{i,t}(\theta^*)$ will be arbitrarily close to one. }

\added{For the first part that lower-bounds $b_i^a(\theta^*)$, we study two different scenarios. For the first scenario, if case one only happens finitely often to an agent $i\in G$ with respect to the true hypothesis $\theta^*$, then by the second condition of Theorem \ref{theorem:main result1}, we know that it must happen that $i\in S(\theta^*,\theta)$ for any $\theta\neq\theta^*$. In other words, agent $i$ can distinguish $\theta^*$ from any other hypothesis $\theta\neq\theta^*$. By Lemma \ref{lemma:convergence}, we know that agent $i$ can then correctly identify $\theta^*$ by only LB update \eqref{equation:acutual belief update rule for case two} that runs infinitely often, i.e., $\lim_{t\rightarrow\infty} b^a_{i,t}(\theta)\rightarrow 1$. Then the whole proof is done.}

The second scenario indicates that case one in ADHT happens infinitely often to an agent $i\in G$ and $\theta^*$. In this scenario, we first show that the AB over the true hypothesis $b^a_{i,t}(\theta^*)$ is lower-bounded. For each good agent $j\in G$, there exist a time $t_j$ and a constant $\alpha$ such that, for all $t\geq t_j$, we have $b^l_{j,t}(\theta^*)\geq\delta_1-\alpha$ where $\alpha<\delta_1$. We define 
\begin{equation}\label{proof:bar t 1}
\bar{t}_1 := \max_{j\in G}t_j. 
\end{equation}
We also define $\delta_2: = \min_{j\in G}b^a_{j,\bar{t}_1}(\theta^*)$\footnote{\added{Intuitively, $\bar{t}_1 $ indicates a time instant since which $b^l_{j,t}(\theta^*)$ is bounded below for any good agent $j\in G$. And $\delta_2$ refers to the minimum AB over the true hypothesis $\theta^*$ for any good agent at that time instant $\bar{t}_1 $.}}. By Lemma \ref{lemma:nonzero of true hypothesis}, we know $\delta_2 >0$. We further define 
\begin{equation}\label{equation:delta}
    \delta:=\min\{\delta_1-\alpha,\delta_2\}. 
\end{equation}

\noindent \added{Since case one happens infinitely often, for agent $i$, there must exist a time $t'_i\geq\bar{t}_1$ such that Algorithm \ref{alg:ABU} returns true. As a result, $ResetFlag$ is set to true and after AB update with \eqref{equation:asynchronous acutual belief update rule for case one} at $t'_i$, all the saved ABs are deleted at $t'_i+1$. Then after $t'_i$, we know that for any good agent $j\in G$ whose AB is collected by agent $i$, it is guaranteed that} 
\begin{equation}\label{proof2:t'_i}
    \added{b_{j,t}^a(\theta^*)\geq\delta,\forall t\geq t'_i.}
\end{equation}
\added{Again, since case one happens infinitely often, there must also exist a time $t''_i>t'_i$ such that case one happens. }
\added{Then we have the following holds.}
\begin{equation}\label{proof2:case one}
 \added{\tilde{b}^a_{i,t''_i}(\theta^*)=\min\{\{b^a_{j}(\theta^*)\}_{j\in \tilde{\mathcal{N}}^{\theta^*}_{i}},b^l_{i,t''_i}(\theta^*)\}\geq\delta. }
\end{equation}
\added{The inequality \eqref{proof2:case one} holds despite possibly altered ABs from $f$ bad agents because of the following. The second term $b^l_{i,t''_i}(\theta^*)$ is no less than $\delta$ by the definition of $\delta$. As to the first term $\{b^a_{j}(\theta^*)\}_{j\in \tilde{\mathcal{N}}^{\theta^*}_{i}}$, its minimum is also guaranteed to be no less than $\delta$. We show this by contradiction. If it does happen that $\min{\{b^a_{j}(\theta^*)\}_{j\in \tilde{\mathcal{N}}^{\theta^*}_{i}}}<\delta$, then from \eqref{proof2:t'_i} we know that this minimum value can only come from a bad agent. However, since there are at most $f$ bad agents and we only eliminate $f$ agents with the $f$ lowest beliefs from $\{b^a_{j}(\theta^*)|j\in  \mathcal{N}^{\theta^*}_{i}\}$ to get $\tilde{\mathcal{N}}_{i}^{\theta^*}$, then it means that there is at least one good agent $k$ whose AB is in the $f$ lowest beliefs and got eliminated. Then it implies that any AB in the remaining set  $\{b^a_{j}(\theta^*)\}_{j\in \tilde{\mathcal{N}}^{\theta^*}_{i}}$ must be no less than this good agent $k$'s AB which got eliminated. But from \eqref{proof2:t'_i}, we know that $b_{k}^a(\theta^*)\geq\delta$, this implies that any AB $\{b^a_{j}(\theta^*)\}_{j\in \tilde{\mathcal{N}}^{\theta^*}_{i}}$ must be no less than $\delta$, which reaches a contradiction.}
Then we perform the normalization as in \eqref{equation:norm} and can derive
\begin{equation}\label{proof2:normalization1}
\begin{split}
b^a_{i,t''_i}(\theta^*) &= 
  \frac{\tilde{b}^a_{i,t''_i}(\theta^*)}{\sum_{p=1}^m \tilde{b}^a_{i,t''_i}(\theta_p)}\geq \frac{\delta}{\sum_{p=1}^m \tilde{b}^a_{i,t''_i}(\theta_p)}\\
  &\geq\frac{\delta}{\sum_{p=1}^m b^l_{i,t''_i}(\theta_p)}=\delta.
\end{split}
\end{equation}
The last inequality in \eqref{proof2:normalization1} holds since by \eqref{equation:acutual belief update rule for case one}, we know that $\tilde{b}^a_{i,t''_i}(\theta)\leq b^l_{i,t''_i}(\theta)$ for any $\theta\in\Theta$.


\added{At $t''_i+1$, if case one happens again,  we know that $\tilde{b}^a_{i,t''_i+1}(\theta^*)\geq\delta$ by the same logic that reaches \eqref{proof2:case one}. Alternatively, if case two happens at $t''_i+1$, we use update rule \eqref{equation:acutual belief update rule for case two}, giving:} 
\begin{equation}\label{proof2:case two}
     \added{\tilde{b}^a_{i,t''_i}(\theta^*)=\min\{b^a_{i,t''_i+1}(\theta^*),b^l_{i,\bar{t}_1+1}(\theta^*)\}\geq\delta.}
\end{equation} 
Therefore, no matter which case occurs,  we have $\tilde{b}^a_{i,t''_i+1}(\theta^*)\geq\delta$ before normalization. By the same logic that reaches \eqref{proof2:normalization1}, we know that

$$
b^a_{i,t''_i+1}(\theta^*)\geq\delta.
$$
after normalization. Then by induction, we have
\begin{equation}\label{proof2:part 1}
    b_{i,t}^a(\theta^*)\geq\delta,\forall t\geq t''_i.
\end{equation}
\added{We further define $\tilde{t}_2 := \max_i t''_i.$} By definition, we have that 
\begin{equation}\label{proof2:part 11}
    \added{b_{i,t}^a(\theta^*)\geq\delta,\forall t\geq \tilde{t}_2,\forall i\in G.}
\end{equation}
We have just proved that $b_{i}^a(\theta^*)$ is lower bounded. \added{Now we move on to prove that the ABs over any $\theta\neq\theta^*$ are upper bounded.}  \added{Given a hypothesis $\theta\neq\theta^*$, for any agent $i\in S(\theta,\theta^*)$, we pick a small $0<\epsilon<1$ such that $\epsilon<\delta$ and define $t_i^\theta$ 
such that}
\begin{equation}\label{proof4: epsilon 3}
\added{b_{i,t}(\theta)^l(\theta)\leq\epsilon^3,\forall t\geq t_i^\theta.}
\end{equation}
We can always find such $\epsilon, \delta$, and $t_i^\theta$ that \eqref{proof4: epsilon 3} holds by definition of $S(\theta,\theta^*)$ and Lemma \ref{lemma:convergence}. Then we further define
$$
\added{\tilde{t}_3 := \max\{\tilde{t}_2,\max_{i\in S(\theta,\theta^*)}\{t_i^\theta\}\}.}
$$
\added{It immediately follows that 
$$
\tilde{b}^a_{i,t}(\theta)\leq b_{i,t}^l(\theta)\leq\epsilon^3\le\epsilon, \forall t\geq \tilde{t}_3+1, \forall i\in S(\theta,\theta^*)\cap G.
$$ before normalization no matter case one or case two occurs. }
Then we perform the normalization as in \eqref{equation:norm} and can derive 
\begin{equation}\label{proof2:epsilon}
\begin{split}
  b^a_{i,\tilde{t}_3+1}(\theta^*)&=\frac{\tilde{b}^a_{i,\tilde{t}_3+1}(\theta)}{\sum_{p=1}^m \tilde{b}^a_{i,\tilde{t}_3+1}(\theta_p)}\leq \frac{\epsilon^3}{\sum_{p=1}^m \tilde{b}^a_{i,\tilde{t}_3+1}(\theta_p)}\\
  &\leq\frac{\epsilon^3}{ b^a_{i,\tilde{t}_3+1}(\theta^*)}\leq\frac{\epsilon^3}{\delta}<\epsilon^2  .
\end{split}
\end{equation}
The last inequality is due to the fact $ \epsilon<\delta$. Therefore, by induction we have that
\begin{equation}\label{proof2:epsilon 2}
    \added{b_{i,t}^a(\theta)\leq\epsilon^2\leq\epsilon,\forall t\geq \tilde{t}_3+1,\forall i\in S(\theta,\theta^*)\cap G.}
\end{equation}
For any good agent $i\notin S(\theta,\theta^*)$, by condition 2 in Theorem \ref{theorem:main result1}, case one will happen infinitely often for $\theta$, \added{and there must exist two time instants $\tilde{t}_{i,1}^\theta$ and $\tilde{t}_{i,2}^\theta$ where case one happens for the first time and the second time after $t\geq\tilde{t}_3+1$.} Following  a similar reasoning that reaches \eqref{proof2:part 1}, we have that
\begin{equation}\label{proof2:part 2}
    \added{b_{i,t}^a(\theta)\leq \epsilon,\forall t\geq \tilde{t}_{i,2}^\theta.}
\end{equation}
Then we further define
\begin{equation}
    \added{\tilde{t}_4 := \max_\theta\max_{i\notin S(\theta,\theta^*)} \tilde{t}_{i,2}^\theta.}
\end{equation}
\added{By definition, we know that $\tilde{t}_4>\tilde{t}_3$,} then it holds that

\begin{equation}\label{proof2:part 22}
    \added{b_{i,t}^a(\theta)\leq\epsilon,\forall t\geq\tilde{t}_4,\forall i\in G, \forall \theta\neq\theta^*.}
\end{equation}
\added{Combining \eqref{proof2:part 11} and \eqref{proof2:part 22}, for any state observation path $\omega\in\hat{\Omega}$, \added{$\lim_{t\rightarrow\infty} b^a_{i,t}(\theta^*)= 1$}. Since the set $\hat{\Omega}$ has measure one as established in Remark \ref{remark:lemma 1}, $\lim_{t\rightarrow\infty} b^a_{i,t}(\theta^*)= 1$ almost surely. The proof of Theorem \ref{theorem:main result1} is thus complete.}
\end{proof}

\section{{Learning with average rule}}\label{section:average rule}
\added{Both the update rules in Section \ref{section:learning rules} and Section \ref{section:Learning rule with asynchronous update} use the minimum rule in case one  when updating the ABs using neighboring information, as shown in \eqref{equation:acutual belief update rule for case one} and \eqref{equation:asynchronous acutual belief update rule for case one}.} While we can prove the convergence, such an update algorithm may result in a large variance and waste the neighbors' information since it will only use information from 
\emph{one} of the neighboring agents for each hypothesis. \added{Therefore, in this section, we introduce an alternative approach for applying the neighbor's ABs in case one as shown below.}

We discuss the changes with respect to SDHT but these results naturally carry over to ADHT. Note that we only discuss the changes for case one, while case two remains the same. First, we change the condition to enter case one as in Line \ref{algorithm:AB update }  of Algorithm \ref{alg:Synchronized Distributed Hypothesis Testing} from \emph{for all $\theta'\neq\theta$, $|S(\theta,\theta')\cap\mathcal{N}_{i,t+1}|\geq 2f+1$} to  \emph{for all $\theta'\neq\theta$, $|S(\theta,\theta')\cap\mathcal{N}_{i,t+1}|\geq 2f+2$.}
Second, as shown in Line \ref{algorithm:remove beliefs} of Algorithm \ref{alg:Synchronized Distributed Hypothesis Testing}, at time $t+1$, notice that previously before updating an agent's AB for an particular hypothesis $\theta$, we first remove the $f$ lowest shared ABs. In this section, instead, we do the following. We sort the ABs with respect to an hypothesis $\theta$ shared by the neighbor set $\mathcal{N}_{i,t+1}$  and divide $\mathcal{N}_{i,t+1}$ into three pair-wise disjoint sets $\mathcal{L}_{i,t+1}^{\theta}$, $\mathcal{M}_{i,t+1}^{\theta}$, and $\mathcal{H}_{i,t+1}^{\theta}$, where
\begin{itemize}
    \item $\mathcal{L}_{i,t+1}^{\theta}$ is the set of neighboring agents that has the lowest $f$ ABs with respect to $\theta$;
    \item $\mathcal{H}_{i,t+1}^{\theta}$ is the smallest set of neighboring agents that has the highest ABs with respect to $\theta$ and 
    \begin{equation}\label{equation:average rule highest set}
        \mathcal{H}_{i,t+1}^{\theta}\cap S(\theta,\theta')\geq f+1, \forall\theta'\neq\theta;
    \end{equation}
    \item $\mathcal{M}_{i,t+1}^{\theta}:=\mathcal{N}_{i,t+1}\backslash(\mathcal{L}_{i,t+1}^{\theta}\bigcup\mathcal{H}_{i,t+1}^{\theta})$.
\end{itemize}
\begin{rem}
The set $\mathcal{L}_{i,t+1}^{\theta}$ is the same set that is eliminated in Line \ref{algorithm:remove beliefs} of Algorithm \ref{alg:Synchronized Distributed Hypothesis Testing}. The set $\mathcal{H}_{i,t+1}^{\theta}$ is more involved. The neighboring agents in this set has the highest ABs with respect to $\theta$, meaning that for any agents $j\in \mathcal{H}_{i,t+1}^{\theta}$ and $j'\notin\mathcal{H}_{i,t+1}^{\theta}$, it is guaranteed that $b^a_{j,t}(\theta)\geq b^a_{j',t}(\theta)$. Furthermore, note that $\mathcal{H}_{i,t+1}^{\theta}$ may not be unique if more than one neighboring agents have the same AB with respect to $\theta$ and this AB value is the lowest for any agents in $\mathcal{H}_{i,t+1}^{\theta}$. In this case, we may pick any combinations of these agents as long as \eqref{equation:average rule highest set} is satisfied (we will illustrate this in Example \ref{example:average rule} below). To find such $\mathcal{H}_{i,t+1}^{\theta}$, we can do a brute force search for agents with the highest ABs and increase the number of agents until \eqref{equation:average rule highest set} is satisfied for the first time, or a binary search if the number of agents and/or the number of hypotheses are large.
\end{rem}
\begin{rem}
Another question one may ask is that can $\mathcal{M}_{i,t+1}^{\theta}$  become empty, since the cardinality of $\mathcal{H}_{i,t+1}^{\theta}$ is only lower bounded by $f+1$ but may not be a fixed number. We show that this is not possible by contradiction. Since $\mathcal{H}_{i,t+1}^{\theta}$ is the smallest set  that satisfies \eqref{equation:average rule highest set}, there must exist at least one hypothesis $\theta'\neq\theta$ such that $\mathcal{H}_{i,t+1}^{\theta}\cap S(\theta,\theta')= f+1$. Then for this particular $\theta'$, if $\mathcal{M}_{i,t+1}^{\theta}$ is indeed empty, it will hold that $\mathcal{N}_{i,t+1}\cap S(\theta,\theta')\leq 2f+1$ since  $\mathcal{L}_{i,t+1}^{\theta}\cap S(\theta,\theta')\leq f$. This will contradict the condition to enter case one in Line \ref{algorithm:AB update }  of Algorithm \ref{alg:Synchronized Distributed Hypothesis Testing} with $2f+2$ instead of $2f+1$.
\end{rem}
\begin{example}\label{example:average rule}
We use this example to illustrate how $\mathcal{L}_{i,t+1}^{\theta}$, $\mathcal{M}_{i,t+1}^{\theta}$, and $\mathcal{H}_{i,t+1}^{\theta}$ are determined for SDHT. Suppose $f=1,\Theta=\{\theta,\theta'\}, N=10, i=0$, and $\mathcal{N}_{0,t+1}=\{0,1,2,3,4,5\}$, meaning that there are $10$ agents in total and at $t+1$, there are $6$ agents within the communication range of agent $0$ (note that agent $0$ is a neighbor of itself). We focus on $\theta$, where $S(\theta,\theta')=\{1,2,3,4\}$, therefore it follows that
$$S(\theta,\theta')\cap\mathcal{N}_{0,t+1}=\{1,2,3,4\}$$ 
and 
$$|S(\theta,\theta')\cap\mathcal{N}_{0,t+1}|=4\geq 2f+2=4.$$
As a result, the condition to enter case one for average rule is satisfied. The shared ABs with respect to $\theta$ are
$
b^a_{0,t}(\theta) = 0.25, b^a_{1,t}(\theta) = 0.22,b^a_{2,t}(\theta) = 0.35,b^a_{3,t}(\theta) = 0.35,b^a_{4,t}(\theta) = 0.35, b^a_{5,t}(\theta) = 0.37.
$
Then we sort ABs and $\mathcal{L}_{i,t+1}^{\theta}=\{1\}$ since agent $1$ has the lowest AB over $\theta$ and we only need $f=1$ agent in $\mathcal{L}_{i,t+1}^{\theta}$. For $\mathcal{H}_{i,t+1}^{\theta}$, we know that agent $5$ will be included since it has the hightest AB and any two of agents $\{2,3,4\}$ can be included since $b^a_{2,t}(\theta) = b^a_{3,t}(\theta)=b^a_{4,t}(\theta)=0.35$ and all of them belong to $S(\theta,\theta')$. Therefore  $\mathcal{H}_{i,t+1}^{\theta}$ is not unique and we can arbitrarily pick any combination of agents $\{2,3,4\}$ in this particular example. Suppose we pick  $\mathcal{H}_{i,t+1}^{\theta}=\{3,4,5\}$, then it immediately follows that  $\mathcal{M}_{i,t+1}=\{0,2\}$.
\end{example}
\begin{lem}\label{lemma:lower and upper bounded beliefs}
In SDHT (Algorithm \ref{alg:Synchronized Distributed Hypothesis Testing}), if 
\begin{itemize}
    \item replace the condition to enter case one (in Line \ref{algorithm:AB update }  of Algorithm \ref{alg:Synchronized Distributed Hypothesis Testing}) to for all $\theta'\neq\theta$, $|S(\theta,\theta')\cap\mathcal{N}_{i,t+1}|\geq 2f+2$, and
    \item replace Line 4 of Algorithm \ref{alg:Synchronized Distributed Hypothesis Testing} for case one with the procedure to get $\mathcal{M}_{i,t+1}^{\theta}$ instead of $\mathcal{N}_{i,t+1}^{\theta}$,
\end{itemize}  
then for any agent $j\in \mathcal{M}_{i,t+1}^{\theta}$ and $\theta'\neq\theta$, there exist  neighboring agents $j'\in\mathcal{N}_{i,t+1}$ and $j''\in\mathcal{N}_{i,t+1}\cap S(\theta,\theta')$ that both are good and
\begin{equation}\label{equation:convex hull}
    b^a_{j',t}(\theta)\leq b^a_{j,t}(\theta)\leq b^a_{j'',t}(\theta).
\end{equation}
\end{lem}

\begin{proof}
\added{We first prove $b^a_{j',t}(\theta)\leq b^a_{j,t}$. If $j\in \mathcal{M}_{i,t+1}^{\theta}$ is a good agent, we can set $j=j'$ and $b^a_{j',t}(\theta)\leq b^a_{j,t}$ trivially holds. If $j\in \mathcal{M}_{i,t+1}^{\theta}$ is a bad agent, since there are at most $f$ bad agents in $\mathcal{N}_{i,t+1}$, it implies that  there exists at least one good agent $j'\in\mathcal{L}_{i,t+1}^{\theta}$. Otherwise, since $|\mathcal{L}_{i,t+1}^{\theta}|=f$, if there is no good agent in $\mathcal{L}_{i,t+1}^{\theta}$, it implies that  $\mathcal{L}_{i,t+1}^{\theta}$ contains all bad agents and thus there cannot exist a bad agent in $\mathcal{M}_{i,t+1}^{\theta}$ which leads to a contradiction. Therefore, $b^a_{j',t}(\theta)\leq b^a_{j,t}$ is proved.}

\added{Now we prove $b^a_{j,t}(\theta)\leq b^a_{j'',t}(\theta)$ for $j''\in\mathcal{N}_{i,t+1}\cap S(\theta,\theta')$. This is an immediate result from the fact that there are at most $f$ bad agents and there are at least $f+1$ agents that belong to  
$\mathcal{N}_{i,t+1}\cap S(\theta,\theta')\cap\mathcal{H}_{i,t+1}^{\theta}$ according to \eqref{equation:average rule highest set}. Therefore, there must exist at least one good agent in $\mathcal{N}_{i,t+1}\cap S(\theta,\theta')\cap\mathcal{H}_{i,t+1}^{\theta} $. Denote this good agent as $j''$, we know that $b^a_{j,t}(\theta)\leq b^a_{j'',t}(\theta)$ holds by the definition of $\mathcal{H}_{i,t+1}^{\theta}$.
Combine $b^a_{j',t}(\theta)\leq b^a_{j,t}$ and $b^a_{j,t}(\theta)\leq b^a_{j'',t}(\theta)$, we know that \eqref{equation:convex hull} holds. }
\end{proof}

For ADHT, we can have a lemma below that is a counterpart of Lemma \ref{lemma:lower and upper bounded beliefs}. 
\begin{lem}\label{lemma:lower and upper bounded beliefs ADHT}
 In ADHT (Algorithm \ref{alg:Asynchronous Distributed Hypothesis Testing}), if
\begin{itemize}
    \item replace the condition to check whether to enter case one (in Line \ref{algorithm:number of neighbor}  of Algorithm \ref{alg:ABU}) to for all $\theta'\neq\theta$, $|\mathcal{N}^{\theta}_{i}\cap S(\theta,\theta')|< 2f+2$, and
    \item replace Line 4 for case one in Algorithm \ref{alg:Asynchronous Distributed Hypothesis Testing} with the procedure to  get $\mathcal{M}_{i,t+1}^{\theta}$ instead of $\tilde{\mathcal{N}}^\theta_{i}$,
\end{itemize}
then for any $j\in \mathcal{M}_{i,t+1}^{\theta}$ and $\theta'\neq\theta$, there exist  agents $j'\in\mathcal{N}_{i}^\theta$ and $j''\in\mathcal{N}_{i}^\theta\cap S(\theta,\theta')$ that both are good and
\begin{equation}\label{equation:convex hull ADHT}
    b^a_{j'}(\theta)\leq b^a_{j}(\theta)\leq b^a_{j''}(\theta).
\end{equation}
\end{lem}
\begin{proof}
We omit the proof here since it is similar to the proof of Lemma \ref{lemma:lower and upper bounded beliefs}.
\end{proof}

\added{We further define the average of beliefs in $\mathcal{M}^\theta_{i,t+1}$ as}
\begin{equation}\label{equation:average belief}
   \bar{b}^{a}_{i,t+1}(\theta) :=\frac{1}{|\mathcal{M}^\theta_{i,t+1}|} \sum_{j\in \mathcal{M}^\theta_{i,t+1}} b^a_{j,t}(\theta),
\end{equation}
\added{and instead of \eqref{equation:acutual belief update rule for case one}, we use the following  rule}
\begin{equation}\label{equation:acutual belief update rule average}
     \tilde{b}^a_{i,t+1}(\theta) = \min\{\bar{b}^{a}_{i,t+1}(\theta),b^l_{i,t+1}(\theta)\}.   
    \end{equation}
    
\added{For ADHT, we use the update rule similar to  \eqref{equation:average belief} and \eqref{equation:acutual belief update rule average} to replace \eqref{equation:asynchronous acutual belief update rule for case one} but the  sets $\mathcal{L}_{i,t+1}^{\theta}$, $\mathcal{M}_{i,t+1}^{\theta}$, and $\mathcal{H}_{i,t+1}^{\theta}$ are found from $\mathcal{N}^\theta_i$ instead}.
\begin{example}
Following Example \ref{example:average rule}, recall that $\mathcal{M}_{i,t+1}=\{0,2\}$ and $
b^a_{0,t}(\theta) = 0.25, b^a_{2,t}(\theta) = 0.35
$. Then by \eqref{equation:average belief}, 
$$
\bar{b}^{a}_{0,t+1}(\theta) = \frac{1}{2}(0.25+0.35) = 0.3.
$$
\end{example}
\added{Next, we show that the convergence is still guaranteed using the average update rule by the following two theorems. }
\begin{thm}\label{thm:3}
\added{If the same changes are made as in Lemma \ref{lemma:lower and upper bounded beliefs} to SDHT (Algorithm \ref{alg:Synchronized Distributed Hypothesis Testing}), the same conditions in Theorem \ref{theorem:main result} hold and we use \eqref{equation:acutual belief update rule average} for case one in SDHT where $\tilde{b}^{a}_{i,t}$ is from \eqref{equation:average belief}, $b^a_{i,t}(\theta^*)\rightarrow 1$ almost surely for any good agent $i$ as $t\rightarrow\infty$.}
\end{thm}

The proof of Theorem \ref{thm:3} is similar to that of Theorem \ref{theorem:main result}. We explain the proof sketch here and the complete proof can be found in the appendix.

\emph{Proof sketch}: 
Like the proof of Theorem \ref{theorem:main result}, we prove the convergences in two parts. The first part shows that the AB over the true hypothesis for any good agent $i$ is lower-bounded from zero. The second part shows that the AB over any hypothesis other than the true hypothesis is upper-bounded by an arbitrarily small constant.

The main differences in the proof, when compared to the proof of Theorem \ref{theorem:main result}, are the following. For part one, to establish that the AB over the true hypothesis is lower-bounded from zero, for case one, instead of referring to the fact that the $\mathcal{N}^\theta_{i,t+1}$ contains at least one good agent whose AB over the true hypothesis is guaranteed to be nonzero from Lemma \ref{lemma:nonzero of true hypothesis}, we use  \eqref{equation:convex hull} in Lemma \ref{lemma:lower and upper bounded beliefs} together with Lemma \ref{lemma:nonzero of true hypothesis}. In other words, the AB over the true hypothesis $\theta^*$ for any agent $j$ in $\mathcal{M}_{i,t+1}^{\theta}$ is guaranteed to be bounded away from zero since $b^a_{j',t}(\theta^*)\leq b^a_{j,t}(\theta^*)$ where $j'\in\mathcal{N}_{i,t+1}$ is good and from Lemma \ref{lemma:nonzero of true hypothesis} we know that $b^a_{j',t}(\theta^*)>0$ when any $j'$ is good. Therefore, the average $\bar{b}^{a}_{i,t+1}(\theta)$ computed in \eqref{equation:average belief} is guaranteed to be lower-bounded from zero, and $\tilde{b}^{a}_{i,t+1}(\theta)$ in \eqref{equation:acutual belief update rule average} is also guaranteed to be lower-bounded from zero.

For part two, we need to establish that the AB over the true hypothesis is upper-bounded by an arbitrarily small constant.  For case one, we use  \eqref{equation:convex hull} in Lemma \ref{lemma:lower and upper bounded beliefs} together with Lemma \ref{lemma:convergence}, which differs from the approach Theorem \ref{theorem:main result} that relies on the fact that the $\mathcal{N}^\theta_{i,t+1}$ contains at least one good agent whose AB over hypotheses other than the true hypothesis is guaranteed to be upper-bounded from Lemma \ref{lemma:nonzero of true hypothesis}. For any hypothesis $\theta\neq\theta^*$, we know that there exists a good agent $j''\in\mathcal{N}_{i,t+1}\cap S(\theta,\theta')$ such that $b^a_{j,t}(\theta)\leq b^a_{j'',t}(\theta)$. Furthermore, it is guaranteed that  $b^a_{j'',t}(\theta)$ will be upper-bounded by an arbitrarily small constant from Remark \ref{remark:for lemma 1} for Lemma \ref{lemma:convergence}.

\begin{thm}\label{thm:4}
\added{If the same changes are made as in Lemma \ref{lemma:lower and upper bounded beliefs ADHT}, the same conditions in Theorem \ref{theorem:main result1} hold and we use \eqref{equation:acutual belief update rule average} for case one in ADHT where $\tilde{b}^{a}_{i,t}$ is from \eqref{equation:average belief}, $b^a_{i,t}(\theta^*)\rightarrow 1$ almost surely for any good agent $i$ as $t\rightarrow\infty$.}
\end{thm}
The proof of Theorem \ref{thm:4} is similar to that of Theorem \ref{theorem:main result1} and Theorem \ref{thm:3}. We again explain the proof sketch here and the complete proof can be found in the appendix.

\added{\emph{Proof sketch}: 
Like the proof of Theorem \ref{theorem:main result}, we prove the convergences in two parts where we show that 1) the AB over the true hypothesis for any good agent $i$ is lower-bounded from zero and 2) the AB over any hypothesis $\theta\neq\theta^*$ is upper-bounded by an arbitrarily small constant. The main difference is similar to the difference between Theorem \ref{theorem:main result} and Theorem \ref{thm:3} where we make use of \eqref{equation:convex hull ADHT} in Lemma \ref{lemma:lower and upper bounded beliefs ADHT} to establish the two parts that we need to prove.}

\added{
By $b^a_{j'}(\theta^*)\leq b^a_{j}(\theta^*)$ from \eqref{equation:convex hull ADHT} for a good agent $j'$ and any $j\in\mathcal{M}_{i,t+1}^{\theta}$ we know that the average $\tilde{b}^{a}_{i,t}(\theta)$ computed in \eqref{equation:average belief} is guaranteed to be lower-bounded from zero, and $\tilde{b}^{a}_{i,t+1}(\theta)$ in \eqref{equation:acutual belief update rule average} is also guaranteed to be lower-bounded from zero. By $b^a_{j}(\theta)\leq b^a_{j''}(\theta)$ for a good agent $j''\in\mathcal{N}_{i}^\theta\cap S(\theta,\theta')$ and any $j\in\mathcal{M}_{i,t+1}^{\theta}$, we know that the average $\bar{b}^{a}_{i,t+1}(\theta)$ in \eqref{equation:average belief} is upper-bounded by an arbitrarily small constant which leads to the fact that $\tilde{b}^{a}_{i,t+1}(\theta)$ in \eqref{equation:acutual belief update rule average} is also upper-bounded by an arbitrarily small constants since $\tilde{b}^{a}_{i,t+1}(\theta)\leq \bar{b}^{a}_{i,t+1}(\theta)$.}

\section{Case Study}\label{section:case study}

In this section, we consider a case study with a team of UAVs in a gridworld environment, as shown in Fig.~\ref{fig:case_env}. The objective is to identify the unknown set of compromised (bad) UAVs out of the UAV team.

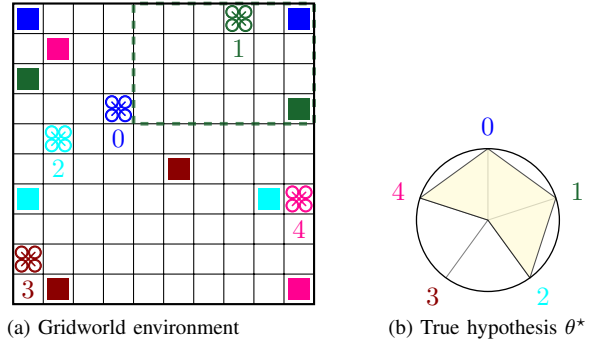
\begin{figure}[h]
    \subfloat[Gridworld environment \label{fig:case_env}]{\newcommand{\gridScale}{0.4}
\newcommand{\fillGridAt}[3]{
	\node [xshift=.5*\gridScale cm,yshift=.5*\gridScale cm] at (#1*\gridScale,#2*\gridScale){#3};
}

\newcommand{\obsactcle}[3]{
\fill[#3] (#1*\gridScale+0.15*\gridScale,#2*\gridScale+0.15*\gridScale) rectangle (#1*\gridScale+0.85*\gridScale,#2*\gridScale+0.85*\gridScale);
}

\definecolor{maroon}{RGB}{139,0,0}
\definecolor{blueish}{RGB}{125,100,255}
\definecolor{greenish}{RGB}{26,102,46}
\definecolor{pink}{RGB}{255,0,144}
\definecolor{cyan}{RGB}{0,255,255}

\newcommand{\UAV}[4]{
\pgfmathsetmacro\linetck{#3/500.0}
\pgfmathsetmacro\uavcorn{#3/4.0+5*\linetck}
\draw[line width = \linetck mm, color=#4] (#1,#2) circle (10*\linetck);
\draw[line width = \linetck mm, color=#4] (#1+\uavcorn,#2+\uavcorn) circle (0.92*\uavcorn);
\draw[line width = \linetck mm, color=#4] (#1-\uavcorn,#2-\uavcorn) circle (0.92*\uavcorn);
\draw[line width = \linetck mm, color=#4] (#1+\uavcorn,#2-\uavcorn) circle (0.92*\uavcorn);
\draw[line width = \linetck mm, color=#4] (#1-\uavcorn,#2+\uavcorn) circle (0.92*\uavcorn);

\draw[line width = \linetck mm, color=#4] (#1+\uavcorn,#2+\uavcorn) circle (0.96*\uavcorn);
\draw[line width = \linetck mm, color=#4] (#1-\uavcorn,#2-\uavcorn) circle (0.96*\uavcorn);
\draw[line width = \linetck mm, color=#4] (#1+\uavcorn,#2-\uavcorn) circle (0.96*\uavcorn);
\draw[line width = \linetck mm, color=#4] (#1-\uavcorn,#2+\uavcorn) circle (0.96*\uavcorn);

\draw[line width = \linetck mm, color=#4] (#1+\uavcorn,#2+\uavcorn) circle (0.8*\uavcorn);
\draw[line width = \linetck mm, color=#4] (#1-\uavcorn,#2-\uavcorn) circle (0.8*\uavcorn);
\draw[line width = \linetck mm, color=#4] (#1+\uavcorn,#2-\uavcorn) circle (0.8*\uavcorn);
\draw[line width = \linetck mm, color=#4] (#1-\uavcorn,#2+\uavcorn) circle (0.8*\uavcorn);

\draw[line width = \linetck mm, color=#4] (#1+\uavcorn,#2+\uavcorn) circle (4*\linetck);
\draw[line width = \linetck mm, color=#4] (#1-\uavcorn,#2-\uavcorn) circle (4*\linetck);
\draw[line width = \linetck mm, color=#4] (#1+\uavcorn,#2-\uavcorn) circle (4*\linetck);
\draw[line width = \linetck mm, color=#4] (#1-\uavcorn,#2+\uavcorn) circle (4*\linetck);

\draw[line width = 20*\linetck mm, color=#4] (#1-7*\linetck, #2+7*\linetck) -- (#1+\uavcorn-3*\linetck,#2+\uavcorn+3*\linetck);
\draw[line width = 20*\linetck mm, color=#4] (#1+7*\linetck, #2-7*\linetck) -- (#1+\uavcorn+3*\linetck,#2+\uavcorn-3*\linetck);

\draw[line width = 20*\linetck mm, color=#4] (#1+7*\linetck, #2+7*\linetck) -- (#1+\uavcorn+3*\linetck,#2-\uavcorn+3*\linetck);
\draw[line width = 20*\linetck mm, color=#4] (#1-7*\linetck, #2-7*\linetck) -- (#1+\uavcorn-3*\linetck,#2-\uavcorn-3*\linetck);

\draw[line width = 20*\linetck mm, color=#4] (#1-7*\linetck, #2+7*\linetck) -- (#1-\uavcorn-3*\linetck,#2-\uavcorn+3*\linetck);
\draw[line width = 20*\linetck mm, color=#4] (#1+7*\linetck, #2-7*\linetck) -- (#1-\uavcorn+3*\linetck,#2-\uavcorn-3*\linetck);

\draw[line width = 20*\linetck mm, color=#4] (#1+7*\linetck, #2+7*\linetck) -- (#1-\uavcorn+3*\linetck,#2+\uavcorn+3*\linetck);
\draw[line width = 20*\linetck mm, color=#4] (#1-7*\linetck, #2-7*\linetck) -- (#1-\uavcorn-3*\linetck,#2+\uavcorn-3*\linetck);
}

\begin{tikzpicture}

\obsactcle{0}{9}{blue}
\obsactcle{9}{9}{blue}
\obsactcle{9}{6}{greenish}
\obsactcle{0}{7}{greenish}
\obsactcle{1}{8}{pink}
\obsactcle{9}{0}{pink}
\obsactcle{0}{3}{cyan}
\obsactcle{8}{3}{cyan}
\obsactcle{1}{0}{maroon}
\obsactcle{5}{4}{maroon}

\UAV{3*\gridScale+0.5*\gridScale}{6*\gridScale+0.5*\gridScale}{0.85*\gridScale}{blue}
\UAV{0*\gridScale+0.5*\gridScale}{1*\gridScale+0.5*\gridScale}{0.85*\gridScale}{maroon}
\UAV{7*\gridScale+0.5*\gridScale}{9*\gridScale+0.5*\gridScale}{0.85*\gridScale}{greenish}
\UAV{9*\gridScale+0.5*\gridScale}{3*\gridScale+0.5*\gridScale}{0.85*\gridScale}{pink}
\UAV{1*\gridScale+0.5*\gridScale}{5*\gridScale+0.5*\gridScale}{0.85*\gridScale}{cyan}

\fillGridAt{3}{5}{\textcolor{blue}{$0$}}
\fillGridAt{7}{8}{\textcolor{greenish}{$1$}}
\fillGridAt{0}{0}{\textcolor{maroon}{$3$}}
\fillGridAt{1}{4}{\textcolor{cyan}{$2$}}
\fillGridAt{9}{2}{\textcolor{pink}{$4$}}

\draw[black,line width=1*\gridScale pt] (0,0) grid[step=\gridScale] (10*\gridScale,10*\gridScale);

\draw[dashed,line width=3.5*\gridScale,color=greenish,opacity=0.75] (4*\gridScale,6*\gridScale) rectangle (10*\gridScale,10*\gridScale);

\draw[black,line width=2*\gridScale pt] (0,0) rectangle (10*\gridScale,10*\gridScale);

\end{tikzpicture}}\hfill
    \subfloat[True hypothesis $\theta^\star$]{\newcommand{\beliefSize}{0.95}





\begin{tikzpicture}
    \draw[black,line width=0.5pt] (0,0) circle (1*\beliefSize);

    \coordinate (origin) at (0, 0);
    
     \coordinate (1) at (90: 1*\beliefSize);
     \coordinate (edge-1) at (90: 1*\beliefSize);
     \node[inner sep=0.7pt,outer xsep=-2pt] (title-1) at (90:1.25*1) {\textcolor{blue}{$0$}};
     \draw[opacity=0.5] (origin) -- (edge-1);

     \coordinate (2) at (18: 1*\beliefSize);
     \coordinate (edge-2) at (18: 1*\beliefSize);
     \node[inner sep=0.7pt,outer xsep=-2pt] (title-2) at (18:1.25*1) {\textcolor{greenish}{$1$}};
      \draw[opacity=0.5] (origin) -- (edge-2);

     \coordinate (3) at (306: 1*\beliefSize);
     \coordinate (edge-3) at (306: 1*\beliefSize);
     \node[inner sep=0.7pt,outer xsep=-2pt] (title-3) at (306:1.25*1) {\textcolor{cyan}{$2$}};
      \draw[opacity=0.5] (origin) -- (edge-3);

     \coordinate (4) at (234: 0*\beliefSize);
     \coordinate (edge-4) at (234: 1*\beliefSize);
     \node[inner sep=0.7pt,outer xsep=-2pt] (title-4) at (234:1.25*1) {\textcolor{maroon}{$3$}};
      \draw[opacity=0.5] (origin) -- (edge-4);

     \coordinate (5) at (162: 1*\beliefSize);
     \coordinate (edge-5) at (162: 1*\beliefSize);
     \node[inner sep=0.7pt,outer xsep=-2pt] (title-5) at (162:1.25*1) {\textcolor{pink}{$4$}};
      \draw[opacity=0.5] (origin) -- (edge-5);

    \draw [fill=yellow!20, opacity=.7] (1)
                                \foreach \i in {2,...,5}{-- (\i)} --cycle;
\end{tikzpicture}}
    \caption{a) Case study environment - 5 agents each with an observation and communication range (only Agent 1's range is shown). b) The true hypothesis of the system $\theta = (1,1,1,0,1)$ in a radar plot over the probability simplex. The closer a vertex is to the edge of the radar plot, the higher that the belief of the corresponding agent is good.}
    \label{fig:environment}
\end{figure}

\begin{figure}
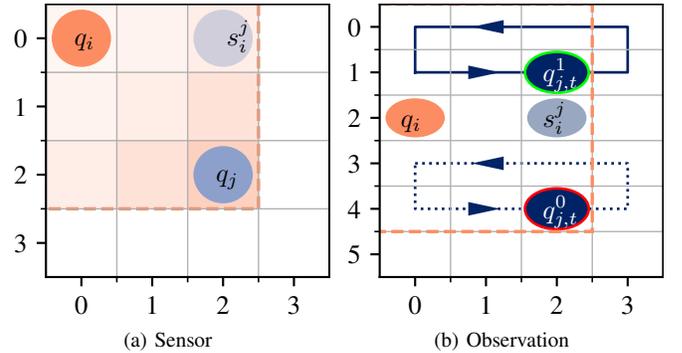

    \centering
    \subfloat[Sensor \label{fig:sen}]{
    \centering
    \input{observation.pgf}
    }
    \subfloat[Observation \label{fig:obs}]{
    \centering
    \input{likelihood_function.pgf}
    }
    \caption{(a) An example distribution and sensor output ($s^j_i = [2,0]$) for agent $i$ (true location $q_i = [0,0]$) sensing agent $j$ (true location $q_j = [2,2]$). Darker shades of orange indicate a higher probability of the sensor reading at that location (left).
    (b) An example likelihood function $l_i(s^j_i|q_{i},\theta(j))$ for a pair of agents $i$ and $j$ with  two possible models for agent $j$ ($\theta(j) \in \lbrace 0,1\rbrace$). 
    }
    \label{fig:observation}
\end{figure}

\subsection{Setting}\label{ssec:setting}

{We examine the proposed algorithms with 5 agents among which there is one bad agent. 
}
All the UAVs are at similar altitudes. Therefore, the state set $Q$ is the set of the two-dimensional locations in the gridworld. For agent $i$ at time $t$, its state is represented by  $q_{i,t} = [q^x_{i,t},q^y_{i,t}]$. We assign each individual agent a persistent surveillance task with a given state path. 

{Every agent has a communication and sensor range of $3$ units, i.e., they can view the locations that are within a $7\times 7$ square centered around the agent's position $q_i$ (see Fig.~\ref{fig:sen} for an example).}
Each agent $i$ could be either good or bad, therefore we denote a set $\Theta_i=\{0,1\}$, where $0$ denotes bad and $1$ denotes good. The hypothesis set is then $\Theta=\prod_i\Theta_i$. {For a hypothesis $\theta\in\Theta$, $\theta(i)$ denotes the hypothesis for agent $i$.}
The true hypothesis $\theta^*$ is the tuple $\theta^*=(1,1,1,0,1)$, i.e., all agents are good except for agent 3 since $\theta^*(3)=0$.



\subsection{Observation Model}
\subsubsection{Sensor}
 \added{If agent $i$ is at a location $q_{i}$, it will make an observation $s_i^j\in Q$ of agent $j$.} We use $ \mathcal{Q}_i(q_i)\subseteq Q $  to denote the set of locations that can be observed by agent $i$ at $q_i$.
 \added{If $q_j\in \mathcal{Q}_i(q_i) $, then agent $j$ is within the observation range of agent $i$. However, note that due to the observation noise, it is possible that $s_i^j\neq q_j$.   The probability of getting an observation $s_i^j$ for agent $j$ follows a probability distribution over $\mathcal{Q}_i(q_i)$ conditioned on $q_i$ and $q_j$, i.e., the locations of agent $i$ and $j$. In this example, we assume that this probability distribution is a truncated Gaussian distribution, a common choice in state estimation with noisy sensors  \cite{simon2006optimal}.} We center the distribution  around the actual location $q_j$ of agent $j$ and with a prescribed variance $\sigma^2$ (see Fig.~\ref{fig:sen}). \added{Intuitively, it means that the probability of observing $s_i^j=q_j$ is the highest and the probabilities of getting observations other than $s_i^j=q_j$ decreases as $s_i^j$ is further way from $q_j$.} As a result, the probability of agent $i$ observing $s_i^j$   is
\begin{align}
    P_{i}(s_i^j|{q_i},q_j) = \frac{e^{-\frac{1}{2\sigma^2}\norm{s_i^j-q_j}^2_2}}{\sum_{q\in \mathcal{Q}_i(q_i) } e^{-\frac{1}{2\sigma^2}\norm{q-q_j}^2_2}}.
    \label{equation:sensor probability}
\end{align}

If $q_j\notin  \mathcal{Q}_i(q_i)$, agent $i$ cannot observe agent $j$ and thus obtains an empty observation, i.e., $s_i^j=\emptyset$. To summarize, the observation $s_i^j$ follows 
\begin{align}\label{equation:obervation function for agent j in case study}
    s_i^j = \begin{cases}
    q \text{ with probability } P_{i}(q|{q_i},q_j) & \text{ if } q_j \in \mathcal{Q}_i(q_i),  \\
    \emptyset \text{ with probability 1 }& \text{ otherwise}.
    \end{cases}
\end{align}
From \eqref{equation:obervation function for agent j in case study}, we know that $s_i^j\in Q\cup\emptyset$. The observation set $S_i$ is then $S_i\subseteq\prod_{j\in\mathcal{N},j\neq i}(Q\cup\emptyset)$.

\subsubsection{Likelihood Functions}
{Given the sensor model, we define the probability to get an observation $s_i^j$ conditioned on agent $i$'s location $q_i$ and the hypothesis $\theta(j)$ by
$$
l_i^j(s_i^j|q_i,\theta(j)) = \sum_{q_j}  P_{i}(s_i^j|q_{i},q_{j})P(q_j|\theta(j),q_i),
$$
where $P(q_j|\theta(j),q_i)$ is the conditional probability of agent $j$ at location $q_j$.}

{We then form the local likelihood function $l_i(s_i|\theta, q_{i})$ by taking the product of the likelihoods for each sensor value $s^j_i$:
$$
    l_i(s_i|\theta, q_{i})  = \prod_{j\in \mathcal{N}, j\neq i} l^j_i(s^j_i|\theta(j),q_{i}).
$$
}

\subsubsection{Enforcing source agent requirements}
\added{For this case study to satisfy the conditions 1 and 2 in Theorem~\ref{theorem:main result}, we use a heuristic method where each agent's local state path $\omega_i \in \Omega_i$ needs to pass within the observable range of each other agent $\mathcal{Q}_j(q_j)$ for some $q_j$ on the path infinitely often. To generate local state paths we synthesize policies that ensure that the agent visits each pair of persistent surveillance task targets (both for the good and bad instances) infinitely often. These local state paths are indefinitely repeated, allowing one to compute which agents belong in the source set $S(\theta,\theta')$ based on the finite periods of these sequences.
While the generalized approach to designing the set of state observation paths $\Omega$ for all agents is outside the scope of this work, we chose to select target pairs such that their local state paths will pass within the observable window of all other agents.
An alternate planning approach that will enforce the source agent requirements for all possible persistent surveillance locations involves formulating the surveillance task as a GR(1) reactive synthesis problem~\cite{bharadwaj2018synthesis}.
}

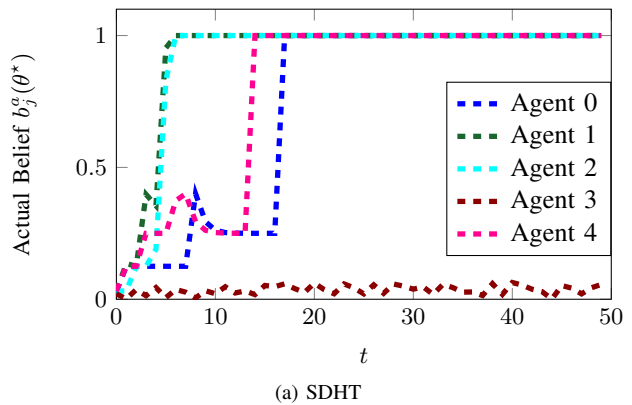
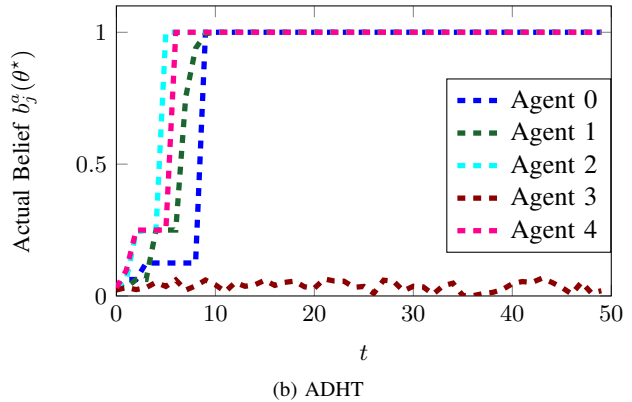
\begin{figure}[t]
    \centering
    \subfloat[SDHT \label{fig:sdht}]{
    \pgfplotstableread[col sep = comma]{pics/plots/Min_sync_Belief.csv}\msB
\begin{tikzpicture}
	\begin{axis}[
	xlabel = {$t$},
	ylabel = {Actual Belief $b^a_j(\theta^\star)$},
	ymin= 0,
	ymax= 1.10,
	xmin=0,
	xmax=50,
	label style={font=\small},
	tick label style={font=\small},
	width = 0.45*\textwidth,
	height = 0.30\textwidth,
	legend style={at={(1,0.75)},anchor=north east}]
	\addplot+[no markers, dashed, line width=2pt,color=blue] table [x index ={0},y index ={1}]{\msB};\label{msb:a0}\addlegendentry{Agent 0}
	\addplot+[no markers, dashed, line width=2pt,color=greenish] table [x index ={0},y index ={2}]{\msB};\label{msb:a1}\addlegendentry{Agent 1} 
	\addplot+[no markers, dashed, line width=2pt,color=cyan] table [x index ={0},y index ={3}]{\msB};\label{msb:a2}\addlegendentry{Agent 2} 
	\addplot+[no markers, dashed, line width=2pt,color=maroon] table [x index ={0},y index ={4}]{\msB};\label{msb:a3}\addlegendentry{Agent 3}
	\addplot+[no markers, dashed, line width=2pt,color=pink] table [x index ={0},y index ={5}]{\msB};\label{msb:a4}\addlegendentry{Agent 4}
	\end{axis}
\end{tikzpicture}
    }\\
    \subfloat[ADHT \label{fig:adht}]{
    \pgfplotstableread[col sep = comma]{pics/plots/Min_async_Belief.csv}\maB
\begin{tikzpicture}
	\begin{axis}[
	xlabel = {$t$},
	ylabel = {Actual Belief $b^a_j(\theta^\star)$},
	ymin= 0,
	ymax= 1.10,
	xmin=0,
	xmax=50,
		label style={font=\small},
	tick label style={font=\small},
	width = 0.45*\textwidth,
	height = 0.30\textwidth,
	legend style={at={(1,0.75)},anchor=north east}]
	\addplot+[no markers, dashed, line width=2pt,color=blue] table [x index ={0},y index ={1}]{\maB};\label{mab:a00}\addlegendentry{Agent 0}
	\addplot+[no markers, dashed, line width=2pt,color=greenish] table [x index ={0},y index ={2}]{\maB};\label{mab:a10}\addlegendentry{Agent 1} 
	\addplot+[no markers, dashed, line width=2pt,color=cyan] table [x index ={0},y index ={3}]{\maB};\label{mab:a20}\addlegendentry{Agent 2} 
	\addplot+[no markers, dashed, line width=2pt,color=maroon] table [x index ={0},y index ={4}]{\maB};\label{mab:a30}\addlegendentry{Agent 3}
	\addplot+[no markers, dashed, line width=2pt,color=pink] table [x index ={0},y index ={5}]{\maB};\label{mab:a40}\addlegendentry{Agent 4}
	\end{axis}
\end{tikzpicture}
    }
    \caption{\added{Each agent's AB $b^a_{j,t}(\theta^*)$ for the true hypothesis $\theta^*=(1,1,1,0,1)$ over time $t$.}
    Agent 3 is the bad agent who shares randomly generated beliefs.
    }
    \label{fig:sdht-adht}
\end{figure}

\begin{figure}
    \centering
    \pgfplotstableread[col sep = comma]{pics/plots/Calls.csv}\CallsB
\begin{tikzpicture}
	\begin{axis}[
	legend pos=north west,
	xlabel = {$t$},
	ylabel = {Cumulative Case One Calls},
	ymin= 0,
	ymax= 4000,
	xmin=0,
	xmax=50,
		label style={font=\small},
	tick label style={font=\small},
	width = 0.45*\textwidth,
	height = 0.32\textwidth]
	\addplot+[no markers, line width=2pt,color=blue] table [x index ={0},y index ={1}]{\CallsB};\label{mab:a0}\addlegendentry{SDHT}
	\addplot+[no markers, line width=2pt,color=red] table [x index ={0},y index ={2}]{\CallsB};\label{mab:a1}\addlegendentry{ADHT} 
	\end{axis}
\end{tikzpicture}
    \caption{\added{Number of times AB is updated with case one.} }
    \label{fig:belief_calls}
\end{figure}
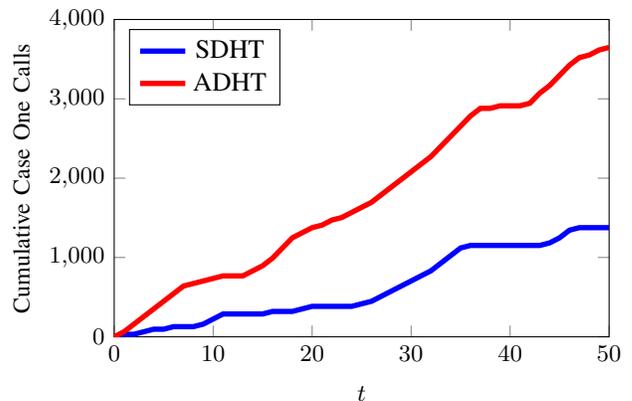 

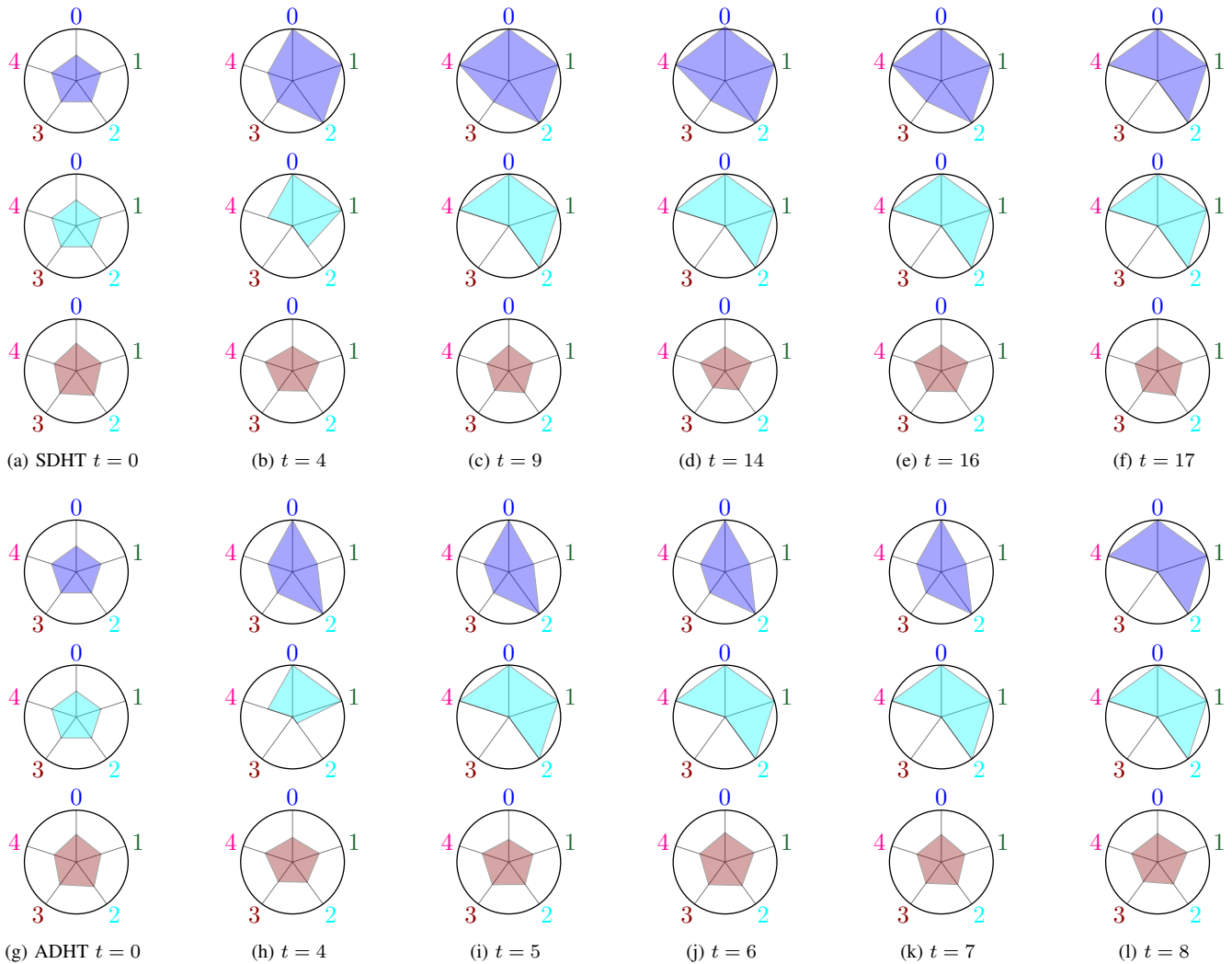
\begin{figure*}[t]
    \subfloat[SDHT $t=0$ \label{fig:ts0}]{ 
    \newcommand{\beliefSize}{0.75}

\begin{tabular}{@{}c@{}}
\begin{tikzpicture}
    \draw[black,line width=0.5pt] (0,0) circle (1*\beliefSize);

    \coordinate (origin) at (0, 0);
    
     \coordinate (1) at (90: 0.5*\beliefSize);
     \coordinate (edge-1) at (90: 1*\beliefSize);
     \node[inner sep=0.7pt,outer xsep=-2pt] (title-1) at (90:1.25*1*\beliefSize) {\textcolor{blue}{$0$}};
     \draw[opacity=0.5] (origin) -- (edge-1);

     \coordinate (2) at (18: 0.5*\beliefSize);
     \coordinate (edge-2) at (18: 1*\beliefSize);
     \node[inner sep=0.7pt,outer xsep=-2pt] (title-2) at (18:1.25*1*\beliefSize) {\textcolor{greenish}{$1$}};
      \draw[opacity=0.5] (origin) -- (edge-2);

     \coordinate (3) at (306: 0.5*\beliefSize);
     \coordinate (edge-3) at (306: 1*\beliefSize);
     \node[inner sep=0.7pt,outer xsep=-2pt] (title-3) at (306:1.25*1*\beliefSize) {\textcolor{cyan}{$2$}};
      \draw[opacity=0.5] (origin) -- (edge-3);

     \coordinate (4) at (234: 0.5*\beliefSize);
     \coordinate (edge-4) at (234: 1*\beliefSize);
     \node[inner sep=0.7pt,outer xsep=-2pt] (title-4) at (234:1.25*1*\beliefSize) {\textcolor{maroon}{$3$}};
      \draw[opacity=0.5] (origin) -- (edge-4);

     \coordinate (5) at (162: 0.5*\beliefSize);
     \coordinate (edge-5) at (162: 1*\beliefSize);
     \node[inner sep=0.7pt,outer xsep=-2pt] (title-5) at (162:1.25*1*\beliefSize) {\textcolor{pink}{$4$}};
      \draw[opacity=0.5] (origin) -- (edge-5);

    \draw [fill=blue, opacity=.35] (1)
                                \foreach \i in {2,...,5}{-- (\i)} --cycle;
\end{tikzpicture}
\\
\begin{tikzpicture}
    \draw[black,line width=0.5pt] (0,0) circle (1*\beliefSize);

    \coordinate (origin) at (0, 0);
    
     \coordinate (1) at (90: 0.5*\beliefSize);
     \coordinate (edge-1) at (90: 1*\beliefSize);
     \node[inner sep=0.7pt,outer xsep=-2pt] (title-1) at (90:1.25*1*\beliefSize) {\textcolor{blue}{$0$}};
     \draw[opacity=0.5] (origin) -- (edge-1);

     \coordinate (2) at (18: 0.5*\beliefSize);
     \coordinate (edge-2) at (18: 1*\beliefSize);
     \node[inner sep=0.7pt,outer xsep=-2pt] (title-2) at (18:1.25*1*\beliefSize) {\textcolor{greenish}{$1$}};
      \draw[opacity=0.5] (origin) -- (edge-2);

     \coordinate (3) at (306: 0.5*\beliefSize);
     \coordinate (edge-3) at (306: 1*\beliefSize);
     \node[inner sep=0.7pt,outer xsep=-2pt] (title-3) at (306:1.25*1*\beliefSize) {\textcolor{cyan}{$2$}};
      \draw[opacity=0.5] (origin) -- (edge-3);

     \coordinate (4) at (234: 0.5*\beliefSize);
     \coordinate (edge-4) at (234: 1*\beliefSize);
     \node[inner sep=0.7pt,outer xsep=-2pt] (title-4) at (234:1.25*1*\beliefSize) {\textcolor{maroon}{$3$}};
      \draw[opacity=0.5] (origin) -- (edge-4);

     \coordinate (5) at (162: 0.5*\beliefSize);
     \coordinate (edge-5) at (162: 1*\beliefSize);
     \node[inner sep=0.7pt,outer xsep=-2pt] (title-5) at (162:1.25*1*\beliefSize) {\textcolor{pink}{$4$}};
      \draw[opacity=0.5] (origin) -- (edge-5);

    \draw [fill=cyan, opacity=.35] (1)
                                \foreach \i in {2,...,5}{-- (\i)} --cycle;
\end{tikzpicture}
\\
\begin{tikzpicture}
    \draw[black,line width=0.5pt] (0,0) circle (1*\beliefSize);

    \coordinate (origin) at (0, 0);
    
     \coordinate (1) at (90: 0.536*\beliefSize);
     \coordinate (edge-1) at (90: 1*\beliefSize);
     \node[inner sep=0.7pt,outer xsep=-2pt] (title-1) at (90:1.25*1*\beliefSize) {\textcolor{blue}{$0$}};
     \draw[opacity=0.5] (origin) -- (edge-1);

     \coordinate (2) at (18: 0.502*\beliefSize);
     \coordinate (edge-2) at (18: 1*\beliefSize);
     \node[inner sep=0.7pt,outer xsep=-2pt] (title-2) at (18:1.25*1*\beliefSize) {\textcolor{greenish}{$1$}};
      \draw[opacity=0.5] (origin) -- (edge-2);

     \coordinate (3) at (306: 0.583*\beliefSize);
     \coordinate (edge-3) at (306: 1*\beliefSize);
     \node[inner sep=0.7pt,outer xsep=-2pt] (title-3) at (306:1.25*1*\beliefSize) {\textcolor{cyan}{$2$}};
      \draw[opacity=0.5] (origin) -- (edge-3);

     \coordinate (4) at (234: 0.537*\beliefSize);
     \coordinate (edge-4) at (234: 1*\beliefSize);
     \node[inner sep=0.7pt,outer xsep=-2pt] (title-4) at (234:1.25*1*\beliefSize) {\textcolor{maroon}{$3$}};
      \draw[opacity=0.5] (origin) -- (edge-4);

     \coordinate (5) at (162: 0.447*\beliefSize);
     \coordinate (edge-5) at (162: 1*\beliefSize);
     \node[inner sep=0.7pt,outer xsep=-2pt] (title-5) at (162:1.25*1*\beliefSize) {\textcolor{pink}{$4$}};
      \draw[opacity=0.5] (origin) -- (edge-5);

    \draw [fill=maroon, opacity=.35] (1)
                                \foreach \i in {2,...,5}{-- (\i)} --cycle;
\end{tikzpicture}

\end{tabular}
    }\hfill
    \subfloat[$t=4$ \label{fig:ts5}]{ 
    \newcommand{\beliefSize}{0.75}

\begin{tabular}{@{}c@{}}
\begin{tikzpicture}
    \draw[black,line width=0.5pt] (0,0) circle (1*\beliefSize);

    \coordinate (origin) at (0, 0);
    
     \coordinate (1) at (90: 1.0*\beliefSize);
     \coordinate (edge-1) at (90: 1*\beliefSize);
     \node[inner sep=0.7pt,outer xsep=-2pt] (title-1) at (90:1.25*1*\beliefSize) {\textcolor{blue}{$0$}};
     \draw[opacity=0.5] (origin) -- (edge-1);

     \coordinate (2) at (18: 1.0*\beliefSize);
     \coordinate (edge-2) at (18: 1*\beliefSize);
     \node[inner sep=0.7pt,outer xsep=-2pt] (title-2) at (18:1.25*1*\beliefSize) {\textcolor{greenish}{$1$}};
      \draw[opacity=0.5] (origin) -- (edge-2);

     \coordinate (3) at (306: 1.0*\beliefSize);
     \coordinate (edge-3) at (306: 1*\beliefSize);
     \node[inner sep=0.7pt,outer xsep=-2pt] (title-3) at (306:1.25*1*\beliefSize) {\textcolor{cyan}{$2$}};
      \draw[opacity=0.5] (origin) -- (edge-3);

     \coordinate (4) at (234: 0.5*\beliefSize);
     \coordinate (edge-4) at (234: 1*\beliefSize);
     \node[inner sep=0.7pt,outer xsep=-2pt] (title-4) at (234:1.25*1*\beliefSize) {\textcolor{maroon}{$3$}};
      \draw[opacity=0.5] (origin) -- (edge-4);

     \coordinate (5) at (162: 0.5*\beliefSize);
     \coordinate (edge-5) at (162: 1*\beliefSize);
     \node[inner sep=0.7pt,outer xsep=-2pt] (title-5) at (162:1.25*1*\beliefSize) {\textcolor{pink}{$4$}};
      \draw[opacity=0.5] (origin) -- (edge-5);

    \draw [fill=blue, opacity=.35] (1)
                                \foreach \i in {2,...,5}{-- (\i)} --cycle;
\end{tikzpicture}
\\
\begin{tikzpicture}
    \draw[black,line width=0.5pt] (0,0) circle (1*\beliefSize);

    \coordinate (origin) at (0, 0);
    
     \coordinate (1) at (90: 1.0*\beliefSize);
     \coordinate (edge-1) at (90: 1*\beliefSize);
     \node[inner sep=0.7pt,outer xsep=-2pt] (title-1) at (90:1.25*1*\beliefSize) {\textcolor{blue}{$0$}};
     \draw[opacity=0.5] (origin) -- (edge-1);

     \coordinate (2) at (18: 1.0*\beliefSize);
     \coordinate (edge-2) at (18: 1*\beliefSize);
     \node[inner sep=0.7pt,outer xsep=-2pt] (title-2) at (18:1.25*1*\beliefSize) {\textcolor{greenish}{$1$}};
      \draw[opacity=0.5] (origin) -- (edge-2);

     \coordinate (3) at (306: 0.5*\beliefSize);
     \coordinate (edge-3) at (306: 1*\beliefSize);
     \node[inner sep=0.7pt,outer xsep=-2pt] (title-3) at (306:1.25*1*\beliefSize) {\textcolor{cyan}{$2$}};
      \draw[opacity=0.5] (origin) -- (edge-3);

     \coordinate (4) at (234: 0.0*\beliefSize);
     \coordinate (edge-4) at (234: 1*\beliefSize);
     \node[inner sep=0.7pt,outer xsep=-2pt] (title-4) at (234:1.25*1*\beliefSize) {\textcolor{maroon}{$3$}};
      \draw[opacity=0.5] (origin) -- (edge-4);

     \coordinate (5) at (162: 0.5*\beliefSize);
     \coordinate (edge-5) at (162: 1*\beliefSize);
     \node[inner sep=0.7pt,outer xsep=-2pt] (title-5) at (162:1.25*1*\beliefSize) {\textcolor{pink}{$4$}};
      \draw[opacity=0.5] (origin) -- (edge-5);

    \draw [fill=cyan, opacity=.35] (1)
                                \foreach \i in {2,...,5}{-- (\i)} --cycle;
\end{tikzpicture}
\\
\begin{tikzpicture}
    \draw[black,line width=0.5pt] (0,0) circle (1*\beliefSize);

    \coordinate (origin) at (0, 0);
    
     \coordinate (1) at (90: 0.473*\beliefSize);
     \coordinate (edge-1) at (90: 1*\beliefSize);
     \node[inner sep=0.7pt,outer xsep=-2pt] (title-1) at (90:1.25*1*\beliefSize) {\textcolor{blue}{$0$}};
     \draw[opacity=0.5] (origin) -- (edge-1);

     \coordinate (2) at (18: 0.540*\beliefSize);
     \coordinate (edge-2) at (18: 1*\beliefSize);
     \node[inner sep=0.7pt,outer xsep=-2pt] (title-2) at (18:1.25*1*\beliefSize) {\textcolor{greenish}{$1$}};
      \draw[opacity=0.5] (origin) -- (edge-2);

     \coordinate (3) at (306: 0.481*\beliefSize);
     \coordinate (edge-3) at (306: 1*\beliefSize);
     \node[inner sep=0.7pt,outer xsep=-2pt] (title-3) at (306:1.25*1*\beliefSize) {\textcolor{cyan}{$2$}};
      \draw[opacity=0.5] (origin) -- (edge-3);

     \coordinate (4) at (234: 0.467*\beliefSize);
     \coordinate (edge-4) at (234: 1*\beliefSize);
     \node[inner sep=0.7pt,outer xsep=-2pt] (title-4) at (234:1.25*1*\beliefSize) {\textcolor{maroon}{$3$}};
      \draw[opacity=0.5] (origin) -- (edge-4);

     \coordinate (5) at (162: 0.556*\beliefSize);
     \coordinate (edge-5) at (162: 1*\beliefSize);
     \node[inner sep=0.7pt,outer xsep=-2pt] (title-5) at (162:1.25*1*\beliefSize) {\textcolor{pink}{$4$}};
      \draw[opacity=0.5] (origin) -- (edge-5);

    \draw [fill=maroon, opacity=.35] (1)
                                \foreach \i in {2,...,5}{-- (\i)} --cycle;
\end{tikzpicture}

\end{tabular}
    }\hfill
    \subfloat[$t=9$ \label{fig:ts9}]{ 
    \newcommand{\beliefSize}{0.75}

\begin{tabular}{@{}c@{}}
\begin{tikzpicture}
    \draw[black,line width=0.5pt] (0,0) circle (1*\beliefSize);

    \coordinate (origin) at (0, 0);
    
     \coordinate (1) at (90: 1.0*\beliefSize);
     \coordinate (edge-1) at (90: 1*\beliefSize);
     \node[inner sep=0.7pt,outer xsep=-2pt] (title-1) at (90:1.25*1*\beliefSize) {\textcolor{blue}{$0$}};
     \draw[opacity=0.5] (origin) -- (edge-1);

     \coordinate (2) at (18: 1.0*\beliefSize);
     \coordinate (edge-2) at (18: 1*\beliefSize);
     \node[inner sep=0.7pt,outer xsep=-2pt] (title-2) at (18:1.25*1*\beliefSize) {\textcolor{greenish}{$1$}};
      \draw[opacity=0.5] (origin) -- (edge-2);

     \coordinate (3) at (306: 1.0*\beliefSize);
     \coordinate (edge-3) at (306: 1*\beliefSize);
     \node[inner sep=0.7pt,outer xsep=-2pt] (title-3) at (306:1.25*1*\beliefSize) {\textcolor{cyan}{$2$}};
      \draw[opacity=0.5] (origin) -- (edge-3);

     \coordinate (4) at (234: 0.5*\beliefSize);
     \coordinate (edge-4) at (234: 1*\beliefSize);
     \node[inner sep=0.7pt,outer xsep=-2pt] (title-4) at (234:1.25*1*\beliefSize) {\textcolor{maroon}{$3$}};
      \draw[opacity=0.5] (origin) -- (edge-4);

     \coordinate (5) at (162: 1.0*\beliefSize);
     \coordinate (edge-5) at (162: 1*\beliefSize);
     \node[inner sep=0.7pt,outer xsep=-2pt] (title-5) at (162:1.25*1*\beliefSize) {\textcolor{pink}{$4$}};
      \draw[opacity=0.5] (origin) -- (edge-5);

    \draw [fill=blue, opacity=.35] (1)
                                \foreach \i in {2,...,5}{-- (\i)} --cycle;
\end{tikzpicture}
\\
\begin{tikzpicture}
    \draw[black,line width=0.5pt] (0,0) circle (1*\beliefSize);

    \coordinate (origin) at (0, 0);
    
     \coordinate (1) at (90: 1.0*\beliefSize);
     \coordinate (edge-1) at (90: 1*\beliefSize);
     \node[inner sep=0.7pt,outer xsep=-2pt] (title-1) at (90:1.25*1*\beliefSize) {\textcolor{blue}{$0$}};
     \draw[opacity=0.5] (origin) -- (edge-1);

     \coordinate (2) at (18: 1.0*\beliefSize);
     \coordinate (edge-2) at (18: 1*\beliefSize);
     \node[inner sep=0.7pt,outer xsep=-2pt] (title-2) at (18:1.25*1*\beliefSize) {\textcolor{greenish}{$1$}};
      \draw[opacity=0.5] (origin) -- (edge-2);

     \coordinate (3) at (306: 1.0*\beliefSize);
     \coordinate (edge-3) at (306: 1*\beliefSize);
     \node[inner sep=0.7pt,outer xsep=-2pt] (title-3) at (306:1.25*1*\beliefSize) {\textcolor{cyan}{$2$}};
      \draw[opacity=0.5] (origin) -- (edge-3);

     \coordinate (4) at (234: 0.0*\beliefSize);
     \coordinate (edge-4) at (234: 1*\beliefSize);
     \node[inner sep=0.7pt,outer xsep=-2pt] (title-4) at (234:1.25*1*\beliefSize) {\textcolor{maroon}{$3$}};
      \draw[opacity=0.5] (origin) -- (edge-4);

     \coordinate (5) at (162: 1.0*\beliefSize);
     \coordinate (edge-5) at (162: 1*\beliefSize);
     \node[inner sep=0.7pt,outer xsep=-2pt] (title-5) at (162:1.25*1*\beliefSize) {\textcolor{pink}{$4$}};
      \draw[opacity=0.5] (origin) -- (edge-5);

    \draw [fill=cyan, opacity=.35] (1)
                                \foreach \i in {2,...,5}{-- (\i)} --cycle;
\end{tikzpicture}
\\
\begin{tikzpicture}
    \draw[black,line width=0.5pt] (0,0) circle (1*\beliefSize);

    \coordinate (origin) at (0, 0);
    
     \coordinate (1) at (90: 0.498*\beliefSize);
     \coordinate (edge-1) at (90: 1*\beliefSize);
     \node[inner sep=0.7pt,outer xsep=-2pt] (title-1) at (90:1.25*1*\beliefSize) {\textcolor{blue}{$0$}};
     \draw[opacity=0.5] (origin) -- (edge-1);

     \coordinate (2) at (18: 0.486*\beliefSize);
     \coordinate (edge-2) at (18: 1*\beliefSize);
     \node[inner sep=0.7pt,outer xsep=-2pt] (title-2) at (18:1.25*1*\beliefSize) {\textcolor{greenish}{$1$}};
      \draw[opacity=0.5] (origin) -- (edge-2);

     \coordinate (3) at (306: 0.518*\beliefSize);
     \coordinate (edge-3) at (306: 1*\beliefSize);
     \node[inner sep=0.7pt,outer xsep=-2pt] (title-3) at (306:1.25*1*\beliefSize) {\textcolor{cyan}{$2$}};
      \draw[opacity=0.5] (origin) -- (edge-3);

     \coordinate (4) at (234: 0.464*\beliefSize);
     \coordinate (edge-4) at (234: 1*\beliefSize);
     \node[inner sep=0.7pt,outer xsep=-2pt] (title-4) at (234:1.25*1*\beliefSize) {\textcolor{maroon}{$3$}};
      \draw[opacity=0.5] (origin) -- (edge-4);

     \coordinate (5) at (162: 0.444*\beliefSize);
     \coordinate (edge-5) at (162: 1*\beliefSize);
     \node[inner sep=0.7pt,outer xsep=-2pt] (title-5) at (162:1.25*1*\beliefSize) {\textcolor{pink}{$4$}};
      \draw[opacity=0.5] (origin) -- (edge-5);

    \draw [fill=maroon, opacity=.35] (1)
                                \foreach \i in {2,...,5}{-- (\i)} --cycle;
\end{tikzpicture}

\end{tabular}
    }\hfill
    \subfloat[$t=14$ \label{fig:ts14}]{
    \newcommand{\beliefSize}{0.75}

\begin{tabular}{@{}c@{}}
\begin{tikzpicture}
    \draw[black,line width=0.5pt] (0,0) circle (1*\beliefSize);

    \coordinate (origin) at (0, 0);
    
     \coordinate (1) at (90: 1.05*\beliefSize);
     \coordinate (edge-1) at (90: 1*\beliefSize);
     \node[inner sep=0.7pt,outer xsep=-2pt] (title-1) at (90:1.25*1*\beliefSize) {\textcolor{blue}{$0$}};
     \draw[opacity=0.5] (origin) -- (edge-1);

     \coordinate (2) at (18: 1.0*\beliefSize);
     \coordinate (edge-2) at (18: 1*\beliefSize);
     \node[inner sep=0.7pt,outer xsep=-2pt] (title-2) at (18:1.25*1*\beliefSize) {\textcolor{greenish}{$1$}};
      \draw[opacity=0.5] (origin) -- (edge-2);

     \coordinate (3) at (306: 1.0*\beliefSize);
     \coordinate (edge-3) at (306: 1*\beliefSize);
     \node[inner sep=0.7pt,outer xsep=-2pt] (title-3) at (306:1.25*1*\beliefSize) {\textcolor{cyan}{$2$}};
      \draw[opacity=0.5] (origin) -- (edge-3);

     \coordinate (4) at (234: 0.471*\beliefSize);
     \coordinate (edge-4) at (234: 1*\beliefSize);
     \node[inner sep=0.7pt,outer xsep=-2pt] (title-4) at (234:1.25*1*\beliefSize) {\textcolor{maroon}{$3$}};
      \draw[opacity=0.5] (origin) -- (edge-4);

     \coordinate (5) at (162: 1.0*\beliefSize);
     \coordinate (edge-5) at (162: 1*\beliefSize);
     \node[inner sep=0.7pt,outer xsep=-2pt] (title-5) at (162:1.25*1*\beliefSize) {\textcolor{pink}{$4$}};
      \draw[opacity=0.5] (origin) -- (edge-5);

    \draw [fill=blue, opacity=.35] (1)
                                \foreach \i in {2,...,5}{-- (\i)} --cycle;
\end{tikzpicture}
\\
\begin{tikzpicture}
    \draw[black,line width=0.5pt] (0,0) circle (1*\beliefSize);

    \coordinate (origin) at (0, 0);
    
     \coordinate (1) at (90: 1.0*\beliefSize);
     \coordinate (edge-1) at (90: 1*\beliefSize);
     \node[inner sep=0.7pt,outer xsep=-2pt] (title-1) at (90:1.25*1*\beliefSize) {\textcolor{blue}{$0$}};
     \draw[opacity=0.5] (origin) -- (edge-1);

     \coordinate (2) at (18: 1.0*\beliefSize);
     \coordinate (edge-2) at (18: 1*\beliefSize);
     \node[inner sep=0.7pt,outer xsep=-2pt] (title-2) at (18:1.25*1*\beliefSize) {\textcolor{greenish}{$1$}};
      \draw[opacity=0.5] (origin) -- (edge-2);

     \coordinate (3) at (306: 1.0*\beliefSize);
     \coordinate (edge-3) at (306: 1*\beliefSize);
     \node[inner sep=0.7pt,outer xsep=-2pt] (title-3) at (306:1.25*1*\beliefSize) {\textcolor{cyan}{$2$}};
      \draw[opacity=0.5] (origin) -- (edge-3);

     \coordinate (4) at (234: 0.0*\beliefSize);
     \coordinate (edge-4) at (234: 1*\beliefSize);
     \node[inner sep=0.7pt,outer xsep=-2pt] (title-4) at (234:1.25*1*\beliefSize) {\textcolor{maroon}{$3$}};
      \draw[opacity=0.5] (origin) -- (edge-4);

     \coordinate (5) at (162: 1.0*\beliefSize);
     \coordinate (edge-5) at (162: 1*\beliefSize);
     \node[inner sep=0.7pt,outer xsep=-2pt] (title-5) at (162:1.25*1*\beliefSize) {\textcolor{pink}{$4$}};
      \draw[opacity=0.5] (origin) -- (edge-5);

    \draw [fill=cyan, opacity=.35] (1)
                                \foreach \i in {2,...,5}{-- (\i)} --cycle;
\end{tikzpicture}
\\
\begin{tikzpicture}
    \draw[black,line width=0.5pt] (0,0) circle (1*\beliefSize);

    \coordinate (origin) at (0, 0);
    
     \coordinate (1) at (90: 0.468*\beliefSize);
     \coordinate (edge-1) at (90: 1*\beliefSize);
     \node[inner sep=0.7pt,outer xsep=-2pt] (title-1) at (90:1.25*1*\beliefSize) {\textcolor{blue}{$0$}};
     \draw[opacity=0.5] (origin) -- (edge-1);

     \coordinate (2) at (18: 0.537*\beliefSize);
     \coordinate (edge-2) at (18: 1*\beliefSize);
     \node[inner sep=0.7pt,outer xsep=-2pt] (title-2) at (18:1.25*1*\beliefSize) {\textcolor{greenish}{$1$}};
      \draw[opacity=0.5] (origin) -- (edge-2);

     \coordinate (3) at (306: 0.454*\beliefSize);
     \coordinate (edge-3) at (306: 1*\beliefSize);
     \node[inner sep=0.7pt,outer xsep=-2pt] (title-3) at (306:1.25*1*\beliefSize) {\textcolor{cyan}{$2$}};
      \draw[opacity=0.5] (origin) -- (edge-3);

     \coordinate (4) at (234: 0.398*\beliefSize);
     \coordinate (edge-4) at (234: 1*\beliefSize);
     \node[inner sep=0.7pt,outer xsep=-2pt] (title-4) at (234:1.25*1*\beliefSize) {\textcolor{maroon}{$3$}};
      \draw[opacity=0.5] (origin) -- (edge-4);

     \coordinate (5) at (162: 0.503*\beliefSize);
     \coordinate (edge-5) at (162: 1*\beliefSize);
     \node[inner sep=0.7pt,outer xsep=-2pt] (title-5) at (162:1.25*1*\beliefSize) {\textcolor{pink}{$4$}};
      \draw[opacity=0.5] (origin) -- (edge-5);

    \draw [fill=maroon, opacity=.35] (1)
                                \foreach \i in {2,...,5}{-- (\i)} --cycle;
\end{tikzpicture}

\end{tabular}
    }\hfill
    \subfloat[$t=16$ \label{fig:ts16}]{
    \newcommand{\beliefSize}{0.75}

\begin{tabular}{@{}c@{}}
\begin{tikzpicture}
    \draw[black,line width=0.5pt] (0,0) circle (1*\beliefSize);

    \coordinate (origin) at (0, 0);
    
     \coordinate (1) at (90: 1.0*\beliefSize);
     \coordinate (edge-1) at (90: 1*\beliefSize);
     \node[inner sep=0.7pt,outer xsep=-2pt] (title-1) at (90:1.25*1*\beliefSize) {\textcolor{blue}{$0$}};
     \draw[opacity=0.5] (origin) -- (edge-1);

     \coordinate (2) at (18: 1.0*\beliefSize);
     \coordinate (edge-2) at (18: 1*\beliefSize);
     \node[inner sep=0.7pt,outer xsep=-2pt] (title-2) at (18:1.25*1*\beliefSize) {\textcolor{greenish}{$1$}};
      \draw[opacity=0.5] (origin) -- (edge-2);

     \coordinate (3) at (306: 1.0*\beliefSize);
     \coordinate (edge-3) at (306: 1*\beliefSize);
     \node[inner sep=0.7pt,outer xsep=-2pt] (title-3) at (306:1.25*1*\beliefSize) {\textcolor{cyan}{$2$}};
      \draw[opacity=0.5] (origin) -- (edge-3);

     \coordinate (4) at (234: 0.492*\beliefSize);
     \coordinate (edge-4) at (234: 1*\beliefSize);
     \node[inner sep=0.7pt,outer xsep=-2pt] (title-4) at (234:1.25*1*\beliefSize) {\textcolor{maroon}{$3$}};
      \draw[opacity=0.5] (origin) -- (edge-4);

     \coordinate (5) at (162: 1.0*\beliefSize);
     \coordinate (edge-5) at (162: 1*\beliefSize);
     \node[inner sep=0.7pt,outer xsep=-2pt] (title-5) at (162:1.25*1*\beliefSize) {\textcolor{pink}{$4$}};
      \draw[opacity=0.5] (origin) -- (edge-5);

    \draw [fill=blue, opacity=.35] (1)
                                \foreach \i in {2,...,5}{-- (\i)} --cycle;
\end{tikzpicture}
\\
\begin{tikzpicture}
    \draw[black,line width=0.5pt] (0,0) circle (1*\beliefSize);

    \coordinate (origin) at (0, 0);
    
     \coordinate (1) at (90: 1.0*\beliefSize);
     \coordinate (edge-1) at (90: 1*\beliefSize);
     \node[inner sep=0.7pt,outer xsep=-2pt] (title-1) at (90:1.25*1*\beliefSize) {\textcolor{blue}{$0$}};
     \draw[opacity=0.5] (origin) -- (edge-1);

     \coordinate (2) at (18: 1.0*\beliefSize);
     \coordinate (edge-2) at (18: 1*\beliefSize);
     \node[inner sep=0.7pt,outer xsep=-2pt] (title-2) at (18:1.25*1*\beliefSize) {\textcolor{greenish}{$1$}};
      \draw[opacity=0.5] (origin) -- (edge-2);

     \coordinate (3) at (306: 1.0*\beliefSize);
     \coordinate (edge-3) at (306: 1*\beliefSize);
     \node[inner sep=0.7pt,outer xsep=-2pt] (title-3) at (306:1.25*1*\beliefSize) {\textcolor{cyan}{$2$}};
      \draw[opacity=0.5] (origin) -- (edge-3);

     \coordinate (4) at (234: 0.0*\beliefSize);
     \coordinate (edge-4) at (234: 1*\beliefSize);
     \node[inner sep=0.7pt,outer xsep=-2pt] (title-4) at (234:1.25*1*\beliefSize) {\textcolor{maroon}{$3$}};
      \draw[opacity=0.5] (origin) -- (edge-4);

     \coordinate (5) at (162: 1.0*\beliefSize);
     \coordinate (edge-5) at (162: 1*\beliefSize);
     \node[inner sep=0.7pt,outer xsep=-2pt] (title-5) at (162:1.25*1*\beliefSize) {\textcolor{pink}{$4$}};
      \draw[opacity=0.5] (origin) -- (edge-5);

    \draw [fill=cyan, opacity=.35] (1)
                                \foreach \i in {2,...,5}{-- (\i)} --cycle;
\end{tikzpicture}
\\
\begin{tikzpicture}
    \draw[black,line width=0.5pt] (0,0) circle (1*\beliefSize);

    \coordinate (origin) at (0, 0);
    
     \coordinate (1) at (90: 0.502*\beliefSize);
     \coordinate (edge-1) at (90: 1*\beliefSize);
     \node[inner sep=0.7pt,outer xsep=-2pt] (title-1) at (90:1.25*1*\beliefSize) {\textcolor{blue}{$0$}};
     \draw[opacity=0.5] (origin) -- (edge-1);

     \coordinate (2) at (18: 0.539*\beliefSize);
     \coordinate (edge-2) at (18: 1*\beliefSize);
     \node[inner sep=0.7pt,outer xsep=-2pt] (title-2) at (18:1.25*1*\beliefSize) {\textcolor{greenish}{$1$}};
      \draw[opacity=0.5] (origin) -- (edge-2);

     \coordinate (3) at (306: 0.493*\beliefSize);
     \coordinate (edge-3) at (306: 1*\beliefSize);
     \node[inner sep=0.7pt,outer xsep=-2pt] (title-3) at (306:1.25*1*\beliefSize) {\textcolor{cyan}{$2$}};
      \draw[opacity=0.5] (origin) -- (edge-3);

     \coordinate (4) at (234: 0.480*\beliefSize);
     \coordinate (edge-4) at (234: 1*\beliefSize);
     \node[inner sep=0.7pt,outer xsep=-2pt] (title-4) at (234:1.25*1*\beliefSize) {\textcolor{maroon}{$3$}};
      \draw[opacity=0.5] (origin) -- (edge-4);

     \coordinate (5) at (162: 0.547*\beliefSize);
     \coordinate (edge-5) at (162: 1*\beliefSize);
     \node[inner sep=0.7pt,outer xsep=-2pt] (title-5) at (162:1.25*1*\beliefSize) {\textcolor{pink}{$4$}};
      \draw[opacity=0.5] (origin) -- (edge-5);

    \draw [fill=maroon, opacity=.35] (1)
                                \foreach \i in {2,...,5}{-- (\i)} --cycle;
\end{tikzpicture}

\end{tabular}
    }\hfill
    \subfloat[$t=17$ \label{fig:ts17}]{ 
    \newcommand{\beliefSize}{0.75}

\begin{tabular}{@{}c@{}}
\begin{tikzpicture}
    \draw[black,line width=0.5pt] (0,0) circle (1*\beliefSize);

    \coordinate (origin) at (0, 0);
    
     \coordinate (1) at (90: 1.0*\beliefSize);
     \coordinate (edge-1) at (90: 1*\beliefSize);
     \node[inner sep=0.7pt,outer xsep=-2pt] (title-1) at (90:1.25*1*\beliefSize) {\textcolor{blue}{$0$}};
     \draw[opacity=0.5] (origin) -- (edge-1);

     \coordinate (2) at (18: 1.0*\beliefSize);
     \coordinate (edge-2) at (18: 1*\beliefSize);
     \node[inner sep=0.7pt,outer xsep=-2pt] (title-2) at (18:1.25*1*\beliefSize) {\textcolor{greenish}{$1$}};
      \draw[opacity=0.5] (origin) -- (edge-2);

     \coordinate (3) at (306: 1.0*\beliefSize);
     \coordinate (edge-3) at (306: 1*\beliefSize);
     \node[inner sep=0.7pt,outer xsep=-2pt] (title-3) at (306:1.25*1*\beliefSize) {\textcolor{cyan}{$2$}};
      \draw[opacity=0.5] (origin) -- (edge-3);

     \coordinate (4) at (234: 0.0*\beliefSize);
     \coordinate (edge-4) at (234: 1*\beliefSize);
     \node[inner sep=0.7pt,outer xsep=-2pt] (title-4) at (234:1.25*1*\beliefSize) {\textcolor{maroon}{$3$}};
      \draw[opacity=0.5] (origin) -- (edge-4);

     \coordinate (5) at (162: 1.0*\beliefSize);
     \coordinate (edge-5) at (162: 1*\beliefSize);
     \node[inner sep=0.7pt,outer xsep=-2pt] (title-5) at (162:1.25*1*\beliefSize) {\textcolor{pink}{$4$}};
      \draw[opacity=0.5] (origin) -- (edge-5);

    \draw [fill=blue, opacity=.35] (1)
                                \foreach \i in {2,...,5}{-- (\i)} --cycle;
\end{tikzpicture}
\\
\begin{tikzpicture}
    \draw[black,line width=0.5pt] (0,0) circle (1*\beliefSize);

    \coordinate (origin) at (0, 0);
    
     \coordinate (1) at (90: 1.0*\beliefSize);
     \coordinate (edge-1) at (90: 1*\beliefSize);
     \node[inner sep=0.7pt,outer xsep=-2pt] (title-1) at (90:1.25*1*\beliefSize) {\textcolor{blue}{$0$}};
     \draw[opacity=0.5] (origin) -- (edge-1);

     \coordinate (2) at (18: 1.0*\beliefSize);
     \coordinate (edge-2) at (18: 1*\beliefSize);
     \node[inner sep=0.7pt,outer xsep=-2pt] (title-2) at (18:1.25*1*\beliefSize) {\textcolor{greenish}{$1$}};
      \draw[opacity=0.5] (origin) -- (edge-2);

     \coordinate (3) at (306: 1.0*\beliefSize);
     \coordinate (edge-3) at (306: 1*\beliefSize);
     \node[inner sep=0.7pt,outer xsep=-2pt] (title-3) at (306:1.25*1*\beliefSize) {\textcolor{cyan}{$2$}};
      \draw[opacity=0.5] (origin) -- (edge-3);

     \coordinate (4) at (234: 0.0*\beliefSize);
     \coordinate (edge-4) at (234: 1*\beliefSize);
     \node[inner sep=0.7pt,outer xsep=-2pt] (title-4) at (234:1.25*1*\beliefSize) {\textcolor{maroon}{$3$}};
      \draw[opacity=0.5] (origin) -- (edge-4);

     \coordinate (5) at (162: 1.0*\beliefSize);
     \coordinate (edge-5) at (162: 1*\beliefSize);
     \node[inner sep=0.7pt,outer xsep=-2pt] (title-5) at (162:1.25*1*\beliefSize) {\textcolor{pink}{$4$}};
      \draw[opacity=0.5] (origin) -- (edge-5);

    \draw [fill=cyan, opacity=.35] (1)
                                \foreach \i in {2,...,5}{-- (\i)} --cycle;
\end{tikzpicture}
\\
\begin{tikzpicture}
    \draw[black,line width=0.5pt] (0,0) circle (1*\beliefSize);

    \coordinate (origin) at (0, 0);
    
     \coordinate (1) at (90: 0.465*\beliefSize);
     \coordinate (edge-1) at (90: 1*\beliefSize);
     \node[inner sep=0.7pt,outer xsep=-2pt] (title-1) at (90:1.25*1*\beliefSize) {\textcolor{blue}{$0$}};
     \draw[opacity=0.5] (origin) -- (edge-1);

     \coordinate (2) at (18: 0.499*\beliefSize);
     \coordinate (edge-2) at (18: 1*\beliefSize);
     \node[inner sep=0.7pt,outer xsep=-2pt] (title-2) at (18:1.25*1*\beliefSize) {\textcolor{greenish}{$1$}};
      \draw[opacity=0.5] (origin) -- (edge-2);

     \coordinate (3) at (306: 0.587*\beliefSize);
     \coordinate (edge-3) at (306: 1*\beliefSize);
     \node[inner sep=0.7pt,outer xsep=-2pt] (title-3) at (306:1.25*1*\beliefSize) {\textcolor{cyan}{$2$}};
      \draw[opacity=0.5] (origin) -- (edge-3);

     \coordinate (4) at (234: 0.480*\beliefSize);
     \coordinate (edge-4) at (234: 1*\beliefSize);
     \node[inner sep=0.7pt,outer xsep=-2pt] (title-4) at (234:1.25*1*\beliefSize) {\textcolor{maroon}{$3$}};
      \draw[opacity=0.5] (origin) -- (edge-4);

     \coordinate (5) at (162: 0.456*\beliefSize);
     \coordinate (edge-5) at (162: 1*\beliefSize);
     \node[inner sep=0.7pt,outer xsep=-2pt] (title-5) at (162:1.25*1*\beliefSize) {\textcolor{pink}{$4$}};
      \draw[opacity=0.5] (origin) -- (edge-5);

    \draw [fill=maroon, opacity=.35] (1)
                                \foreach \i in {2,...,5}{-- (\i)} --cycle;
\end{tikzpicture}

\end{tabular}
    }\\
    \subfloat[ADHT $t=0$ \label{fig:ta0}]{ 
    \newcommand{\beliefSize}{0.75}

\begin{tabular}{@{}c@{}}
\begin{tikzpicture}
    \draw[black,line width=0.5pt] (0,0) circle (1*\beliefSize);

    \coordinate (origin) at (0, 0);
    
     \coordinate (1) at (90: 0.5*\beliefSize);
     \coordinate (edge-1) at (90: 1*\beliefSize);
     \node[inner sep=0.7pt,outer xsep=-2pt] (title-1) at (90:1.25*1*\beliefSize) {\textcolor{blue}{$0$}};
     \draw[opacity=0.5] (origin) -- (edge-1);

     \coordinate (2) at (18: 0.5*\beliefSize);
     \coordinate (edge-2) at (18: 1*\beliefSize);
     \node[inner sep=0.7pt,outer xsep=-2pt] (title-2) at (18:1.25*1*\beliefSize) {\textcolor{greenish}{$1$}};
      \draw[opacity=0.5] (origin) -- (edge-2);

     \coordinate (3) at (306: 0.5*\beliefSize);
     \coordinate (edge-3) at (306: 1*\beliefSize);
     \node[inner sep=0.7pt,outer xsep=-2pt] (title-3) at (306:1.25*1*\beliefSize) {\textcolor{cyan}{$2$}};
      \draw[opacity=0.5] (origin) -- (edge-3);

     \coordinate (4) at (234: 0.5*\beliefSize);
     \coordinate (edge-4) at (234: 1*\beliefSize);
     \node[inner sep=0.7pt,outer xsep=-2pt] (title-4) at (234:1.25*1*\beliefSize) {\textcolor{maroon}{$3$}};
      \draw[opacity=0.5] (origin) -- (edge-4);

     \coordinate (5) at (162: 0.5*\beliefSize);
     \coordinate (edge-5) at (162: 1*\beliefSize);
     \node[inner sep=0.7pt,outer xsep=-2pt] (title-5) at (162:1.25*1*\beliefSize) {\textcolor{pink}{$4$}};
      \draw[opacity=0.5] (origin) -- (edge-5);

    \draw [fill=blue, opacity=.35] (1)
                                \foreach \i in {2,...,5}{-- (\i)} --cycle;
\end{tikzpicture}
\\
\begin{tikzpicture}
    \draw[black,line width=0.5pt] (0,0) circle (1*\beliefSize);

    \coordinate (origin) at (0, 0);
    
     \coordinate (1) at (90: 0.5*\beliefSize);
     \coordinate (edge-1) at (90: 1*\beliefSize);
     \node[inner sep=0.7pt,outer xsep=-2pt] (title-1) at (90:1.25*1*\beliefSize) {\textcolor{blue}{$0$}};
     \draw[opacity=0.5] (origin) -- (edge-1);

     \coordinate (2) at (18: 0.5*\beliefSize);
     \coordinate (edge-2) at (18: 1*\beliefSize);
     \node[inner sep=0.7pt,outer xsep=-2pt] (title-2) at (18:1.25*1*\beliefSize) {\textcolor{greenish}{$1$}};
      \draw[opacity=0.5] (origin) -- (edge-2);

     \coordinate (3) at (306: 0.5*\beliefSize);
     \coordinate (edge-3) at (306: 1*\beliefSize);
     \node[inner sep=0.7pt,outer xsep=-2pt] (title-3) at (306:1.25*1*\beliefSize) {\textcolor{cyan}{$2$}};
      \draw[opacity=0.5] (origin) -- (edge-3);

     \coordinate (4) at (234: 0.5*\beliefSize);
     \coordinate (edge-4) at (234: 1*\beliefSize);
     \node[inner sep=0.7pt,outer xsep=-2pt] (title-4) at (234:1.25*1*\beliefSize) {\textcolor{maroon}{$3$}};
      \draw[opacity=0.5] (origin) -- (edge-4);

     \coordinate (5) at (162: 0.5*\beliefSize);
     \coordinate (edge-5) at (162: 1*\beliefSize);
     \node[inner sep=0.7pt,outer xsep=-2pt] (title-5) at (162:1.25*1*\beliefSize) {\textcolor{pink}{$4$}};
      \draw[opacity=0.5] (origin) -- (edge-5);

    \draw [fill=cyan, opacity=.35] (1)
                                \foreach \i in {2,...,5}{-- (\i)} --cycle;
\end{tikzpicture}
\\
\begin{tikzpicture}
    \draw[black,line width=0.5pt] (0,0) circle (1*\beliefSize);

    \coordinate (origin) at (0, 0);
    
     \coordinate (1) at (90: 0.536*\beliefSize);
     \coordinate (edge-1) at (90: 1*\beliefSize);
     \node[inner sep=0.7pt,outer xsep=-2pt] (title-1) at (90:1.25*1*\beliefSize) {\textcolor{blue}{$0$}};
     \draw[opacity=0.5] (origin) -- (edge-1);

     \coordinate (2) at (18: 0.502*\beliefSize);
     \coordinate (edge-2) at (18: 1*\beliefSize);
     \node[inner sep=0.7pt,outer xsep=-2pt] (title-2) at (18:1.25*1*\beliefSize) {\textcolor{greenish}{$1$}};
      \draw[opacity=0.5] (origin) -- (edge-2);

     \coordinate (3) at (306: 0.583*\beliefSize);
     \coordinate (edge-3) at (306: 1*\beliefSize);
     \node[inner sep=0.7pt,outer xsep=-2pt] (title-3) at (306:1.25*1*\beliefSize) {\textcolor{cyan}{$2$}};
      \draw[opacity=0.5] (origin) -- (edge-3);

     \coordinate (4) at (234: 0.537*\beliefSize);
     \coordinate (edge-4) at (234: 1*\beliefSize);
     \node[inner sep=0.7pt,outer xsep=-2pt] (title-4) at (234:1.25*1*\beliefSize) {\textcolor{maroon}{$3$}};
      \draw[opacity=0.5] (origin) -- (edge-4);

     \coordinate (5) at (162: 0.447*\beliefSize);
     \coordinate (edge-5) at (162: 1*\beliefSize);
     \node[inner sep=0.7pt,outer xsep=-2pt] (title-5) at (162:1.25*1*\beliefSize) {\textcolor{pink}{$4$}};
      \draw[opacity=0.5] (origin) -- (edge-5);

    \draw [fill=maroon, opacity=.35] (1)
                                \foreach \i in {2,...,5}{-- (\i)} --cycle;
\end{tikzpicture}

\end{tabular}
    }\hfill
    \subfloat[$t=4$ \label{fig:ta4}]{
    \newcommand{\beliefSize}{0.75}

\begin{tabular}{@{}c@{}}
\begin{tikzpicture}
    \draw[black,line width=0.5pt] (0,0) circle (1*\beliefSize);

    \coordinate (origin) at (0, 0);
    
     \coordinate (1) at (90: 1.0*\beliefSize);
     \coordinate (edge-1) at (90: 1*\beliefSize);
     \node[inner sep=0.7pt,outer xsep=-2pt] (title-1) at (90:1.25*1*\beliefSize) {\textcolor{blue}{$0$}};
     \draw[opacity=0.5] (origin) -- (edge-1);

     \coordinate (2) at (18: 0.5*\beliefSize);
     \coordinate (edge-2) at (18: 1*\beliefSize);
     \node[inner sep=0.7pt,outer xsep=-2pt] (title-2) at (18:1.25*1*\beliefSize) {\textcolor{greenish}{$1$}};
      \draw[opacity=0.5] (origin) -- (edge-2);

     \coordinate (3) at (306: 1.0*\beliefSize);
     \coordinate (edge-3) at (306: 1*\beliefSize);
     \node[inner sep=0.7pt,outer xsep=-2pt] (title-3) at (306:1.25*1*\beliefSize) {\textcolor{cyan}{$2$}};
      \draw[opacity=0.5] (origin) -- (edge-3);

     \coordinate (4) at (234: 0.5*\beliefSize);
     \coordinate (edge-4) at (234: 1*\beliefSize);
     \node[inner sep=0.7pt,outer xsep=-2pt] (title-4) at (234:1.25*1*\beliefSize) {\textcolor{maroon}{$3$}};
      \draw[opacity=0.5] (origin) -- (edge-4);

     \coordinate (5) at (162: 0.5*\beliefSize);
     \coordinate (edge-5) at (162: 1*\beliefSize);
     \node[inner sep=0.7pt,outer xsep=-2pt] (title-5) at (162:1.25*1*\beliefSize) {\textcolor{pink}{$4$}};
      \draw[opacity=0.5] (origin) -- (edge-5);

    \draw [fill=blue, opacity=.35] (1)
                                \foreach \i in {2,...,5}{-- (\i)} --cycle;
\end{tikzpicture}
\\
\begin{tikzpicture}
    \draw[black,line width=0.5pt] (0,0) circle (1*\beliefSize);

    \coordinate (origin) at (0, 0);
    
     \coordinate (1) at (90: 1.0*\beliefSize);
     \coordinate (edge-1) at (90: 1*\beliefSize);
     \node[inner sep=0.7pt,outer xsep=-2pt] (title-1) at (90:1.25*1*\beliefSize) {\textcolor{blue}{$0$}};
     \draw[opacity=0.5] (origin) -- (edge-1);

     \coordinate (2) at (18: 1.0*\beliefSize);
     \coordinate (edge-2) at (18: 1*\beliefSize);
     \node[inner sep=0.7pt,outer xsep=-2pt] (title-2) at (18:1.25*1*\beliefSize) {\textcolor{greenish}{$1$}};
      \draw[opacity=0.5] (origin) -- (edge-2);

     \coordinate (3) at (306: 0.15*\beliefSize);
     \coordinate (edge-3) at (306: 1*\beliefSize);
     \node[inner sep=0.7pt,outer xsep=-2pt] (title-3) at (306:1.25*1*\beliefSize) {\textcolor{cyan}{$2$}};
      \draw[opacity=0.5] (origin) -- (edge-3);

     \coordinate (4) at (234: 0.0*\beliefSize);
     \coordinate (edge-4) at (234: 1*\beliefSize);
     \node[inner sep=0.7pt,outer xsep=-2pt] (title-4) at (234:1.25*1*\beliefSize) {\textcolor{maroon}{$3$}};
      \draw[opacity=0.5] (origin) -- (edge-4);

     \coordinate (5) at (162: 0.5*\beliefSize);
     \coordinate (edge-5) at (162: 1*\beliefSize);
     \node[inner sep=0.7pt,outer xsep=-2pt] (title-5) at (162:1.25*1*\beliefSize) {\textcolor{pink}{$4$}};
      \draw[opacity=0.5] (origin) -- (edge-5);

    \draw [fill=cyan, opacity=.35] (1)
                                \foreach \i in {2,...,5}{-- (\i)} --cycle;
\end{tikzpicture}
\\
\begin{tikzpicture}
    \draw[black,line width=0.5pt] (0,0) circle (1*\beliefSize);

    \coordinate (origin) at (0, 0);
    
     \coordinate (1) at (90: 0.473*\beliefSize);
     \coordinate (edge-1) at (90: 1*\beliefSize);
     \node[inner sep=0.7pt,outer xsep=-2pt] (title-1) at (90:1.25*1*\beliefSize) {\textcolor{blue}{$0$}};
     \draw[opacity=0.5] (origin) -- (edge-1);

     \coordinate (2) at (18: 0.541*\beliefSize);
     \coordinate (edge-2) at (18: 1*\beliefSize);
     \node[inner sep=0.7pt,outer xsep=-2pt] (title-2) at (18:1.25*1*\beliefSize) {\textcolor{greenish}{$1$}};
      \draw[opacity=0.5] (origin) -- (edge-2);

     \coordinate (3) at (306: 0.481*\beliefSize);
     \coordinate (edge-3) at (306: 1*\beliefSize);
     \node[inner sep=0.7pt,outer xsep=-2pt] (title-3) at (306:1.25*1*\beliefSize) {\textcolor{cyan}{$2$}};
      \draw[opacity=0.5] (origin) -- (edge-3);

     \coordinate (4) at (234: 0.467*\beliefSize);
     \coordinate (edge-4) at (234: 1*\beliefSize);
     \node[inner sep=0.7pt,outer xsep=-2pt] (title-4) at (234:1.25*1*\beliefSize) {\textcolor{maroon}{$3$}};
      \draw[opacity=0.5] (origin) -- (edge-4);

     \coordinate (5) at (162: 0.556*\beliefSize);
     \coordinate (edge-5) at (162: 1*\beliefSize);
     \node[inner sep=0.7pt,outer xsep=-2pt] (title-5) at (162:1.25*1*\beliefSize) {\textcolor{pink}{$4$}};
      \draw[opacity=0.5] (origin) -- (edge-5);

    \draw [fill=maroon, opacity=.35] (1)
                                \foreach \i in {2,...,5}{-- (\i)} --cycle;
\end{tikzpicture}

\end{tabular}
    }\hfill
    \subfloat[$t=5$ \label{fig:ta5}]{
    \newcommand{\beliefSize}{0.75}

\begin{tabular}{@{}c@{}}
\begin{tikzpicture}
    \draw[black,line width=0.5pt] (0,0) circle (1*\beliefSize);

    \coordinate (origin) at (0, 0);
    
     \coordinate (1) at (90: 1.0*\beliefSize);
     \coordinate (edge-1) at (90: 1*\beliefSize);
     \node[inner sep=0.7pt,outer xsep=-2pt] (title-1) at (90:1.25*1*\beliefSize) {\textcolor{blue}{$0$}};
     \draw[opacity=0.5] (origin) -- (edge-1);

     \coordinate (2) at (18: 0.5*\beliefSize);
     \coordinate (edge-2) at (18: 1*\beliefSize);
     \node[inner sep=0.7pt,outer xsep=-2pt] (title-2) at (18:1.25*1*\beliefSize) {\textcolor{greenish}{$1$}};
      \draw[opacity=0.5] (origin) -- (edge-2);

     \coordinate (3) at (306: 1.0*\beliefSize);
     \coordinate (edge-3) at (306: 1*\beliefSize);
     \node[inner sep=0.7pt,outer xsep=-2pt] (title-3) at (306:1.25*1*\beliefSize) {\textcolor{cyan}{$2$}};
      \draw[opacity=0.5] (origin) -- (edge-3);

     \coordinate (4) at (234: 0.5*\beliefSize);
     \coordinate (edge-4) at (234: 1*\beliefSize);
     \node[inner sep=0.7pt,outer xsep=-2pt] (title-4) at (234:1.25*1*\beliefSize) {\textcolor{maroon}{$3$}};
      \draw[opacity=0.5] (origin) -- (edge-4);

     \coordinate (5) at (162: 0.5*\beliefSize);
     \coordinate (edge-5) at (162: 1*\beliefSize);
     \node[inner sep=0.7pt,outer xsep=-2pt] (title-5) at (162:1.25*1*\beliefSize) {\textcolor{pink}{$4$}};
      \draw[opacity=0.5] (origin) -- (edge-5);

    \draw [fill=blue, opacity=.35] (1)
                                \foreach \i in {2,...,5}{-- (\i)} --cycle;
\end{tikzpicture}
\\
\begin{tikzpicture}
    \draw[black,line width=0.5pt] (0,0) circle (1*\beliefSize);

    \coordinate (origin) at (0, 0);
    
     \coordinate (1) at (90: 1.0*\beliefSize);
     \coordinate (edge-1) at (90: 1*\beliefSize);
     \node[inner sep=0.7pt,outer xsep=-2pt] (title-1) at (90:1.25*1*\beliefSize) {\textcolor{blue}{$0$}};
     \draw[opacity=0.5] (origin) -- (edge-1);

     \coordinate (2) at (18: 1.0*\beliefSize);
     \coordinate (edge-2) at (18: 1*\beliefSize);
     \node[inner sep=0.7pt,outer xsep=-2pt] (title-2) at (18:1.25*1*\beliefSize) {\textcolor{greenish}{$1$}};
      \draw[opacity=0.5] (origin) -- (edge-2);

     \coordinate (3) at (306: 1.0*\beliefSize);
     \coordinate (edge-3) at (306: 1*\beliefSize);
     \node[inner sep=0.7pt,outer xsep=-2pt] (title-3) at (306:1.25*1*\beliefSize) {\textcolor{cyan}{$2$}};
      \draw[opacity=0.5] (origin) -- (edge-3);

     \coordinate (4) at (234: 0.0*\beliefSize);
     \coordinate (edge-4) at (234: 1*\beliefSize);
     \node[inner sep=0.7pt,outer xsep=-2pt] (title-4) at (234:1.25*1*\beliefSize) {\textcolor{maroon}{$3$}};
      \draw[opacity=0.5] (origin) -- (edge-4);

     \coordinate (5) at (162: 1.0*\beliefSize);
     \coordinate (edge-5) at (162: 1*\beliefSize);
     \node[inner sep=0.7pt,outer xsep=-2pt] (title-5) at (162:1.25*1*\beliefSize) {\textcolor{pink}{$4$}};
      \draw[opacity=0.5] (origin) -- (edge-5);

    \draw [fill=cyan, opacity=.35] (1)
                                \foreach \i in {2,...,5}{-- (\i)} --cycle;
\end{tikzpicture}
\\
\begin{tikzpicture}
    \draw[black,line width=0.5pt] (0,0) circle (1*\beliefSize);

    \coordinate (origin) at (0, 0);
    
     \coordinate (1) at (90: 0.435*\beliefSize);
     \coordinate (edge-1) at (90: 1*\beliefSize);
     \node[inner sep=0.7pt,outer xsep=-2pt] (title-1) at (90:1.25*1*\beliefSize) {\textcolor{blue}{$0$}};
     \draw[opacity=0.5] (origin) -- (edge-1);

     \coordinate (2) at (18: 0.491*\beliefSize);
     \coordinate (edge-2) at (18: 1*\beliefSize);
     \node[inner sep=0.7pt,outer xsep=-2pt] (title-2) at (18:1.25*1*\beliefSize) {\textcolor{greenish}{$1$}};
      \draw[opacity=0.5] (origin) -- (edge-2);

     \coordinate (3) at (306: 0.534*\beliefSize);
     \coordinate (edge-3) at (306: 1*\beliefSize);
     \node[inner sep=0.7pt,outer xsep=-2pt] (title-3) at (306:1.25*1*\beliefSize) {\textcolor{cyan}{$2$}};
      \draw[opacity=0.5] (origin) -- (edge-3);

     \coordinate (4) at (234: 0.540*\beliefSize);
     \coordinate (edge-4) at (234: 1*\beliefSize);
     \node[inner sep=0.7pt,outer xsep=-2pt] (title-4) at (234:1.25*1*\beliefSize) {\textcolor{maroon}{$3$}};
      \draw[opacity=0.5] (origin) -- (edge-4);

     \coordinate (5) at (162: 0.536*\beliefSize);
     \coordinate (edge-5) at (162: 1*\beliefSize);
     \node[inner sep=0.7pt,outer xsep=-2pt] (title-5) at (162:1.25*1*\beliefSize) {\textcolor{pink}{$4$}};
      \draw[opacity=0.5] (origin) -- (edge-5);

    \draw [fill=maroon, opacity=.35] (1)
                                \foreach \i in {2,...,5}{-- (\i)} --cycle;
\end{tikzpicture}

\end{tabular}
    }\hfill
    \subfloat[$t=6$ \label{fig:ta6}]{
    \newcommand{\beliefSize}{0.75}

\begin{tabular}{@{}c@{}}
\begin{tikzpicture}
    \draw[black,line width=0.5pt] (0,0) circle (1*\beliefSize);

    \coordinate (origin) at (0, 0);
    
     \coordinate (1) at (90: 1.0*\beliefSize);
     \coordinate (edge-1) at (90: 1*\beliefSize);
     \node[inner sep=0.7pt,outer xsep=-2pt] (title-1) at (90:1.25*1*\beliefSize) {\textcolor{blue}{$0$}};
     \draw[opacity=0.5] (origin) -- (edge-1);

     \coordinate (2) at (18: 0.5*\beliefSize);
     \coordinate (edge-2) at (18: 1*\beliefSize);
     \node[inner sep=0.7pt,outer xsep=-2pt] (title-2) at (18:1.25*1*\beliefSize) {\textcolor{greenish}{$1$}};
      \draw[opacity=0.5] (origin) -- (edge-2);

     \coordinate (3) at (306: 1.0*\beliefSize);
     \coordinate (edge-3) at (306: 1*\beliefSize);
     \node[inner sep=0.7pt,outer xsep=-2pt] (title-3) at (306:1.25*1*\beliefSize) {\textcolor{cyan}{$2$}};
      \draw[opacity=0.5] (origin) -- (edge-3);

     \coordinate (4) at (234: 0.5*\beliefSize);
     \coordinate (edge-4) at (234: 1*\beliefSize);
     \node[inner sep=0.7pt,outer xsep=-2pt] (title-4) at (234:1.25*1*\beliefSize) {\textcolor{maroon}{$3$}};
      \draw[opacity=0.5] (origin) -- (edge-4);

     \coordinate (5) at (162: 0.5*\beliefSize);
     \coordinate (edge-5) at (162: 1*\beliefSize);
     \node[inner sep=0.7pt,outer xsep=-2pt] (title-5) at (162:1.25*1*\beliefSize) {\textcolor{pink}{$4$}};
      \draw[opacity=0.5] (origin) -- (edge-5);

    \draw [fill=blue, opacity=.35] (1)
                                \foreach \i in {2,...,5}{-- (\i)} --cycle;
\end{tikzpicture}
\\
\begin{tikzpicture}
    \draw[black,line width=0.5pt] (0,0) circle (1*\beliefSize);

    \coordinate (origin) at (0, 0);
    
     \coordinate (1) at (90: 1.0*\beliefSize);
     \coordinate (edge-1) at (90: 1*\beliefSize);
     \node[inner sep=0.7pt,outer xsep=-2pt] (title-1) at (90:1.25*1*\beliefSize) {\textcolor{blue}{$0$}};
     \draw[opacity=0.5] (origin) -- (edge-1);

     \coordinate (2) at (18: 1.0*\beliefSize);
     \coordinate (edge-2) at (18: 1*\beliefSize);
     \node[inner sep=0.7pt,outer xsep=-2pt] (title-2) at (18:1.25*1*\beliefSize) {\textcolor{greenish}{$1$}};
      \draw[opacity=0.5] (origin) -- (edge-2);

     \coordinate (3) at (306: 1.0*\beliefSize);
     \coordinate (edge-3) at (306: 1*\beliefSize);
     \node[inner sep=0.7pt,outer xsep=-2pt] (title-3) at (306:1.25*1*\beliefSize) {\textcolor{cyan}{$2$}};
      \draw[opacity=0.5] (origin) -- (edge-3);

     \coordinate (4) at (234: 0.0*\beliefSize);
     \coordinate (edge-4) at (234: 1*\beliefSize);
     \node[inner sep=0.7pt,outer xsep=-2pt] (title-4) at (234:1.25*1*\beliefSize) {\textcolor{maroon}{$3$}};
      \draw[opacity=0.5] (origin) -- (edge-4);

     \coordinate (5) at (162: 1.0*\beliefSize);
     \coordinate (edge-5) at (162: 1*\beliefSize);
     \node[inner sep=0.7pt,outer xsep=-2pt] (title-5) at (162:1.25*1*\beliefSize) {\textcolor{pink}{$4$}};
      \draw[opacity=0.5] (origin) -- (edge-5);

    \draw [fill=cyan, opacity=.35] (1)
                                \foreach \i in {2,...,5}{-- (\i)} --cycle;
\end{tikzpicture}
\\
\begin{tikzpicture}
    \draw[black,line width=0.5pt] (0,0) circle (1*\beliefSize);

    \coordinate (origin) at (0, 0);
    
     \coordinate (1) at (90: 0.569*\beliefSize);
     \coordinate (edge-1) at (90: 1*\beliefSize);
     \node[inner sep=0.7pt,outer xsep=-2pt] (title-1) at (90:1.25*1*\beliefSize) {\textcolor{blue}{$0$}};
     \draw[opacity=0.5] (origin) -- (edge-1);

     \coordinate (2) at (18: 0.586*\beliefSize);
     \coordinate (edge-2) at (18: 1*\beliefSize);
     \node[inner sep=0.7pt,outer xsep=-2pt] (title-2) at (18:1.25*1*\beliefSize) {\textcolor{greenish}{$1$}};
      \draw[opacity=0.5] (origin) -- (edge-2);

     \coordinate (3) at (306: 0.559*\beliefSize);
     \coordinate (edge-3) at (306: 1*\beliefSize);
     \node[inner sep=0.7pt,outer xsep=-2pt] (title-3) at (306:1.25*1*\beliefSize) {\textcolor{cyan}{$2$}};
      \draw[opacity=0.5] (origin) -- (edge-3);

     \coordinate (4) at (234: 0.539*\beliefSize);
     \coordinate (edge-4) at (234: 1*\beliefSize);
     \node[inner sep=0.7pt,outer xsep=-2pt] (title-4) at (234:1.25*1*\beliefSize) {\textcolor{maroon}{$3$}};
      \draw[opacity=0.5] (origin) -- (edge-4);

     \coordinate (5) at (162: 0.511*\beliefSize);
     \coordinate (edge-5) at (162: 1*\beliefSize);
     \node[inner sep=0.7pt,outer xsep=-2pt] (title-5) at (162:1.25*1*\beliefSize) {\textcolor{pink}{$4$}};
      \draw[opacity=0.5] (origin) -- (edge-5);

    \draw [fill=maroon, opacity=.35] (1)
                                \foreach \i in {2,...,5}{-- (\i)} --cycle;
\end{tikzpicture}

\end{tabular}
    }\hfill
    \subfloat[$t=7$ \label{fig:ta7}]{ 
    \newcommand{\beliefSize}{0.75}

\begin{tabular}{@{}c@{}}
\begin{tikzpicture}
    \draw[black,line width=0.5pt] (0,0) circle (1*\beliefSize);

    \coordinate (origin) at (0, 0);
    
     \coordinate (1) at (90: 1.0*\beliefSize);
     \coordinate (edge-1) at (90: 1*\beliefSize);
     \node[inner sep=0.7pt,outer xsep=-2pt] (title-1) at (90:1.25*1*\beliefSize) {\textcolor{blue}{$0$}};
     \draw[opacity=0.5] (origin) -- (edge-1);

     \coordinate (2) at (18: 0.5*\beliefSize);
     \coordinate (edge-2) at (18: 1*\beliefSize);
     \node[inner sep=0.7pt,outer xsep=-2pt] (title-2) at (18:1.25*1*\beliefSize) {\textcolor{greenish}{$1$}};
      \draw[opacity=0.5] (origin) -- (edge-2);

     \coordinate (3) at (306: 1.0*\beliefSize);
     \coordinate (edge-3) at (306: 1*\beliefSize);
     \node[inner sep=0.7pt,outer xsep=-2pt] (title-3) at (306:1.25*1*\beliefSize) {\textcolor{cyan}{$2$}};
      \draw[opacity=0.5] (origin) -- (edge-3);

     \coordinate (4) at (234: 0.5*\beliefSize);
     \coordinate (edge-4) at (234: 1*\beliefSize);
     \node[inner sep=0.7pt,outer xsep=-2pt] (title-4) at (234:1.25*1*\beliefSize) {\textcolor{maroon}{$3$}};
      \draw[opacity=0.5] (origin) -- (edge-4);

     \coordinate (5) at (162: 0.5*\beliefSize);
     \coordinate (edge-5) at (162: 1*\beliefSize);
     \node[inner sep=0.7pt,outer xsep=-2pt] (title-5) at (162:1.25*1*\beliefSize) {\textcolor{pink}{$4$}};
      \draw[opacity=0.5] (origin) -- (edge-5);

    \draw [fill=blue, opacity=.35] (1)
                                \foreach \i in {2,...,5}{-- (\i)} --cycle;
\end{tikzpicture}
\\
\begin{tikzpicture}
    \draw[black,line width=0.5pt] (0,0) circle (1*\beliefSize);

    \coordinate (origin) at (0, 0);
    
     \coordinate (1) at (90: 1.0*\beliefSize);
     \coordinate (edge-1) at (90: 1*\beliefSize);
     \node[inner sep=0.7pt,outer xsep=-2pt] (title-1) at (90:1.25*1*\beliefSize) {\textcolor{blue}{$0$}};
     \draw[opacity=0.5] (origin) -- (edge-1);

     \coordinate (2) at (18: 1.0*\beliefSize);
     \coordinate (edge-2) at (18: 1*\beliefSize);
     \node[inner sep=0.7pt,outer xsep=-2pt] (title-2) at (18:1.25*1*\beliefSize) {\textcolor{greenish}{$1$}};
      \draw[opacity=0.5] (origin) -- (edge-2);

     \coordinate (3) at (306: 1.0*\beliefSize);
     \coordinate (edge-3) at (306: 1*\beliefSize);
     \node[inner sep=0.7pt,outer xsep=-2pt] (title-3) at (306:1.25*1*\beliefSize) {\textcolor{cyan}{$2$}};
      \draw[opacity=0.5] (origin) -- (edge-3);

     \coordinate (4) at (234: 0.0*\beliefSize);
     \coordinate (edge-4) at (234: 1*\beliefSize);
     \node[inner sep=0.7pt,outer xsep=-2pt] (title-4) at (234:1.25*1*\beliefSize) {\textcolor{maroon}{$3$}};
      \draw[opacity=0.5] (origin) -- (edge-4);

     \coordinate (5) at (162: 1.0*\beliefSize);
     \coordinate (edge-5) at (162: 1*\beliefSize);
     \node[inner sep=0.7pt,outer xsep=-2pt] (title-5) at (162:1.25*1*\beliefSize) {\textcolor{pink}{$4$}};
      \draw[opacity=0.5] (origin) -- (edge-5);

    \draw [fill=cyan, opacity=.35] (1)
                                \foreach \i in {2,...,5}{-- (\i)} --cycle;
\end{tikzpicture}
\\
\begin{tikzpicture}
    \draw[black,line width=0.5pt] (0,0) circle (1*\beliefSize);

    \coordinate (origin) at (0, 0);
    
     \coordinate (1) at (90: 0.533*\beliefSize);
     \coordinate (edge-1) at (90: 1*\beliefSize);
     \node[inner sep=0.7pt,outer xsep=-2pt] (title-1) at (90:1.25*1*\beliefSize) {\textcolor{blue}{$0$}};
     \draw[opacity=0.5] (origin) -- (edge-1);

     \coordinate (2) at (18: 0.477*\beliefSize);
     \coordinate (edge-2) at (18: 1*\beliefSize);
     \node[inner sep=0.7pt,outer xsep=-2pt] (title-2) at (18:1.25*1*\beliefSize) {\textcolor{greenish}{$1$}};
      \draw[opacity=0.5] (origin) -- (edge-2);

     \coordinate (3) at (306: 0.536*\beliefSize);
     \coordinate (edge-3) at (306: 1*\beliefSize);
     \node[inner sep=0.7pt,outer xsep=-2pt] (title-3) at (306:1.25*1*\beliefSize) {\textcolor{cyan}{$2$}};
      \draw[opacity=0.5] (origin) -- (edge-3);

     \coordinate (4) at (234: 0.505*\beliefSize);
     \coordinate (edge-4) at (234: 1*\beliefSize);
     \node[inner sep=0.7pt,outer xsep=-2pt] (title-4) at (234:1.25*1*\beliefSize) {\textcolor{maroon}{$3$}};
      \draw[opacity=0.5] (origin) -- (edge-4);

     \coordinate (5) at (162: 0.494*\beliefSize);
     \coordinate (edge-5) at (162: 1*\beliefSize);
     \node[inner sep=0.7pt,outer xsep=-2pt] (title-5) at (162:1.25*1*\beliefSize) {\textcolor{pink}{$4$}};
      \draw[opacity=0.5] (origin) -- (edge-5);

    \draw [fill=maroon, opacity=.35] (1)
                                \foreach \i in {2,...,5}{-- (\i)} --cycle;
\end{tikzpicture}

\end{tabular}
    }\hfill
    \subfloat[$t=8$ \label{fig:ta8}]{ 
    \newcommand{\beliefSize}{0.75}

\begin{tabular}{@{}c@{}}
\begin{tikzpicture}
    \draw[black,line width=0.5pt] (0,0) circle (1*\beliefSize);

    \coordinate (origin) at (0, 0);
    
     \coordinate (1) at (90: 1.0*\beliefSize);
     \coordinate (edge-1) at (90: 1*\beliefSize);
     \node[inner sep=0.7pt,outer xsep=-2pt] (title-1) at (90:1.25*1*\beliefSize) {\textcolor{blue}{$0$}};
     \draw[opacity=0.5] (origin) -- (edge-1);

     \coordinate (2) at (18: 1.0*\beliefSize);
     \coordinate (edge-2) at (18: 1*\beliefSize);
     \node[inner sep=0.7pt,outer xsep=-2pt] (title-2) at (18:1.25*1*\beliefSize) {\textcolor{greenish}{$1$}};
      \draw[opacity=0.5] (origin) -- (edge-2);

     \coordinate (3) at (306: 1.0*\beliefSize);
     \coordinate (edge-3) at (306: 1*\beliefSize);
     \node[inner sep=0.7pt,outer xsep=-2pt] (title-3) at (306:1.25*1*\beliefSize) {\textcolor{cyan}{$2$}};
      \draw[opacity=0.5] (origin) -- (edge-3);

     \coordinate (4) at (234: 0.0*\beliefSize);
     \coordinate (edge-4) at (234: 1*\beliefSize);
     \node[inner sep=0.7pt,outer xsep=-2pt] (title-4) at (234:1.25*1*\beliefSize) {\textcolor{maroon}{$3$}};
      \draw[opacity=0.5] (origin) -- (edge-4);

     \coordinate (5) at (162: 1.0*\beliefSize);
     \coordinate (edge-5) at (162: 1*\beliefSize);
     \node[inner sep=0.7pt,outer xsep=-2pt] (title-5) at (162:1.25*1*\beliefSize) {\textcolor{pink}{$4$}};
      \draw[opacity=0.5] (origin) -- (edge-5);

    \draw [fill=blue, opacity=.35] (1)
                                \foreach \i in {2,...,5}{-- (\i)} --cycle;
\end{tikzpicture}
\\
\begin{tikzpicture}
    \draw[black,line width=0.5pt] (0,0) circle (1*\beliefSize);

    \coordinate (origin) at (0, 0);
    
     \coordinate (1) at (90: 1.0*\beliefSize);
     \coordinate (edge-1) at (90: 1*\beliefSize);
     \node[inner sep=0.7pt,outer xsep=-2pt] (title-1) at (90:1.25*1*\beliefSize) {\textcolor{blue}{$0$}};
     \draw[opacity=0.5] (origin) -- (edge-1);

     \coordinate (2) at (18: 1.0*\beliefSize);
     \coordinate (edge-2) at (18: 1*\beliefSize);
     \node[inner sep=0.7pt,outer xsep=-2pt] (title-2) at (18:1.25*1*\beliefSize) {\textcolor{greenish}{$1$}};
      \draw[opacity=0.5] (origin) -- (edge-2);

     \coordinate (3) at (306: 1.0*\beliefSize);
     \coordinate (edge-3) at (306: 1*\beliefSize);
     \node[inner sep=0.7pt,outer xsep=-2pt] (title-3) at (306:1.25*1*\beliefSize) {\textcolor{cyan}{$2$}};
      \draw[opacity=0.5] (origin) -- (edge-3);

     \coordinate (4) at (234: 0.0*\beliefSize);
     \coordinate (edge-4) at (234: 1*\beliefSize);
     \node[inner sep=0.7pt,outer xsep=-2pt] (title-4) at (234:1.25*1*\beliefSize) {\textcolor{maroon}{$3$}};
      \draw[opacity=0.5] (origin) -- (edge-4);

     \coordinate (5) at (162: 1.0*\beliefSize);
     \coordinate (edge-5) at (162: 1*\beliefSize);
     \node[inner sep=0.7pt,outer xsep=-2pt] (title-5) at (162:1.25*1*\beliefSize) {\textcolor{pink}{$4$}};
      \draw[opacity=0.5] (origin) -- (edge-5);

    \draw [fill=cyan, opacity=.35] (1)
                                \foreach \i in {2,...,5}{-- (\i)} --cycle;
\end{tikzpicture}
\\
\begin{tikzpicture}
    \draw[black,line width=0.5pt] (0,0) circle (1*\beliefSize);

    \coordinate (origin) at (0, 0);
    
     \coordinate (1) at (90: 0.552*\beliefSize);
     \coordinate (edge-1) at (90: 1*\beliefSize);
     \node[inner sep=0.7pt,outer xsep=-2pt] (title-1) at (90:1.25*1*\beliefSize) {\textcolor{blue}{$0$}};
     \draw[opacity=0.5] (origin) -- (edge-1);

     \coordinate (2) at (18: 0.599*\beliefSize);
     \coordinate (edge-2) at (18: 1*\beliefSize);
     \node[inner sep=0.7pt,outer xsep=-2pt] (title-2) at (18:1.25*1*\beliefSize) {\textcolor{greenish}{$1$}};
      \draw[opacity=0.5] (origin) -- (edge-2);

     \coordinate (3) at (306: 0.523*\beliefSize);
     \coordinate (edge-3) at (306: 1*\beliefSize);
     \node[inner sep=0.7pt,outer xsep=-2pt] (title-3) at (306:1.25*1*\beliefSize) {\textcolor{cyan}{$2$}};
      \draw[opacity=0.5] (origin) -- (edge-3);

     \coordinate (4) at (234: 0.468*\beliefSize);
     \coordinate (edge-4) at (234: 1*\beliefSize);
     \node[inner sep=0.7pt,outer xsep=-2pt] (title-4) at (234:1.25*1*\beliefSize) {\textcolor{maroon}{$3$}};
      \draw[opacity=0.5] (origin) -- (edge-4);

     \coordinate (5) at (162: 0.526*\beliefSize);
     \coordinate (edge-5) at (162: 1*\beliefSize);
     \node[inner sep=0.7pt,outer xsep=-2pt] (title-5) at (162:1.25*1*\beliefSize) {\textcolor{pink}{$4$}};
      \draw[opacity=0.5] (origin) -- (edge-5);

    \draw [fill=maroon, opacity=.35] (1)
                                \foreach \i in {2,...,5}{-- (\i)} --cycle;
\end{tikzpicture}

\end{tabular}
    }
    \caption{\added{Evolution of the  ABs of three agents (0,2, and 3) using radar plots over the probability simplex.}
    The first three rows of radar plots show the evolution of the SDHT algorithm and the second three rows show the evolution of the ADHT algorithm.} 
    \label{fig:belief_evo}
\end{figure*}

\subsection{Results}

In this section we present two simulation results\footnote{For videos and source code of all of these simulations see \url{https://u-t-autonomous.github.io/Decentralized_Hypothesis_Testing/}.}. The first result compares the SDHT and ADHT algorithms in the 5-agent scenario as shown in Fig.~\ref{fig:case_env}. The second result compares the minimum and averaging rules with high and low levels of sensor noise. {A high (low) sensor noise means local likelihood functions with high (low) variances}

In the simulations, at any time instant, we assume there are two possible locations of the agent $j$, namely $q_{j}^0$ and $q_{j}^1$, depending on the value of $\theta(j) \in \{ 0,1 \}$ (see Fig.~\ref{fig:obs}). {Therefore, $P(q_{j}^0|0,q_{i})=P(q_{j}^1|1,q_{i})=1$.}

For a given  hypothesis $\theta(j)$ and its corresponding  location $q_{j}^{\theta(j)}$,  
from \eqref{equation:obervation function for agent j in case study}, the likelihood function $l^j_i(s_i^j|\theta(j),q_{i})$ to get $s_i^j$ for agent $i$ is:
\begin{align}\label{equation:local_single}
  &l^j_i(s_i^j|\theta(j),q_{i})=\nonumber\\
  &\begin{cases}
   P_{i}(s_i^j|{q_{i}},q_{j}^{\theta(j)}) & \text{ if } s_i^j\neq\emptyset,\\
   0 & \text{ if } s_i^j=\emptyset \land q_{j}^{\theta(j)}\in\mathcal{Q}_i(q_{i}),\\
   1 & \text{ if } s_i^j=\emptyset \land q_{j}^{\theta(j)}\notin\mathcal{Q}_i(q_{i}).
\end{cases}
\end{align}

\subsubsection{SDHT vs ADHT}

\added{Fig.~\ref{fig:sdht-adht} compares how each agent's AB on the true hypothesis  $\theta^*=(1,1,1,0,1)$ evolves over time for SDHT and ADHT.} Agent 3 (grey) is a bad agent. All the good agents have the same prior belief that each agent is equally likely to be good or bad. \added{Both algorithms converge to the true hypothesis despite the bad agent (agent 3) sharing randomly generated ABs.}
SDHT in Fig.~\ref{fig:sdht} converges at around $t=16$, while the convergence with ADHT is faster at $t=9$ as shown in Fig.~\ref{fig:adht}. 
We also empirically observe that ADHT enters case one much more frequently from Fig.~\ref{fig:belief_calls}. 
Therefore, the agents make much more frequent use of neighbor information in ADHT  and converge faster than they do in SDHT.

\added{To better illustrate the agents' belief evolution, we pick agents $0,2$, and $3$ and show  their ABs at different time instants in both SDHT and ADHT algorithms in Fig.~\ref{fig:belief_evo}.  The radar plots indicate each agent's AB, where each vertex $i$ ($i\in\{0,1,2,3,4\}$) represents the probability that agent $i$ is bad.} From Fig.~\ref{fig:belief_evo}, agent $2$ converges to the true belief at $t=8$ for both algorithms. However, for SDHT, it is not until $t=17$ does agent $0$ make use of agent $2$'s AB and converge. \added{While in ADHT at $t=8$, agent $0$ has already accumulated enough shared beliefs to update its AB and converges.} 

\begin{figure}[h]
    \centering
    \pgfplotstableread[col sep = comma]{pics/plots/Local_Low_Var.csv}\lLV
\pgfplotstableread[col sep = comma]{pics/plots/Local_UHigh_Var.csv}\lHV
\pgfplotstableread[col sep = comma]{pics/plots/Min_Low_Var.csv}\mLV
\pgfplotstableread[col sep = comma]{pics/plots/Min_High_Var.csv}\mHV

\begin{tikzpicture}
	\begin{axis}[
    xlabel = {$t$},
	ylabel = {Belief over $\theta^*$},
	ymin= 0,
	ymax= 1.20,
	xmin=0,
	xmax=75,
	label style={font=\small},
	tick label style={font=\small},
	width = 0.45*\textwidth,
	height = 0.25*\textwidth,
	legend style={at={($(0,0)+(3.0cm,-1.6cm)$)},legend columns=2,fill=none,draw=black,anchor=center,align=center}]
	
	\addplot+[no markers, line width=2pt,color=blue] table [x index ={0},y index ={1}]{\mLV};\label{mLv:a41}\addlegendentry{LV actual  $b^a$\quad}
	\addplot+[no markers, dotted, line width=2pt,color=blue] table [x index ={0},y index ={1}]{\lLV};\label{lLv:a01}\addlegendentry{LV local $b^l$}
    \addplot+[no markers, line width=2pt,color=red] table [x index ={0},y index ={3}]{\mHV};\label{mhv:a41}\addlegendentry{HV actual  $b^a$\quad}
	\addplot+[no markers, dotted, line width=2pt,color=red] table [x index ={0},y index ={1}]{\lHV};\label{lLv:a0}\addlegendentry{HV local $b^l$}
	
	\end{axis}
\end{tikzpicture}
    \caption{\added{Evolution of LBs and ABs for agent 0 for the low variance (LV) and high variance (HV) sensor noise cases. The solid line uses belief sharing with ADHT and the minimum rule.}}
    \label{fig:local_belief}
\end{figure}
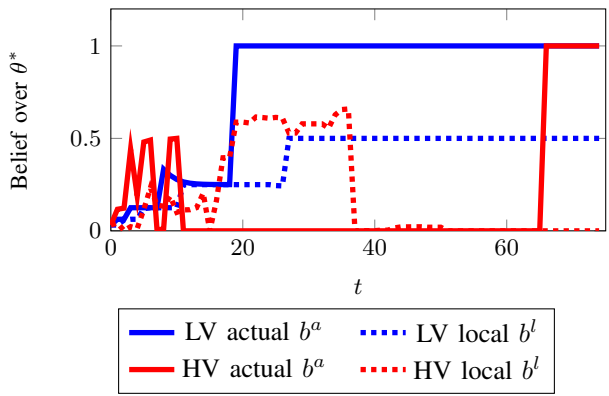

\begin{figure}[h]
    \centering
    \subfloat[Averaging rule -- high sensor noise]{
    \pgfplotstableread[col sep = comma]{pics/plots/Avg_High_Var.csv}\aHV
\begin{tikzpicture}
	\begin{axis}[
	xlabel = {$t$},
	ylabel = {Actual Belief $b^a_j(\theta^\star)$},
	ymin= 0,
	ymax= 1.10,
	xmin=0,
	xmax=75,
	label style={font=\small},
	tick label style={font=\small},
	width = 0.45*\textwidth,
	height = 0.27*\textwidth,
	legend style={at={(1,0.75)},anchor=north east}]
	\addplot+[no markers, dashed, line width=2pt,color=blue] table [x index ={0},y index ={3}]{\aHV};\label{ahv:a0}\addlegendentry{Agent 0}
	\addplot+[no markers, dashed, line width=2pt,color=greenish] table [x index ={0},y index ={4}]{\aHV};\label{ahv:a1}\addlegendentry{Agent 1} 
	\addplot+[no markers, dashed, line width=2pt,color=cyan] table [x index ={0},y index ={5}]{\aHV};\label{ahv:a2}\addlegendentry{Agent 2} 
	\addplot+[no markers, dashed, line width=2pt,color=maroon] table [x index ={0},y index ={2}]{\aHV};\label{ahv:a3}\addlegendentry{Agent 3}
	\addplot+[no markers, dashed, line width=2pt,color=pink] table [x index ={0},y index ={1}]{\aHV};\label{ahv:a4}\addlegendentry{Agent 4}
	\end{axis}
\end{tikzpicture}
    }\\
    \subfloat[Minimum rule -- high sensor noise]{
    \pgfplotstableread[col sep = comma]{pics/plots/Min_High_Var.csv}\mHV
\begin{tikzpicture}
	\begin{axis}[
	xlabel = {$t$},
	ylabel = {Actual Belief $b^a_j(\theta^\star)$},
	ymin= 0,
	ymax= 1.10,
	xmin=0,
	xmax=75,
	label style={font=\small},
	tick label style={font=\small},
	width = 0.45*\textwidth,
	height = 0.25*\textwidth]
	\addplot+[no markers, dashed, line width=2pt,color=blue] table [x index ={0},y index ={5}]{\mHV};\label{mhv:a0}
	\addplot+[no markers, dashed, line width=2pt,color=greenish] table [x index ={0},y index ={1}]{\mHV};\label{mhv:a1}
	\addplot+[no markers, dashed, line width=2pt,color=cyan] table [x index ={0},y index ={4}]{\mHV};\label{mhv:a2}
	\addplot+[no markers, dashed, line width=2pt,color=maroon] table [x index ={0},y index ={2}]{\mHV};\label{mhv:a3}
	\addplot+[no markers, dashed, line width=2pt,color=pink] table [x index ={0},y index ={3}]{\mHV};\label{mhv:a4}
	\end{axis}
\end{tikzpicture}
    }
    \caption{\added{Each agent's AB $b^a_{j,t}(\theta^*)$ over time $t$ for the true hypothesis $\theta^*$ where $\theta^*=(1,1,1,0,1)$  with high sensor noises.} Two figures share the same legend.
    \label{fig:average_rule}}
    \label{fig:belief_plot}
\end{figure}
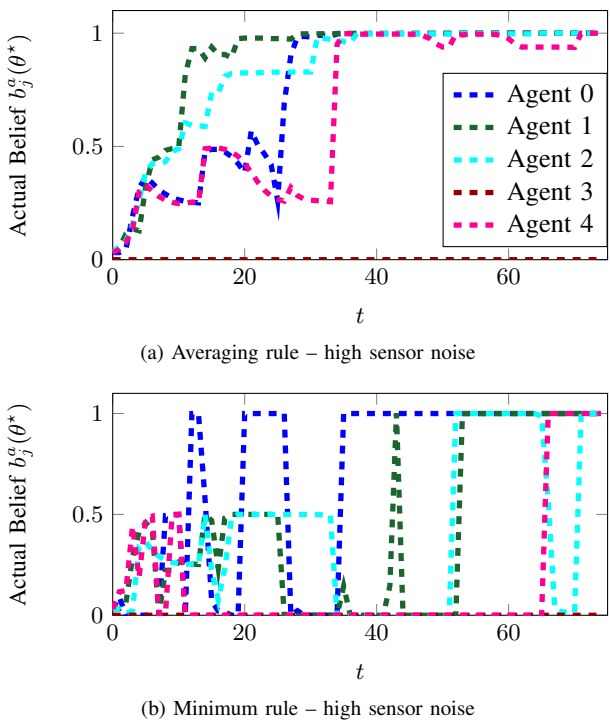

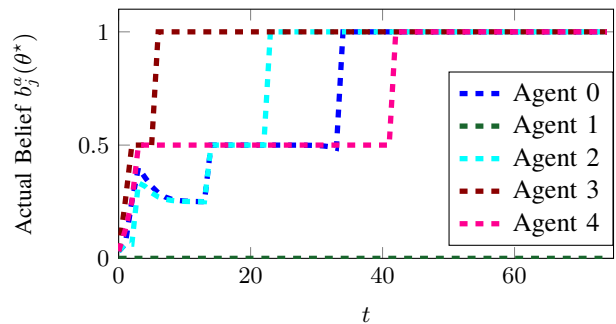
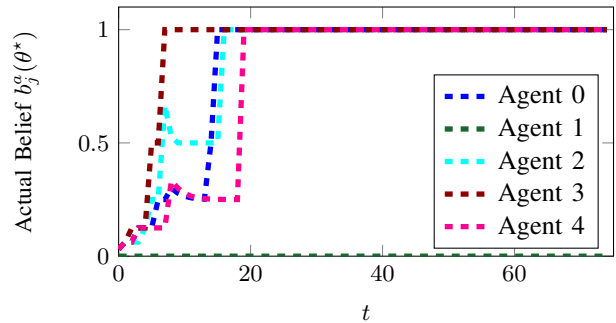
\begin{figure}[h]
    \centering
    \subfloat[Averaging rule -- low sensor noise\label{fig:average_rule_low}]{
    \pgfplotstableread[col sep = comma]{pics/plots/Avg_Low_Var.csv}\aLV
\begin{tikzpicture}
	\begin{axis}[
	xlabel = {$t$},
	ylabel = {Actual Belief $b^a_j(\theta^\star)$},
	ymin= 0,
	ymax= 1.10,
	xmin=0,
	xmax=75,
	label style={font=\small},
	tick label style={font=\small},
	width = 0.45*\textwidth,
	height = 0.27*\textwidth,
	legend style={at={(1,0.75)},anchor=north east}]
	\addplot+[no markers, dashed, line width=2pt,color=blue] table [x index ={0},y index ={3}]{\aLV};\label{aLv:a0}\addlegendentry{Agent 0}
	\addplot+[no markers, dashed, line width=2pt,color=greenish] table [x index ={0},y index ={4}]{\aLV};\label{aLv:a1}\addlegendentry{Agent 1} 
	\addplot+[no markers, dashed, line width=2pt,color=cyan] table [x index ={0},y index ={5}]{\aLV};\label{aLv:a2}\addlegendentry{Agent 2} 
	\addplot+[no markers, dashed, line width=2pt,color=maroon] table [x index ={0},y index ={2}]{\aLV};\label{aLv:a3}\addlegendentry{Agent 3}
	\addplot+[no markers, dashed, line width=2pt,color=pink] table [x index ={0},y index ={1}]{\aLV};\label{aLv:a4}\addlegendentry{Agent 4}
	\end{axis}
\end{tikzpicture}
    }\\
    \subfloat[Minimum rule -- low sensor noise]{
    \pgfplotstableread[col sep = comma]{pics/plots/Min_Low_Var.csv}\mLV
\begin{tikzpicture}
	\begin{axis}[
	xlabel = {$t$},
	ylabel = {Actual Belief $b^a_j(\theta^\star)$},
	ymin= 0,
	ymax= 1.10,
	xmin=0,
	xmax=75,
	label style={font=\small},
	tick label style={font=\small},
	width = 0.45*\textwidth,
	height = 0.27*\textwidth,
	legend pos=south east]
	\addplot+[no markers, dashed, line width=2pt,color=blue] table [x index ={0},y index ={3}]{\mLV};\label{mLv:a0}\addlegendentry{Agent 0}
	\addplot+[no markers, dashed, line width=2pt,color=greenish] table [x index ={0},y index ={4}]{\mLV};\label{mLv:a1}\addlegendentry{Agent 1} 
	\addplot+[no markers, dashed, line width=2pt,color=cyan] table [x index ={0},y index ={5}]{\mLV};\label{mLv:a2}\addlegendentry{Agent 2} 
	\addplot+[no markers, dashed, line width=2pt,color=maroon] table [x index ={0},y index ={2}]{\mLV};\label{mLv:a3}\addlegendentry{Agent 3}
	\addplot+[no markers, dashed, line width=2pt,color=pink] table [x index ={0},y index ={1}]{\mLV};\label{mLv:a40}\addlegendentry{Agent 4}
	\end{axis}
\end{tikzpicture}
    }
    \caption{\added{Each agent's AB $b^a_{j,t}(\theta^*)$ over time $t$  for the true hypothesis $\theta^*$ where $\theta^*(i)=(1,1,1,0,1)$  with low sensor noises.}
    }
    \label{fig:belief_plot_low}
\end{figure}

\subsubsection{\added{AB Update Rule}}

We showed in Section~\ref{section:average rule} that the average rule also guarantees the  convergence to the true underlying belief. Examining the effect of sensor noise in terms of the variances of the local likelihood functions provides a comparison between the average rule and the minimum rule. Agent 3 is the bad agent who always shares the same false belief $b_{3,t}^a(\theta)=1$ where $\theta=(1,0,1,1,1)$. In other words, it always broadcasts to its neighbors that agent $1$ is the bad agent almost surely. 

\added{Fig.~\ref{fig:local_belief} shows the evolution of LBs and ABs for agent $0$. It can be seen that, especially for the first $15$ time steps,  a high sensor noise frequently leads to fluctuations in an agent's LBs. Such fluctuations propagate to its ABs that are shared to its neighbors.} Consequently, as shown in Fig.~\ref{fig:average_rule}, the average rule outperforms the minimum rule in identifying the true hypothesis since it relies on more than one neighboring agent which may average out the fluctuation for each hypothesis. \added{In the low sensor noise scenario, the LB has much less fluctuations as shown in Fig.~\ref{fig:local_belief}.} Then we observe that the minimum rule converges faster since it may quickly and correctly rule out the wrong hypotheses by taking the minimum of the beliefs  as illustrated in Fig.~\ref{fig:belief_plot_low}.

\subsection{Expanded Case Studies}
\added{
We demonstrate the algorithm for Byzantine fault tolerance on two alternative case studies: one is the same setting as in \ref{ssec:setting} with an agent transmitting a fixed false hypothesis and another with an expanded version of the environment in Fig.~\ref{fig:environment} with ten good agents and two coordinating bad agents.}
\footnote{Videos of these case studies can be found at \url{https://u-t-autonomous.github.io/Decentralized_Hypothesis_Testing/}.}
\added{
The two bad agents are coordinating by constantly transmitting the same false hypothesis to its neighbours in the system.}

\added{
In Fig~\ref{fig:expanded}, instead of showing ten curves for ten good agents in one figure that may affect readability, we plot the average of the ABs (solid line) and LBs (dashed line) over time. One can see that the ADHT method converges to the correct hypothesis significantly faster than if no information was shared (LB that only rely on local information). Further, ADHT is robust against two coordinated bad agents.
Note that in this case study where the agent's location on the state paths define the system status, the size of the hypothesis set $|\Theta|$ scales exponentially with the number of agents.
}
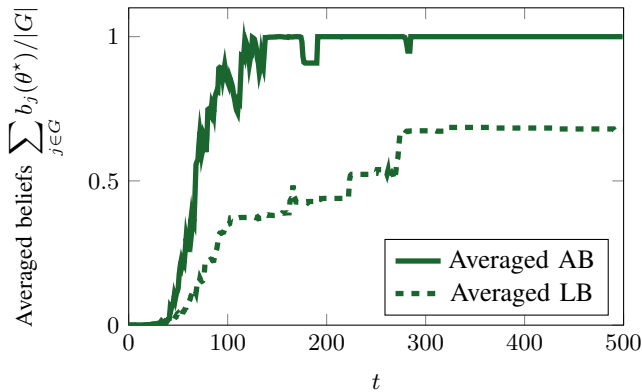
\begin{figure}
    \centering
    \pgfplotstableread[col sep = comma]{pics/plots/Expanded.csv}\expAND
\begin{tikzpicture}
	\begin{axis}[
	xlabel = {$t$},
	ylabel = {Averaged beliefs $\displaystyle \sum_{j\in G} b_j(\theta^\star)/|G|$},
	ymin= 0,
	ymax= 1.10,
	xmin=0,
	xmax=500,
	label style={font=\small},
	tick label style={font=\small},
	width = 0.45*\textwidth,
	height = 0.32*\textwidth,
	legend pos=south east]
	\addplot+[no markers, solid, line width=2pt,color=greenish] table [x index ={0},y index ={2}]{\expAND};\label{Expand:AB}\addlegendentry{Averaged AB}
	\addplot+[no markers, dashed, line width=2pt,color=greenish] table [x index ={0},y index ={1}]{\expAND};\label{Expand:LB}\addlegendentry{Averaged LB} 
	\end{axis}
\end{tikzpicture}
    \caption{\added{The averaged LBs and ABs in the expanded case study with ten good agents and two bad agents. Both will eventually converge to one but with ADHT the AB converges more quickly despite the two coordinated bad agents transmitting false hypothesis data.}}
    \label{fig:expanded}
\end{figure}

\section{Conclusion}\label{section:conclusion}
In this paper, we introduce two resilient distributed hypothesis testing algorithms in a time-varying network topology.  \added{Each agent makes local observations and keeps simulating shared information to update its LBs and ABs over all possible hypotheses.}   We prove that the proposed algorithms guarantee almost-sure convergence to the true hypothesis in the limit without {requiring that the underlying network topology to be connected.} The proposed algorithms are simple to implement and resilient to adversarial agents. The results in the simulated case studies illustrate the validity of the proposed approaches and compare their performance in different scenarios. \added{In particular, we show that the asynchronous algorithm constantly converges faster than the synchronous algorithm. Furthermore,  the performances of average and minimum rules that make use of shared ABs depend heavily on the sensor noise. With higher sensor noise, the former outperforms the latter. And with lower sensor noise, the reverse is true.}  Future work will study how to plan the state paths of the team in a distributed manner to satisfy the convergence conditions. 
\bibliographystyle{IEEEtran}
\bibliography{ref}
\begin{IEEEbiography}[{\includegraphics[width=1in,height=1.25in,clip,keepaspectratio]{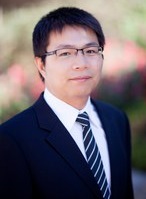}}]{Bo Wu}
received his B.E. degree from Harbin Institute of Technology, China, in 2008, an M.S. degree from Lund University, Sweden, in 2011 and Ph.D. degree from the University of Notre Dame, USA, in 2018, all in electrical engineering. He is currently a postdoctoral researcher at the Oden Institute for Computational Engineering and Sciences at the University of Texas at Austin. His research interest is to apply formal methods, learning, and control in autonomous systems, such as robotic systems, communication systems, and human-in-the-loop systems, to provide privacy, security, and performance guarantees.
\end{IEEEbiography}
\begin{IEEEbiography}
[{\includegraphics[width=1in,height=1.25in,clip,keepaspectratio]{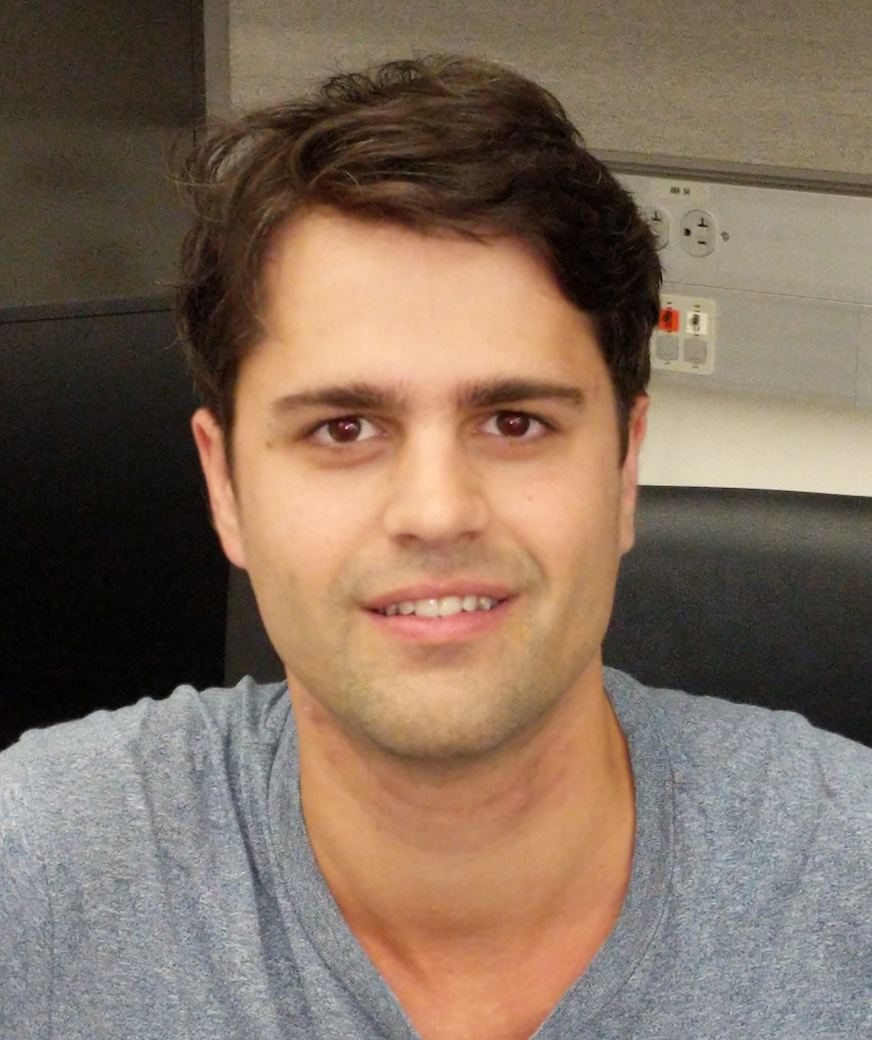}}]{Steven Carr} is currently pursing his Ph.D. degree from the University of Texas at Austin in the Department of Aerospace Engineering. He received the B.Eng./B.Sc. in aerospace and mathematics from the University of Sydney in 2014 and the M.Sc in aerospace engineering in 2018. His research interests include the intersection of control and learning in autonomous systems with a focus on aerospace applications.
\end{IEEEbiography}

\begin{IEEEbiography}[{\includegraphics[width=1in,height=1.25in,clip,keepaspectratio]{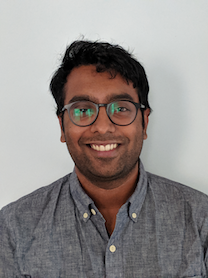}}]{Suda Bharadwaj}
Suda Bharadwaj received B.Sc. and B.E degrees in applied mathematics and aerospace engineering from the University of Sydney, NSW, Australia, in 2014. In 2016, he received an M.S. degree in aerospace engineering from the University of Texas at Austin, TX, USA. He is currently pursuing his Ph.D degree at the Department of Aerospace Engineering and Engineering Mechanics at the University of Texas at Austin. His research interests include the intersection of formal methods, reinforcement learning, and control with a focus on provable safety guarantees. 
\end{IEEEbiography}

\begin{IEEEbiography}[{\includegraphics[width=1in,height=1.25in,clip,keepaspectratio]{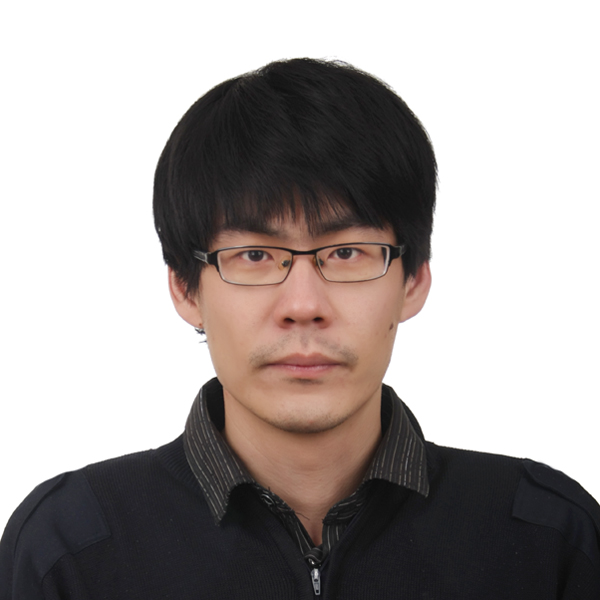}}]{Zhe Xu}
received the B.S. and M.S. degrees in Electrical Engineering from Tianjin University, Tianjin, China, in 2011 and 2014, respectively. He received the Ph.D. degree in Electrical Engineering at Rensselaer Polytechnic Institute, Troy, NY, in 2018. He is currently an assistant professor in the School for Engineering of Matter, Transport, and Energy at Arizona State University. Before joining ASU, he was a postdoctoral researcher in the Oden Institute for Computational Engineering and Sciences at the University of Texas at Austin, Austin, TX. His research interests include formal methods, autonomous systems, control systems and reinforcement learning. 
\end{IEEEbiography}
\begin{IEEEbiography}[{\includegraphics[width=1in,height=1.25in,clip,keepaspectratio]{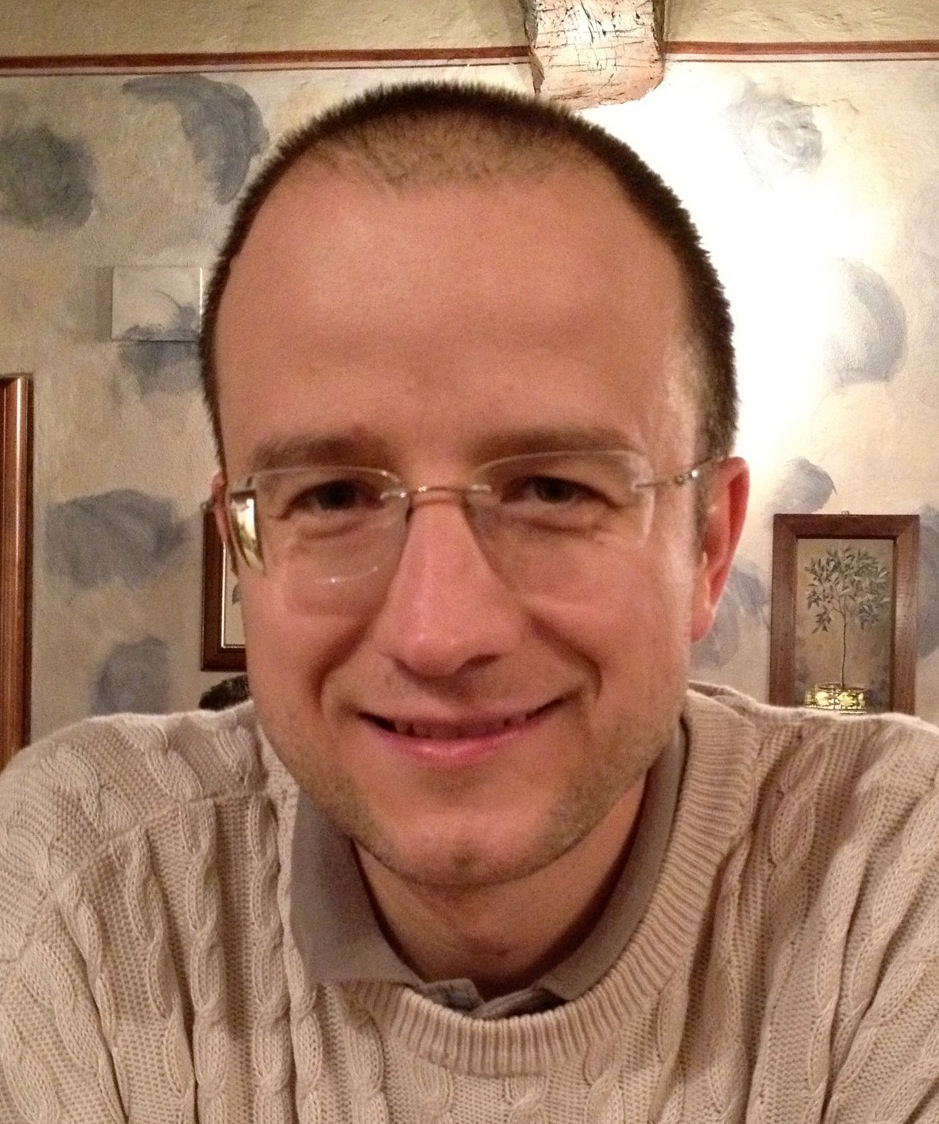}}]{Ufuk Topcu}
Ufuk Topcu joined the Department of Aerospace Engineering at the University of Texas at Austin as an assistant professor in Fall 2015. He received his Ph.D. degree from the University of California at Berkeley in 2008. He held research positions at the University of Pennsylvania and California Institute of Technology. His research focuses on the theoretical, algorithmic and computational aspects of design and verification of autonomous systems through novel connections between formal methods, learning theory and controls. 
\end{IEEEbiography}

\clearpage

\section{Appendix}
\beginsupplement

\subsection{Proof of Lemma \ref{lemma:convergence}}
\begin{proof}
For any good agent $i\in G$,  we define
\begin{equation}\label{equation:lemma1}
\begin{split}
     \rho_{i,t}(\theta)&:=\log\frac{b^l_{i,t}(\theta)}{b^l_{i,t}(\theta^*)} , \text{ and }\\   
    \lambda_{i,t}(\theta)&:=\log\frac{l_{i}(s_{i,t}|\theta,q_{i,t})}{l_{i}(s_{i,t}|\theta^*,q_{i,t})}.
\end{split}
\end{equation}
Note that $l_{i}(s_{i,t}|\theta^*,q_{i,t})>0$ for all $t$, $q_{i,t}$ and $s_{i,t}$ since $\theta^*$ is the true hypothesis that generates the observation $s_{i,t}$. Therefore, we know that, for any finite $t$, $b^l_{i,t}(\theta^*)>0$ and \eqref{equation:lemma1} is always well-defined. 
\added{Then according to the LB-update rule (\ref{equation:LB update rule}), we have }
$$
    \rho_{i,t+1}(\theta)=\rho_{i,t}(\theta)+\lambda_{i,t}(\theta),
$$
which yields
\begin{equation}\label{equation:rho sum}
    \added{\rho_{i,T+1}(\theta)=\rho_{i,0}(\theta)+\sum_{t=0}^{T}\lambda_{i,t}(\theta).}
\end{equation}
Note that, according to equation (\ref{equation:lemma convergence}), there are cases where $q_{i,t}\notin O_i(\theta,\theta^*)$, which implies
$$
l_i(.|\theta^*,q_{i,t})=l_i(.|\theta,q_{i,t}).
$$
In this case, $\lambda_{i,t}(\theta)=0$ and does not contribute to the sum in (\ref{equation:rho sum}). Therefore, we may only focus on the case where $q_{i,t}\in O_i(\theta,\theta^*)$ and thus $\lambda_{i,t}(\theta)\neq 0$. 

Note that  $\{\lambda_{i,t}(\theta)\}$ is a sequence of independent random variables. For a given $t$, we have
$$E_{\theta^*}[\lambda_{i,t}(\theta)]=-D(l_i(.|\theta^*,q_{i,t})||l_i(.|\theta,q_{i,t})).$$

We denote a set $Q_\infty\subseteq Q$ for those locations where $\theta$ and $\theta^*$ can be differentiated and are visited infinite times by agent $i$. Formally, 
\added{$$Q_\infty:=\{q|q\in O_i(\theta,\theta^*)\text{ and   }\lim_{T\rightarrow\infty}\sum_{t=0}^{T}I(q_{i,t}=q)=\infty\}.$$}
We claim that $Q_\infty$ is non-empty by contradiction. If $Q_\infty$ is empty, it implies that the agent visits none of the states $q\in O_i(\theta,\theta^*)$ infinitely often, which violates the condition implied by $i\in S(\theta,\theta^*)$ and equation (\ref{equation:lemma convergence}). 

For any $q\in Q_\infty$, the following is true based on the strong law of large numbers.
\begin{equation}
 \begin{split}
      &\lim_{T\rightarrow\infty}\frac{1}{T}\sum_{t=1}^TI(q_{i,t}=q)\lambda_{i,t}(\theta)\\&= -D(l_i(.|\theta^*,q)||l_i(.|\theta,q)) \text{ almost surely.}
 \end{split}   
\end{equation}

\added{We divide both sides of \eqref{equation:rho sum} by $T$ and take the limit which yields}
\begin{equation}\label{equation:kl}
\begin{split}
     &\lim_{T\rightarrow\infty}\frac{1}{T}\rho_{i,T+1}(\theta)=\lim_{T\rightarrow\infty} \frac{1}{T} (\rho_{i,0}(\theta)+\sum_{t=0}^{T}\lambda_{i,t}(\theta))    \\
     &=\lim_{T\rightarrow\infty} \frac{1}{T} \sum_{t=0}^{T}\lambda_{i,t}(\theta)\\
     &=-\sum_{q\in Q_\infty} D(l_i(.|\theta^*,q)||l_i(.|\theta,q))\text{ almost surely.}
\end{split}
\end{equation}
Note that, for those $q\in O_i(\theta,\theta^*)$ but $q\notin Q_\infty$, their contribution in \eqref{equation:kl} is zero since they are only visited a finite number of times. By definition of $O_i(\theta,\theta^*)$, we know that $D(l_i(.|\theta^*,q)||l_i(.|\theta,q))>0$ for $q\in O_i(\theta,\theta^*)$. Then from \eqref{equation:kl}, $\rho_{i,t+1}(\theta)\rightarrow-\infty$ almost surely which implies $b_{i,t}^l(\theta)\rightarrow 0$ almost surely and proves \eqref{equation:zero of nontrue hypothesis}.

\added{Additionally, to prove \eqref{equation:nonzero of true hypothesis}, we define a set $$\bar{\Theta}:=\{\theta|i\notin S(\theta,\theta^*)\}$$ to include every hypothesis $\theta$ that agent $i$ is not able to differentiate from $\theta^*$.} \added{Then from the second condition of Theorem \ref{theorem:main result}, for each $\theta\in\bar{\Theta}$, there must exist a time $T_\theta$ such that
$$
\lim_{T\rightarrow\infty}\sum_{t=T_\theta+1}^{T}I(q_{i,t}\in S(\theta,\theta^*))=0.
$$}
\noindent That is, there exists a time $T_\theta$ after which agent $i$ will never visit any position that can differentiate $\theta$ and $\theta^*$\footnote{\added{Since the time zone is discrete, we use
$\cdot +1$ in $ T_{\theta}+1$.}}. Given any local state observation path $\omega_i= \{(q_{i,0},s_{i,0}),(q_{i,1},s_{i,1}),...\}$ where \eqref{equation:zero of nontrue hypothesis} holds, it is immediate from \eqref{equation:rho sum} that 
\begin{equation}\label{equation:rho i t theta}
 \rho_{i,t}(\theta)=\rho_{i,0}(\theta)+\sum_{j=0}^{T_\theta}\lambda_{i,j}(\theta)=C_{\theta,\omega_i}<\infty, 
\end{equation}
for any $t\geq T_\theta$  and some constant $C_{\theta,\omega_i}$ that depends on both $\theta$ and $\omega_i$ due to the term $\lambda_{i,j}(\theta)$. For fixed $\omega_i$, it is then possible to find $\lim_{t\rightarrow\infty}b^l_{i,t}(\theta^*)$ from \eqref{equation:rho i t theta}, which is nonzero. \added{When combining with the fact that $b^l_{i,t}(\theta^*)$ is nonzero for any finite $t$ stated as a pre-assumption in Lemma \ref{lemma:convergence}, we conclude that \eqref{equation:nonzero of true hypothesis} is proved.}
\end{proof}

\subsection{Proof of Lemma \ref{lemma:nonzero of true hypothesis}}
\begin{proof}
We prove this lemma by contradiction. Suppose there is a time $t$ where $b^a_{i,t}(\theta^*)=0$ for the first time for a good agent $i$.  \added{From Lemma \ref{lemma:convergence} we know that
$b^{\ell}_{i,t-1} (\theta^*)>0$, consequently and logically,
$b^{a}_{i,t-1} (\theta^*)>0$ holds. }
Therefore, from \eqref{equation:acutual belief update rule for case two} it immediately follows that it cannot happen in case two in SDHT. 

Therefore we  infer that $b^a_{i,t}(\theta^*)=0$ can only result from an update in case one in SDHT. From \eqref{equation:acutual belief update rule for case one}, this is only possible when $\min_{j\in \mathcal{N}^{\theta^*}_{i,t}}\{b^a_{j,t-1}(\theta^*)\}=0$. Note that in case one, we remove $f$ number of lowest beliefs on $\theta^*$ as in Line \ref{algorithm:remove beliefs} of SDHT. \added{In the worst case, we remove all the $f$ ABs that are zero from the bad agents. Then what is left are the ABs from good agents, which are nonzero from the definition of this time $t$.} For all other cases, the removed lowest $f$ ABs must contain nonzero entries, which implies that all the beliefs for agents in $\mathcal{N}^{\theta^*}_{i,t}$ are nonzero as well. In either case, we have that $\min_{j\in \mathcal{N}^{\theta^*}_{i,t}}\{b^a_{j,t-1}(\theta^*)\}>0$ which leads to a contradiction. 
\end{proof}

\subsection{Proof of Theorem \ref{theorem:main result}}\label{subsec:proof of Theorem 1}
\begin{proof}
\added{With the proof of Lemma \ref{lemma:convergence} and \ref{lemma:lower and upper bounded beliefs}, now we are ready to give the proof for Theorem \ref{theorem:main result}.} We are interested in state observation path set $\hat{\Omega}$ as defined in Remark \ref{remark:lemma 1} since $\hat{\Omega}$ has measure one.

The proof consists of two parts. \added{First, we prove that the AB over the true hypothesis $b^a_{i,t}(\theta^*)$ for any good agent $i$ is lower-bounded. Then we show that  the AB over the rest of the hypotheses will become arbitrarily small.} These two parts together are sufficient to prove that the $b^a_{i,t}(\theta^*)$ will be arbitrarily close to one almost surely. 

\added{For the first part, if case one happens only finitely often for a good agent $i\in G$ for true hypothesis $\theta^*$, then by condition two in Theorem \ref{theorem:main result}, we know that $i\in S(\theta,\theta^*)$ for any $\theta\neq\theta^*$. Therefore, by Lemma \ref{lemma:convergence} we know that LB $b_{i,t}^l(\theta^*)\rightarrow 1$ almost surely and so is AB, then the proof is done. Otherwise, if case on happens infinitely often for a good agent $i$, we fix a path $\omega\in\hat{\Omega}$ and define $\delta_1:=\min_{i\in G} \lim_{t\rightarrow\infty}b^l_{i,t}(\theta^*)$.} Then, for each good agent $i\in G$, there exist a time $t_i$ and a constant $\alpha$ such that, for all $t\geq t_i$, we have $b^l_{i,t}(\theta^*)\geq\delta_1-\alpha$ where $\alpha<\delta_1$. We define 
\begin{equation}\label{equation:bar t 1}
\bar{t}_1 := \max_{i\in G}t_i. 
\end{equation}
We also define $\delta_2: = \min_{i\in G}b^a_{i,\bar{t}_1}(\theta^*)$. By Lemma \ref{lemma:nonzero of true hypothesis}, we know $\delta_2 >0$. We further define
\begin{equation}\label{equation:delta1}
    \delta:=\min\{\delta_1-\alpha,\delta_2\}. 
\end{equation}
\added{Then at $t= \bar{t}_1+1$, in SDHT, for AB update, either case one or case two happens.}  If case one happens, we use \eqref{equation:acutual belief update rule for case one} to update the belief for $\theta^*$, then we will have
\begin{equation}\label{proof:case one}
 \tilde{b}^a_{i,\bar{t}_1+1}(\theta^*)=\min\{\{b^a_{j,\bar{t}_1}(\theta^*)\}_{j\in \mathcal{N}^{\theta^*}_{i,\bar{t}_1+1}},b^l_{i,\bar{t}_1+1}(\theta^*)\}\geq\delta. 
\end{equation}
\added{\eqref{proof:case one} holds despite possible altered ABs from $f$ bad agents because in the update rule for case one, there is at least one good agent $i\in G$ in $\mathcal{N}^{\theta^*}_{i,\bar{t}_1+1}$ since we only eliminate $f$ smallest beliefs and we have at least $2f+1$ neighbors out of which at most $f$ are bad.} Therefore, the beliefs remaining in $\mathcal{N}^{\theta^*}_{i,\bar{t}_1+1}$ are lower-bounded by $\delta$.

\noindent \added{If case two happens in SDHT, we use \eqref{equation:acutual belief update rule for case two} which gives }
\begin{equation}\label{proof:case two}
     \tilde{b}^a_{i,\bar{t}_1+1}(\theta^*)=\min\{b^a_{i,\bar{t}_1}(\theta^*),b^l_{i,\bar{t}_1+1}(\theta^*)\}\geq\delta.
\end{equation}
Therefore, no matter which case occurs,  we have $\tilde{b}^a_{i,\bar{t}_1+1}(\theta^*)\geq\delta$ before normalization. Then we perform the normalization as in \eqref{equation:norm} and can derive 
\begin{equation}\label{proof:normalization1}
\begin{split}
b^a_{i,\bar{t}_1+1}(\theta^*) &= 
  \frac{\tilde{b}^a_{i,\bar{t}_1+1}(\theta^*)}{\sum_{p=1}^m \tilde{b}^a_{i,\bar{t}_1+1}(\theta_p)}\geq \frac{\delta}{\sum_{p=1}^m \tilde{b}^a_{i,\bar{t}_1+1}(\theta_p)}\\
  &\geq\frac{\delta}{\sum_{p=1}^m b^l_{i,\bar{t}_1+1}(\theta_p)}=\delta.
\end{split}
\end{equation}
The last inequality in \eqref{proof:normalization1} holds since by \eqref{equation:acutual belief update rule for case one} and \eqref{equation:acutual belief update rule for case two}, we know that $\tilde{b}^a_{i,\bar{t}_1+1}(\theta)\leq b^l_{i,\bar{t}_1+1}(\theta)$ for any $\theta\in\Theta$.

Because for all $t\geq \bar{t}_1$, we have $\tilde{b}^a_{i,t}(\theta^*)\geq\delta$, by induction, we can claim that 
\begin{equation}\label{proof:part 1}
    b^a_{i,t}(\theta^*)\geq\delta,\forall t\geq \bar{t}_1,\forall i\in G.
\end{equation}
Now we are ready to prove the second part, which establishes the fact that the beliefs for hypotheses other than the $\theta^*$ are upper-bounded. We pick a small $\epsilon>0$ such that $\epsilon < \delta$. Given a hypothesis $\theta\neq\theta^*$, for any agent $i\in S(\theta,\theta^*)$, by Lemma \ref{lemma:convergence}, we know that there exists a time $t_i^\theta$  such that
\begin{equation}\label{proof:t i theta}
  b_{i,t}^l(\theta)\leq \epsilon^3, \forall t\geq t_i^\theta.  
\end{equation}
We further define
$
\bar{t}_2:=\max\{\bar{t}_1,\max_{i\in S(\theta,\theta^*)}\{t_i^\theta\}\}.
$
Note that, since $\bar{t}_2\geq\bar{t}_1$, from \eqref{proof:part 1} we have that 
$$
b^a_{i,\bar{t}_2+1}(\theta^*)\geq\delta.
$$
\added{For any agent $i\in G$, if case one applies for AB update in SDHT, then we use \eqref{equation:acutual belief update rule for case one} to update $\theta\neq\theta^*$ and obtain}
\begin{equation}\label{proof:case one1}
 \tilde{b}^a_{i,\bar{t}_2+1}(\theta)=\min\{\{b^a_{j,\bar{t}_2}(\theta)\}_{j\in \mathcal{N}^{\theta}_{i,\bar{t}_2+1}},b^l_{i,\bar{t}_2+1}(\theta)\}\leq\epsilon^3.
\end{equation}
\added{If $i\in S(\theta,\theta^*)$, then \eqref{proof:case one1} holds trivially by the definition of $\epsilon$ in \eqref{proof:t i theta}. Otherwise, note that \eqref{proof:case one1} holds even with altered ABs shared from up to $f$ bad agents following  similar reasoning with \eqref{proof:case one}. From the belief update condition in case one, there is at least one good agent $j\in G\cap S(\theta,\theta^*)$ in $\mathcal{N}^{\theta^*}_{i,\bar{t}_1+1}$ since we only eliminate $f$ smallest beliefs and we have at least $2f+1$ neighbors that belong to $S(\theta,\theta^*)$.} \added{On the other hand, if SDHT is in the condition of case two, then for $i\in S(\theta,\theta^*)\cap G$ we have}
\begin{equation}\label{proof:case two1}
\tilde{b}^a_{i,\bar{t}_2+1}(\theta)=\min\{b^a_{i,\bar{t}_2}(\theta),b^l_{i,\bar{t}_2+1}(\theta)\}\leq\epsilon^3.    
\end{equation}
Therefore, no matter which case occurs, we have that
$$
\tilde{b}^a_{i,\bar{t}_2+1}(\theta)\leq\epsilon^3, \forall i\in S(\theta,\theta^*)\cap G
$$
before normalization. Then we perform the normalization as in \eqref{equation:norm} and can derive 
\begin{equation}\label{proof:epsilon}
\begin{split}
  b^a_{i,\bar{t}_2+1}(\theta^*)&=\frac{\tilde{b}^a_{i,\bar{t}_2+1}(\theta)}{\sum_{p=1}^m \tilde{b}b^a_{i,\bar{t}_2+1}(\theta_p)}\leq \frac{\epsilon^3}{\sum_{p=1}^m b\tilde{b}^a_{i,\bar{t}_2+1}(\theta_p)}\\
  &\leq\frac{\epsilon^3}{ b^l_{i,\bar{t}_2+1}(\theta^*)}\leq\frac{\epsilon^3}{\delta}<\epsilon^2  .
\end{split}
\end{equation}
The last inequality is due to the fact $ \epsilon<\delta$. Therefore, by induction we have proved that, ,
\begin{equation}\label{proof:epsilon1}
    b^a_{i,t}(\theta)<  \epsilon^2\leq\epsilon,\forall t\geq \bar{t}_2+1,\forall i\in S(\theta,\theta^*)\cap G.
\end{equation}
For any $i\in G\backslash S(\theta,\theta^*)$, by condition 2 in  Theorem \ref{theorem:main result}, we know that case one will happen infinitely often. As a result, for such agent  $i$, there exists a time $\bar{t}_{i,1}^\theta\geq \bar{t}_2+1$ such that case one occurs for the first time for $t\geq \bar{t}_2+1$. Then at $\bar{t}_{i,1}^\theta$ from \eqref{proof:epsilon1}, we know that
\begin{equation}
  b^a_{j,\bar{t}_{j,1}^\theta}(\theta)\leq \epsilon^2,\forall j\in S(\theta,\theta^*)
  \cap G.  
\end{equation}
Following a  reasoning similar to \eqref{proof:case one1} through \eqref{proof:epsilon}, we obtain that, after normalization, for any agent $i\in G\backslash S(\theta,\theta^*)$, 
\begin{equation}\label{proof:case one first time}
b^a_{i,\bar{t}_{i,1}^\theta}(\theta)<\epsilon.
\end{equation}
Then we define another time instant  $\bar{t}_{i,2}^\theta$ such that $\bar{t}_{i,2}^\theta\geq \bar{t}_{i,1}^\theta+1$ where the case one happens for  second time for $t\geq \bar{t}_2+1$. Notice that, from the conditions in Theorem \ref{theorem:main result}, case two may occur infinitely often  for agent $i\notin S(\theta,\theta^*)$. If this is the case, it then follows that case two happens for any $t\in(\bar{t}_{i,1}^\theta,\bar{t}_{i,2}^\theta)$. By \eqref{equation:acutual belief update rule for case two} and \eqref{proof:case one first time}, we have that
\begin{equation}\label{proof:case one second time}
    b^a_{i,t}(\theta)<\epsilon,\forall t\in(\bar{t}_{i,1}^\theta,\bar{t}_{i,2}^\theta).
\end{equation}
Combining \eqref{proof:case one first time} and \eqref{proof:case one second time},  we obtain that
\begin{equation}\label{proof:case one until}
    b^a_{i,t}(\theta)<\epsilon,\forall t\in[\bar{t}_{i,1}^\theta,\bar{t}_{i,2}^\theta-1].
\end{equation}
Note that \eqref{proof:case one until} holds trivially if $\bar{t}_{i,1}^\theta=\bar{t}_{i,2}^\theta-1$, i.e., there is no occurrence of the case two between two consecutive case one updates. So even if  case two happens only finitely often, \eqref{proof:case one until} still holds.  Then by induction, for agent $i\in G\backslash S(\theta,\theta^*)$, we have that 
\begin{equation}\label{proof:case one until 2}
    b^a_{i,t}(\theta)<\epsilon,\forall t\geq \bar{t}_{i,1}^\theta.
\end{equation}
We further define
$
\bar{t}_3 := \max_\theta\max_{i\notin S(\theta,\theta^*)} \bar{t}_{i,1}^\theta.
$
Since $\bar{t}_3> \bar{t}_2$, 
\begin{equation}\label{proof:part 2}
     b^a_{i,t}(\theta)<\epsilon,\forall t\geq \bar{t}_{3},\forall i\in G,\forall \theta\neq\theta^*.
\end{equation}
\added{Combining \eqref{proof:part 1} and \eqref{proof:part 2}, for any $\omega\in\hat{\Omega}$, \added{$\lim_{t\rightarrow\infty} b^a(i,t)(\theta)= 1$}.} Since the set $\hat{\Omega}$ has measure one as established in Remark \ref{remark:lemma 1}, the proof of Theorem \ref{theorem:main result} is complete.
\end{proof}

\subsection{Proof of Theorem \ref{thm:3}}
\begin{proof}

\added{Like the proof of Theorem \ref{theorem:main result}, 1) we are only interested in state observation path set $\hat{\Omega}$ as defined in Remark \ref{remark:lemma 1} since $\hat{\Omega}$ has measure one. 2) we prove the convergence in two steps for an arbitrary state observation path from $\hat{\Omega}$. The first step establishes that that the AB over the true hypothesis for any good agent $i\in G$ is always lower-bounded from zero. The second step shows that the AB over any hypothesis other than the true hypothesis is upper-bounded by an arbitrarily small constant over time.}

\added{We only consider the scenario that case one happens infinitely often since otherwise the proof trivially holds as discussed in proof of Theorem \ref{theorem:main result}. We fix a path $\omega\in\hat{\Omega}$ and define $$\delta_1:=\min_{i\in G} \lim_{t\rightarrow\infty}b^l_{i,t}(\theta^*).$$ Then, as in the proof of Theorem \ref{theorem:main result},  for each good agent $i\in G$, there exist a time $t_i$ and a constant $\alpha$ such that, for all $t\geq t_i$, we have $b^l_{i,t}(\theta^*)\geq\delta_1-\alpha$ where $\alpha<\delta_1$.}We define $\bar{t}_1$ as in \eqref{equation:bar t 1} and $\delta$ as in \eqref{equation:delta}.

\added{Then at $t= \bar{t}_1+1$, in SDHT, for AB update, either case one or case two happens. If case one happens, for average rule we know that for all $\theta'\neq\theta$, $|S(\theta,\theta')\cap\mathcal{N}_{i,\bar{t}+1}|\geq 2f+2$, then we use \eqref{equation:average belief} and \eqref{equation:acutual belief update rule average} instead of \eqref{equation:asynchronous acutual belief update rule for case one} to update the AB as in the following equation, where}
\begin{equation}\label{proof:case one avg}
\added{\tilde{b}^a_{i,\bar{t}_1+1}(\theta^*) = \min\{\bar{b}^{a}_{i,\bar{t}_1+1}(\theta^*),b^l_{i,\bar{t}_1+1}(\theta^*)\}\geq\delta.}
\end{equation}
\added{The inequality \eqref{proof:case one avg} holds despite possible altered ABs from $f$ bad agents because in the update rule for case one, from Lemma \ref{lemma:lower and upper bounded beliefs} we know that the ABs remaining in $\mathcal{M}^{\theta^*}_{i,\bar{t}_1+1}$ are lower-bounded by $\delta$ and so is the average $\bar{b}^{a}_{i,\bar{t}_1+1}(\theta^*)$ defined in \eqref{equation:average belief}. Combined with the fact that $b^l_{i,\bar{t}+1}(\theta^*)\geq\delta$ by the definition of $\delta$, we know that \eqref{proof:case one avg} holds true.}

\noindent \added{If case two happens in SDHT, we use \eqref{equation:acutual belief update rule for case two} and also have $\tilde{b}^a_{i,\bar{t}_1+1}(\theta^*)\geq\delta$ from \eqref{proof:case two}. Therefore, no matter case one or case two occurs,  we have $\tilde{b}^a_{i,\bar{t}_1+1}(\theta^*)\geq\delta$ before normalization. Then we perform the normalization as in \eqref{equation:norm} and can derive $b^a_{i,\bar{t}_1+1}(\theta^*)\geq\delta$ following the same steps as in \eqref{proof:normalization1}.}

\added{Then following the same induction logic that reaches \eqref{proof:part 1}, we can prove the first step where }
\begin{equation}\label{proof:part 1 avg}
    \added{b^a_{i,t}(\theta^*)\geq\delta,\forall t\geq \bar{t}_1,\forall i\in G.}
\end{equation}
\added{Now we move on to prove the second part, which establishes the fact that the beliefs for hypotheses other than the $\theta^*$ are upper-bounded by an arbitrarily small constant. We pick a small $\epsilon>0$ such that $\epsilon < \delta$. Given a hypothesis $\theta\neq\theta^*$, for any agent $i\in S(\theta,\theta^*)$, by Lemma \ref{lemma:convergence}, we know that there exists a time $t_i^\theta$  such that}
\begin{equation}\label{proof:t i theta avg}
  b_{i,t}^l(\theta)\leq \epsilon^3, \forall t\geq t_i^\theta.  
\end{equation}
We further define
$$
\bar{t}_2:=\max\{\bar{t}_1,\max_{i\in S(\theta,\theta^*)}\{t_i^\theta\}\}.
$$
Note that, since $\bar{t}_2\geq\bar{t}_1$, from \eqref{proof:part 1 avg} we have that 
$$
b^a_{i,\bar{t}_2+1}(\theta^*)\geq\delta.
$$
\added{If case one happens, for average rule we know that for all $\theta'\neq\theta$, $|S(\theta,\theta')\cap\mathcal{N}_{i,\bar{t}+1}|\geq 2f+2$, then we use \eqref{equation:average belief} and \eqref{equation:acutual belief update rule average} instead of \eqref{equation:asynchronous acutual belief update rule for case one} to update the AB and}
\begin{equation}\label{proof:case one1 avg}
\added{\tilde{b}^a_{i,\bar{t}_2+1}(\theta) = \min\{\tilde{b}^{a}_{i,\bar{t}_2+1}(\theta),b^l_{i,\bar{t}_2+1}(\theta)\}\leq\epsilon^3.}
\end{equation}
\added{If $i\in S(\theta,\theta^*)$, \eqref{proof:case one1 avg} holds trivially by the definition of $\epsilon$ in \eqref{proof:t i theta avg}. Otherwise, note that \eqref{proof:case one1 avg} holds even with altered ABs shared from up to $f$ bad agents. From the belief update condition in case one, we know that there exists at least one good agent $j''\in\mathcal{N}_{i,\bar{t}_2+1}\cap S(\theta,\theta')$ such that $b^a_{j,\bar{t}_2}(\theta)\leq b^a_{j'',\bar{t}_2}(\theta)$ for any agent $j$ in $\mathcal{M}^{\theta^*}_{i,\bar{t}_2+1}$. Furthermore, it is guaranteed that  $b^a_{j'',t_2}(\theta)$ will be upper-bounded by an arbitrarily small constant from Remark \ref{remark:for lemma 1} for Lemma \ref{lemma:convergence} and so is the average $\bar{b}^{a}_{i,\bar{t}_2+1}(\theta^*)$ in \eqref{equation:average belief}. On the other hand, if SDHT is in the condition of case two, then we have $\tilde{b}^a_{i,\bar{t}_2+1}(\theta)\leq\epsilon^3$ as in \eqref{proof:case two1}.}
\added{
Therefore, no matter which case occurs, we have that
$$
\tilde{b}^a_{i,\bar{t}_2+1}(\theta)\leq\epsilon^3, \forall i\in S(\theta,\theta^*)\cap G
$$
before normalization. Then we perform the normalization as in \eqref{equation:norm} and can derive 
$b^a_{i,\bar{t}_2+1}(\theta^*)<\epsilon^2$ following the same steps that reach \eqref{proof:epsilon}.}

\added{The rest of proof follows the proof of  Theorem \ref{theorem:main result} from \eqref{proof:epsilon1} on.}
\end{proof}

\subsection{Proof of Theorem \ref{thm:4}}
\begin{proof}

\added{Like the proof of Theorem \ref{theorem:main result1}, 1) we are only interested in state observation path set $\hat{\Omega}$ as defined in Remark \ref{remark:lemma 1} since $\hat{\Omega}$ has measure one. 2) we prove the convergences in two steps for an arbitrary state observation path from $\hat{\Omega}$.}

\added{For the first part that lower-bounds $b_i^a(\theta^*)$, as in the proof of Theorem \ref{theorem:main result1}, we study two different scenarios. In the first scenario where case one only happens finitely often to an agent $i\in G$ and $\theta^*$, the proof follows that of Theorem \ref{theorem:main result1}.}

\added{The second scenario indicates that case one happens infinitely often to an agent $i\in G$ and $\theta^*$. Then, as in the proof of Theorem \ref{theorem:main result1},  for each good agent $i\in G$, there exist a time $t_i$ and a constant $\alpha$ such that, for all $t\geq t_i$, we have $b^l_{i,t}(\theta^*)\geq\delta_1-\alpha$ where $\alpha<\delta_1$.} We define $\bar{t}_1$ as in \eqref{proof:bar t 1} and $\delta$ as in \eqref{equation:delta}.

\added{Since case one happens infinitely often, there must exist a time $t'_i\geq\bar{t}_1$ that Algorithm \ref{alg:ABU} returns true. For average rule, it means that for all $\theta'\neq\theta$, $|\mathcal{N}^{\theta}_{i}\cap S(\theta,\theta')|\geq 2f+2$. As a result, $ResetFlag$ is set to true and after AB update with \eqref{equation:acutual belief update rule average} at $t'_i$, all the saved ABs are deleted at $t'_i+1$. Therefore,  if $j\in\mathcal{N}_{i,t}^{\theta^*}$, we know that}
\begin{equation}\label{proof4:t'_i}
    \added{{b}_j^a(\theta^*)=b_{j,t}^a(\theta^*)\geq\delta,\forall t\geq t'_i.}
\end{equation}
\added{There must also exist a time $t''_i>t'_i$ such that case one happens again in ADHT for AB update with average rule, where we use \eqref{equation:average belief} and \eqref{equation:acutual belief update rule average} instead of \eqref{equation:asynchronous acutual belief update rule for case one} and obtain}
\begin{equation}\label{proof2:case one avg}
\added{\tilde{b}^a_{i,t''_i}(\theta^*)=\min\{\{b^a_{j}(\theta^*)\}_{j\in \tilde{\mathcal{N}}^{\theta^*}_{i,t''_i}},b^l_{i,t''_i}(\theta^*)\}\geq\delta.}
\end{equation}
\added{If $i\in S(\theta^*,\theta)$,  \eqref{proof2:case one avg} holds trivially. Otherwise, the inequality \eqref{proof2:case one avg} holds despite possible altered ABs from $f$ bad agents because in the update rule for case one, from Lemma \ref{lemma:lower and upper bounded beliefs ADHT} we know that the ABs remaining in $\mathcal{M}^{\theta^*}_{i,t''_i}$ are lower-bounded by $\delta$ and so is the average AB $\bar{b}^{a}_{i,t''}(\theta^*)$ defined in \eqref{equation:average belief}. Combined with the fact that $b^l_{i,t''_i}(\theta^*)\geq\delta$ by the definition of $\delta$, we know that \eqref{proof2:case one avg} holds true.}

\added{At $t''_i+1$, if case one happens again,  we know that $\tilde{b}^a_{i,t''_i+1}(\theta^*)\geq\delta$ by the same logic that reaches \eqref{proof2:case one avg}. Alternatively, if case two happens at $t''_i+1$, we use update rule \eqref{equation:acutual belief update rule for case two} and we have} 
\begin{equation}\label{proof2:case two avg}
     \added{\tilde{b}^a_{i,t''_i+1}(\theta^*)=\min\{b^a_{i,t''_i}(\theta^*),b^l_{i,t''_i+1}(\theta^*)\}\geq\delta}
\end{equation} 
by the definition of $\delta$. Therefore, no matter which case occurs,  we have $\tilde{b}^a_{i,t''_i+1}(\theta^*)\geq\delta$ before normalization. Then we perform the normalization as in \eqref{equation:norm} and can derive $b^a_{i,t''_i+1}(\theta^*)\geq\delta$ as in \eqref{proof2:normalization1}. Consequently, following the same logic and by induction, we reach \eqref{proof2:part 11}, where we rewritten below for readability. 

\begin{equation}\label{proof2:part 11 repeat}
        \added{b_{i,t}^a(\theta^*)\geq\delta,\forall t\geq \tilde{t}_2,\forall i\in G.}
\end{equation}

\added{Now we move on to prove that the ABs over $\theta\neq\theta^*$ are upper bounded by an arbitrarily small constant.}  \added{Given a hypothesis $\theta\neq\theta^*$, for any agent $i\in S(\theta,\theta^*)$, we pick a small $0<\epsilon<1$ such that $\epsilon<\delta$ and define $t_i^\theta$ 
such that}
\begin{equation}\label{proof2: epsilon 3 avg}
b_{i,t}^l\leq\epsilon^3,\forall t\geq t_i^\theta.
\end{equation}
Then we further define
$$
\added{\tilde{t}_3 := \max\{\tilde{t}_2,\max_{i\in S(\theta,\theta^*)}\{t_i^\theta\}\}.}
$$

For any agent $i\in G\cap S(\theta,\theta^*)$, if case one applies for AB update in ADHT, then we use \eqref{equation:average belief} and \eqref{equation:acutual belief update rule average} instead of \eqref{equation:asynchronous acutual belief update rule for case one} to update $\theta\neq\theta^*$ and obtain
\begin{equation}\label{proof2:case one1 avg}
 \tilde{b}^a_{i,\tilde{t}_3+1}(\theta)=\min\{\bar{b}^{a}_{i,\tilde{t}_3+1}(\theta^*),b^l_{i,\tilde{t}_3+1}(\theta)\}\leq\epsilon^3.
\end{equation}
\added{The inequality \eqref{proof2:case one1 avg} holds even with altered ABs shared from up to $f$ bad agents. From the belief update condition in case one, we know that there exists at least one good agent $j''\in\mathcal{N}_{i,\tilde{t}_3+1}\cap S(\theta,\theta^*)$ such that $b^a_{j,\tilde{t}_3}(\theta)\leq b^a_{j'',\tilde{t}_3}(\theta)$ for any agent $j$ in $\mathcal{M}^{\theta^*}_{i,\tilde{t}_3+1}$. Furthermore, it is guaranteed that  $b^a_{j'',\tilde{t}_3+1}(\theta)$ will be upper-bounded by $\epsilon^3$ and so is the average $\bar{b}^{a}_{i,\tilde{t}_3+1}(\theta)$ in \eqref{equation:average belief}.} On the other hand, if ADHT is in the condition of case two, then we have
\begin{equation}\label{proof2:case two1 avg}
\tilde{b}^a_{i,\tilde{t}_3+1}(\theta)=\min\{b^a_{i,\tilde{t}_3}(\theta),b^l_{i,\tilde{t}_3+1}(\theta)\}\leq\epsilon^3.    
\end{equation}
Therefore, no matter which case occurs, we have that
$$
\tilde{b}^a_{i,\tilde{t}_3+1}(\theta)\leq\epsilon^3, \forall i\in S(\theta,\theta^*)\cap G
$$
before normalization. Then we perform the normalization as in \eqref{equation:norm} and can derive $b^a_{i,\tilde{t}_3+1}(\theta^*<\epsilon^2$ following the same reasoning that reaches \eqref{proof2:epsilon}.

\added{The rest of proof follows the proof of  Theorem \ref{theorem:main result1} from \eqref{proof2:epsilon 2} on.}
\end{proof}

\end{document}